\documentclass[12pt]{article}

\usepackage{amsmath}
\usepackage{graphicx}
\usepackage{enumerate}
\usepackage{natbib}

\newcommand{\blind}{0}

\usepackage{amsmath, amsthm, amssymb}
\usepackage{mathtools}
\usepackage{setspace}
\usepackage{comment}

\usepackage{enumerate}
\usepackage{enumitem}
\usepackage{graphicx}
\usepackage{natbib}
\usepackage{xcolor}
\usepackage{soul}
\usepackage{subfiles}
\usepackage{multirow}
\usepackage{lscape}
\usepackage{bbm}
\usepackage[pagewise]{lineno}
\usepackage{algorithm}
\usepackage{algorithmic}

\usepackage{array,etoolbox}
\preto\tabular{\setcounter{magicrownumbers}{0}}
\newcounter{magicrownumbers}

\definecolor{maroon}{RGB}{0,0,0}
\definecolor{navy}{RGB}{0,0,0}
\definecolor{indigo}{RGB}{0,0,0}

\DeclareMathOperator*{\Tr}{Tr}

\DeclareMathOperator*{\argmin}{arg\,min}

\graphicspath{ {./Fig/} }
\newcommand{\bone}{{\mathbf 1}}
\newcommand{\bzero}{{\mathbf 0}}
\newcommand{\mone}{{\mathbbm{1}}}

\newcommand{\ba}{{\mathbf a}}

\newcommand{\bd}{{\mathbf d}}
\newcommand{\be}{{\mathbf e}}

\newcommand{\bg}{{\mathbf g}}

\newcommand{\bl}{{\mathbf l}}

\newcommand{\bu}{{\mathbf u}}

\newcommand{\bA}{{\mathbf A}}

\newcommand{\bE}{{\mathbf E}}

\newcommand{\bG}{{\mathbf G}}
\newcommand{\bH}{{\mathbf H}}
\newcommand{\bI}{{\mathbf I}}

\newcommand{\bQ}{{\mathbf Q}}
\newcommand{\bR}{{\mathbf R}}
\newcommand{\bS}{{\mathbf S}}
\newcommand{\bT}{{\mathbf T}}
\newcommand{\bU}{{\mathbf U}}

\newcommand{\bW}{{\mathbf W}}
\newcommand{\bX}{{\mathbf X}}

\newcommand{\bZ}{{\mathbf Z}}

\newcommand{\Gau}{\textit{N}}

\newcommand{\btheta}{\boldsymbol{\theta}}

\newcommand{\bdelta}{\boldsymbol{\delta}}

\newcommand{\ve}{\varepsilon}

\newcommand\T{\top}

\newcommand{\bDelta}{\boldsymbol{\Delta}}

\newcommand{\bOmega}{\boldsymbol{\Omega}}
\newcommand{\bTheta}{\boldsymbol{\Theta}}

\newcommand{\E}{\mathbb{E}}

\newcommand{\Op}{O_{\P}}
\newcommand{\op}{o_{\P}}

\newcommand{\ol}[1]{\overline{#1}}
\newcommand{\wh}[1]{\widehat{#1}}
\newcommand{\wt}[1]{\widetilde{#1}}

\newcommand{\blkeq}{\phantom{{} = {}}}

\newcommand{\lrm}[1]{\{{#1}\}}
\newcommand{\lrM}[1]{\left\{{#1}\right\}}
\newcommand{\lri}[1]{({#1})}
\newcommand{\lrI}[1]{\left({#1}\right)}
\newcommand{\lro}[1]{[{#1}]}
\newcommand{\lrO}[1]{\left[{#1}\right]}
\newcommand{\lrN}[1]{\left\|{#1}\right\|}
\newcommand{\lrn}[1]{\|{#1}\|}

\newcommand{\lrt}[1]{{\left\vert\kern-0.25ex\left\vert\kern-0.25ex\left\vert #1 
    \right\vert\kern-0.25ex\right\vert\kern-0.25ex\right\vert}}
\newcommand{\lrA}[1]{\left|{#1}\right|}

\newcommand{\cid}{\rightsquigarrow}

\newcommand*\diff{\mathop{}\!\mathrm{d}}

\renewcommand{\P}{\mathbb{P}}

\newtheorem{theorem}{Theorem}[section]
\newtheorem{lemma}[theorem]{Lemma}
\newtheorem{corollary}{Corollary}[section]

\newtheorem{assump}{Assumption}

\newtheorem{prop}{Proposition}[section]

\let\oldproof\proof
\let\oldendproof\endproof
\renewenvironment{proof}
  {\oldproof\normalsize\color{navy}}
  {\oldendproof}
  
\begin{document}

\def\spacingset#1{\renewcommand{\baselinestretch}%
{#1}\small\normalsize} \spacingset{1}

%%%%%%%%%%%%%%%%%%%%%%%%%%%%%%%%%%%%%%%%%%%%%%%%%%%%%%%%%%%%%%%%%%%%%%%%%%%%%%

\if0\blind
{
  \title{\large \bf Communication-efficient Distributed Newton-like Optimization with Gradients and $M$-estimators}
  \author{Ziyan Yin \\%\thanks{
    %The authors gratefully acknowledge \textit{please remember to list all relevant funding sources in the unblinded version}}\hspace{.2cm}\\
   % Department of Statistical Science, Temple University\\
% \\
   % Department of Risk Management and Insurance, Temple University \\
    %and \\
    %Cheng Yong Tang \\
   % Department of Statistical Science, Temple University
	}
	\date{}
  \maketitle
} \fi

\if1\blind
{
  \bigskip
  \bigskip
  \bigskip
  \begin{center}
    { \Large Communication-efficient Distributed Newton-like Optimization with Gradients and $M$-estimators }
\end{center}
  \medskip

%\centerline{ Ziyan Yin and Cheng Yong Tang}

} \fi

\bigskip

\begin{abstract}
  In modern data science, it is common that large-scale data are stored and processed parallelly across a great number of locations. For reasons including confidentiality concerns, only limited data information from each parallel center is eligible to be transferred. To solve these problems more efficiently, a group of communication-efficient methods are being actively developed. We propose two communication-efficient Newton-type algorithms, combining the $M$-estimator and the gradient collected from each data center. They are created by constructing two Fisher information estimators globally with those communication-efficient statistics. Enjoying a higher rate of convergence, this framework improves upon existing Newton-like methods. Moreover, we present two bias-adjusted one-step distributed estimators. When the square of the center-wise sample size is of a greater magnitude than the total number of centers, they are as efficient as the global $M$-estimator asymptotically. The advantages of our methods are illustrated by extensive theoretical and empirical evidences. 
\end{abstract}

\section{Introduction}
The statistical inference and optimization problem under a distributed setting is a popular topic in modern data science applications. For example, online marketing data are often too big to be stored within one hard drive; medical records are usually stored and processed separately at regional healthcare centers, and the inter-hospital sharing of the patient-level information is highly regulated and restricted. In these scenarios, despite the fact that more accurate inference can be obtained by combining data from all centers, practical realities prevent these centers from sharing the raw data stored therein. Among others, two foremost concerns are the data confidentiality and the technical feasibility. 
Consequently, only limited statistics satisfying some restrictive rules are eligible to be infrequently transferred among centers. 

Recent years have witnessed a flurry of important technological and methodological developments of the communication-efficient methods to handle such distributed data problems. One broadly applied class is the so-called divide-and-conquer strategy which takes average of certain statistics calculated at each center separately. Recent advances include \cite{Zhang2013, ZhangImprove, AvgSparLM, AvgKernel, AvgTesting, AvgQR, Chen2020Quan, AvgSVM, AvgPCA}, etc.
%regression analysis \citep{Zhang2013, AvgSparLM}, kernel ridge regression \citep{AvgKernel},  hypothesis  testing \citep{AvgTesting}, quantile regression \citep{AvgQR, Chen2020Quan}, linear support vector machine \citep{AvgSVM}, principal component analysis \citep{AvgPCA} and feature screening  \citep{Avg_FS}.

However, one of the major limitations of these one-shot averaging methods is that, their efficacy is heavily limited by the center-wise sample size. %they can only reduce the variance but not the bias. 
In other words, the dataset cannot be split across too many centers, and adding extra data centers to the existing dataset does not always improve accuracy. \cite{Zhang2013} showed when the center-wise sample size is fixed, the mean squared error of the simple average of $M$-estimators can hardly be reduced by adding more data centers. 

Recently, another communication-efficient framework has received considerable attention.  \cite{Wang17} and \cite{MJ} proposed to use some surrogate loss function that can be evaluated  efficiently at a preselected data center (also known as the local center). The advantage of this framework is that only gradients are collected from parallel centers, upon receiving an initial estimator.  More  specifically, this strategy replaces the  higher-order derivatives of the global loss function  (that requires data from all data centers) with those local surrogates.  Due to its convenience and promising properties, this strategy has been broadly investigated and applied in different directions; for example, see \cite{Fan2019, Duan2021, Applied_DistBT, Applied_DistBT2, Applied_LDA, AvgSVM, Applied_SGD}, among others. 

Although the estimator proposed in \cite{Wang17} and \cite{MJ} appears efficient in some cases, when the number of centers $m$ and the center-wise sample size $n$ are of the same magnitude, the accuracy of the local surrogates becomes the bottleneck, limiting further improvement. The local high-order derivatives remain barely changed no matter how many parallel centers there are. More specifically, in each iteration, this framework uses the local Hessian estimator as the ``learning rate'', which always has a fixed $\Op(n^{-1/2})$ gap to the true Fisher information matrix. 

In light of this finding, in this paper, we are motivated to design a more efficient distributed estimation procedure by constructing some more accurate Fisher information estimators. We find in each parallel center, the Fisher information can be seen as the ``coefficient'' of a ``linear'' map from the gradient evaluated at the truth to the estimation error of the corresponding $M$-estimator. Although this ``true'' gradient is unavailable, when a given estimator is close enough to the truth, the ``linearity'' between the errors of the $M$-estimators and the gradients evaluated at that given estimator still holds with negligible discrepancy. This kind of pseudo-linearity indicates that, its pseudo-coefficient --- the true Fisher information matrix --- can be estimated via linear regression. Thus, we propose two Fisher information estimators by combining the $M$-estimators and gradients collected from parallel centers. We show that under regular existence and continuity assumptions, the distances of both our estimators to their estimands are $\Op(n^{-1} + m^{-1})$, which outperform the local Hessian estimator used in  \cite{Wang17} and \cite{MJ}. 

With these $M$-estimator/gradient-enhanced Fisher information estimators, we propose two iterative distributed algorithms correspondingly. Both have higher rates of convergence than the traditional method does. %Also, in practice, it is often more desirable to have as few rounds of communications as possible, so 
Additionally, we present two bias-adjusted one-step distributed estimators using the aforementioned Fisher information estimators. Their distances to the global $M$-estimator are reduced by removing quadratic terms of the estimation error of the initial estimator. As a result, these two refined one-shot estimators have the same efficiency asymptotically as the global $M$-estimator does when $m = o(n^2)$.

We introduce some notations. Let $F_X(u)$ and $F_X^{-1}(q)$ be the cumulative distribution function and the quantile function of a random variable $X$. For a vector $\ba = (a_1, \dots, a_m)^{\T}$, $\text{supp}(\ba) = \lrm{i:a_i \ne 0}$, $\|\ba\|_q = (\sum_i |a_i|^q)^{1/q}$, $\|\ba\|_{\infty} = \max_i |a_i|$, and $\|\ba\|_0 = \#\lrm{i:a_i \ne 0}$. Also, $\bX \cid F$ denotes a random element $\bX$ converges in distribution to $F$. For a matrix $\bA$, $\{\bA\}_{ij}$ denote the entry in the $i$th row and $j$th column of $\bA$. The Kronecker product is denoted by $\otimes$.

\section{Motivation and Problem Set-Up} \label{Hsec:PS} 
We first state the structure and storage of the dataset.
Let $\lrm{\bX_i}_{i=1}^{N}$ denote $N$  independent and identically distributed samples with marginal distribution $\mathcal{P}$ over some sample space $\mathcal{X}$. For any parameter $\btheta$ containing in some convex space $\Theta \subseteq \mathbb{R}^d$, define a convex and three-times differentiable loss function $L: \Theta \times \mathbb{R}^d \to \mathbb{R}$, such that the true parameter $\btheta_0$ is a minimizer of the population risk, that is
$$\btheta_0 \in \argmin_{\btheta \in \Theta} \E_{\mathcal{P}}\lrm{L(\btheta, \bX)}.$$

We consider the evenly distributed setting where $\lrm{\bX_i}_i^{N}$ are uniformly stored in one local (or central) center $\mathcal{L}$ (or $\mathcal{M}_1$) and $m - 1$ global (or parallel) centers $\lrm{\mathcal{M}_i}_{i=2}^{m}$ with $\mathcal{G} = \bigcup_{i=1}^m \mathcal{M}_i$; therefore, $N = mn$. Suppose we have full access to the data stored in $\mathcal{L}$ which also plays the role of processing and delivering the final results. %Following this setting, let $\bX_{j}^{(l)}$ be the $j$th copy and $\bX_{ij}$ denote the $j$th copy in the $i$th center $\mathcal{M}_i$, for $i = 1, \dots, m$ and $j = 1, \dots, n$. 
Despite the requirement that the data are evenly stored with equal size $n$, our findings can be potentially generalized to unequal sample size cases without compromising the spirit.

%Given an initial estimator $\wt\btheta_0$ close enough to $\btheta_0$, we first propose two Fisher information estimators.

Let $\wh\btheta_i$ be the $M$-estimator from the $i$th center $\mathcal{M}_i$, and the average $\ol\btheta = m^{-1}\sum_{i=2}^m \wh\btheta_i$. We refer the readers to \cite{Zhang2013} for a thorough picture of properties of $M$-estimators under this distributed setting.
%error is known to be  $\lrn{\ol\btheta - \btheta_0}_2 = \Op(n^{-1} + n^{-1/2}m^{-1/2})$ in \cite{Zhang2013}. 

%When applying quadratic approximation under the distributed setting, we can use these two estimators to speed up the whole process. 

%With these two Hessian replacements, we propose our iterative distributed estimators:
%\begin{align}
%    & \wt\btheta_{t+1} = \wt\btheta_t - \wh\bI_0^{-1} \ol\bl(\wt\btheta_t); \label{Heq:DefTheta2}\\
%    & \wt\btheta_{t+1}' = \wt\btheta_t - \wh\bOmega \ol\bl(\wt\btheta_t). \label{Heq:DefTheta2p}
%\end{align}
%\cite{MJ} and \cite{Wang17} proposed the Newton-type distributed estimator $\wt\btheta^c$: 
%\begin{equation}
%    \wt\btheta^c = \wt\btheta_0 - (\wh\bH^{(l)})^{-1}(nm)^{-1} \sum_i \sum_j \nabla L(\wt\btheta_0;\bX_{ij}), \label{HHeq:MEDef}
%\end{equation}
Newton's method is a powerful tool for optimization problems. However, when the dataset is split across many centers, and the inter-center communication cost is a major concern, it cannot be applied directly because the Hessian matrices can be huge and hard to transfer. One approach is the so-called communication-efficient surrogate likelihood (CSL) framework  proposed in In \cite{MJ} and \cite{Wang17}. They replaced the unavailable global Hessian with the local Hessian matrix $\wh\bH^{(l)}(\btheta) = n^{-1} \sum_{j \in \mathcal{L}} \nabla^2 L(\btheta; \bX_j)$. Then, Newton-like methods are applied upon receiving gradients from parallel centers. This method is efficient, and the communication cost is low. However, it depends heavily on the quality and quantity of the local data --- $\wh\bH^{(l)}(\btheta)$ remains almost unchanged no matter how large the global sample size is. That is, $$
    \lrn{\wh\bH^{(l)}(\btheta) - \bI_0}_2 %& \leq n^{-1}\sum_j\lrN{\nabla^2 L(\wt\btheta_0;\bX_i^{(l)}) - \nabla^2 L(\btheta_0;\bX_i^{(l)})}_2 + \lrN{n^{-1}\sum_j \nabla^2 L(\btheta_0;\bX_i^{(l)}) - \bI_0}_2 \nonumber \\
         = \Op(\lrn{\btheta - \btheta_0}_2) + \Op(n^{-1/2}). %, \label{Heq:HlBias}
$$
The estimation error of $\wh\bH^{(l)}(\btheta)$ depends on $\lrn{\btheta - \btheta_0}_2$ and $n^{-1/2}$, neither of which decreases with a larger $m$. Also, the whole optimization procedure may be ruined, if we accidentally choose a ``bad'' data center to be the local one.  
%That is, $\wt\btheta^c_{t+1} = \wt\btheta_0 - (\wh\bH^{(l)})^{-1}\ol\bl(\wt\btheta_0)$. Under some mild existence and continuity assumptions, the estimation errors from $\wt\btheta_0$, $\ol\bl(\wt\btheta_0)$, and the distance from the local Hessian to its estimand $\bI_0$, all contribute to the inaccuracy of $\wt\btheta^c$, leading to the error $\lrn{\wt\btheta^c - \btheta_0}_2 = \Op(\lrn{\bDelta_1}_2^2) + \Op(n^{-1/2}\lrn{\bDelta_1}_2) + \Op(n^{-1/2}m^{-1/2})$.

%To improve $\wt\btheta^c$, one way  is finding a better initial estimator; another way is to have a greater $N$ to improve $\ol\bl(\wt\btheta_0)$. 

% Clearly, the accuracy of $\wh\bH^{(l)}$ does not increase by adding more data, and its error contains that from $\wt\btheta_0$ as well. This leads to the fact that, 
%when $n \gg m$, $\wt\btheta^c$ is efficient and it only requires few iterations to reach the global minimizer $\btheta^*$. However, starting from $m = O(n)$, the estimation error from the local observed Hessian becomes a major limitation in each iteration.  Consequently, an improvement can be made by constructing an alternative that is closer to $\bI_0$, which leads to the proposed estimators $\wt\bI_0$ and $\wt\bOmega$.

Thus, we are trying to find communication-efficient alternatives of $\wh\bH^{(l)}(\btheta)$, and we hope their errors would decrease when $m$ increases. This brings the estimators of the Fisher information matrix $\bI_0 = \E_{\mathcal{P}}\lrm{\nabla^2 L(\btheta_0;\bX)}$ and its inverse $\bI_0^{-1}$: %with the following two matrices:
\begin{align}
    & \wh\bI_0(\btheta) = -\lrM{\sum_{i=1}^m (\wh\btheta_i - \ol\btheta)(\wh\btheta_i - \ol\btheta)^\T}^{-1} \lrO{\sum_{i=1}^m (\wh\btheta_i - \ol\btheta)\lrM{\bl_i(\btheta) - \ol\bl(\btheta)}^\T} \label{Heq:I0HatDef} \\
    & \wh\bOmega(\btheta) = -\lrO{\sum_{i=1}^m\lrM{\bl_i(\btheta) - \ol\bl(\btheta)}\lrM{\bl_i(\btheta) - \ol\bl(\btheta)}^\T}^{-1} \lrO{\sum_{i=1}^m\lrM{\bl_i(\btheta) - \ol\bl(\btheta)}(\wh\btheta_i - \ol\btheta)^\T}, \label{Heq:OmegaHatDef}
\end{align}
where $\bl_i(\btheta) = n^{-1}\sum_{j \in \mathcal{M}_i} \nabla L(\btheta;\bX_{j})$ is the gradient evaluated at $\btheta$ in the $i$th center, and $\ol\bl(\btheta) = m^{-1} \sum_{i=1}^m \bl_i(\btheta)$ is their average. Their construction needs no high-order derivative; we only collect gradients and $M$-estimators. We name $\wh\bI_0(\btheta)$ and $\wh\bOmega(\btheta)^{-1}$ the $M$-estimator/gradient (MG) and gradient/$M$-estimator (GM) Fisher information estimators, and will show they are closer to $\bI_0$ than $\wh\bH^{(l)}(\btheta)$. % in probability. 

%In the following section, we will demonstrate that $\wh\bI_0$ and $\wh\bOmega$ are closer to their estimands in probability, which results in $\wt\btheta_1$ and $\wt\btheta_1'$ performing better than the existing method $\wt\btheta^c$. 

\subsection{Estimation Accuracy of the $M$-Estimator/Gradient Fisher Information Estimators}
Let $\bDelta = \btheta - \btheta_0$, and define two cross-products involving the first and second derivatives of $L(\btheta;\bX)$: 
\begin{align*}
& \bQ_{11} = \E_\mathcal{P}\lrO{\nabla L(\btheta_0;\bX_1) \nabla L(\btheta_0;\bX_1)^\T}, \\
& \bQ_{12} = \E_\mathcal{P}\lrO{\nabla L(\btheta_0;\bX_1) \otimes \lrm{\nabla^2 L(\btheta_0;\bX_1) - \bI_0}}.
\end{align*}

\begin{prop} \label{Hpr:DGGD}  Let $U(\rho)$ = $\{\btheta \in \bTheta; \lrn{\btheta - \btheta_0}_2 \leq \rho\}$ be some Euclidean ball of radius $\rho > 0$ in which Assumptions \ref{as:T0Def}, \ref{as:T0Ind}, \ref{Has:LH}, and \ref{Has:L3Cont} hold. Suppose $n$ and $m$ are large enough, such that $\wh\bI_0(\btheta)$ and $\wh\bOmega(\btheta)$ exist almost surely. 
\begin{enumerate}[label={(\arabic*)}]
\item For any $\btheta \in U(\rho)$, 
    $$
        \wh\bI_0(\btheta) - \bI_0 = \bI_0\bQ_{11}^{-1}(\bI_{d } \otimes \bDelta^\T) \bQ_{12} + \bR_{D},
    $$
    where $\lrn{\bR_{D}}_2 = \Op\lrm{(n^{3/2}\lrn{\bDelta}_2^{3} + 1)(n^{-1} + m^{-1})}$;

\item When $\btheta \in U(\rho)$ and $\lrn{\bDelta}_2 = \op(n^{-1/4})$, 
    $$
        \wh\bOmega(\btheta) - \bI_0^{-1} = \bQ_{11}^{-1}(\bI_{d} \otimes \bDelta^\T) \bQ_{12}\bI_0^{-1} + \bR_{G},
    $$
    where $\lrn{\bR_{G}}_2 = \Op\lri{n\lrn{\bDelta}_2^4 + m^{-1} + n^{-1}}$.
\end{enumerate}
\end{prop}
It is clear that the estimation errors of $\wh\bI_0(\btheta)$ and $\wh\bOmega(\btheta)$ consist of both $\bDelta$-constant parts and sub-$\bDelta$ parts. When $\lrn{\bDelta}_2 = \Op(n^{-1/2})$, both $\lrn{\wh\bI_0(\btheta) - \bI_0}_2$ and $\lrn{\wh\bOmega(\btheta) - \bI_0^{-1}}_2$ are $O(\lrn{\bDelta}_2) + \Op(n^{-1} + m^{-1})$. By contrast, using the local Hessian leads to the error $\lrn{\wh\bH^{(l)}(\btheta) - \bI_0}_2 = \Op(\lrn{\bDelta}_2) + \Op(n^{-1/2})$. For a given $\bDelta$, our proposed Fisher information estimators converge to their estimands stochastically faster than the local Hessian does, if  $n = o(m^2)$. Also, for the $\bDelta$-constant part, a local estimation can be conducted to further reduce their errors and speed up the whole process; see Theorem \ref{Hthm:EEst}. 

Remarkably, in generalized linear models (e.g., logistic regression), $\bQ_{11} = \bI_0$ and $\bQ_{12} = \bzero$. Our proposed estimators would outperform $\wh\bH^{(l)}(\btheta)$ without any adjustment.

%Therefore, in this case, when $\btheta$ is $n^{1/2}$-consistent, the estimation errors of $\wh\bI_0(\btheta)$ and $\wh\bOmega(\btheta)$ are both $\Op(n^{-1} + m^{-1})$, which outperform $\wh\bH^{(l)}(\btheta)$ as long as $n = o(m^2)$. 

%These facts indicate our proposed estimators converge stochastically faster to $\bI_0$ than $\wh\bH^{(l)}(\btheta)$ does as long as $n = o(m^2)$. 

\color{navy}
\subsection{Decomposition of the Estimation Errors}
Proposition \ref{Hpr:DGGD} comes from two key properties: the mean-value theorem describing the gap between two Hessian matrices with different parameters, and the asymptotic ``linearity'' between the $M$-estimator and the gradient at $\btheta_0$. 

%For convenience, let $\wh\bd_i = \wh\btheta_i - \ol\btheta$ and $\wh\bg_i = \bl_i(\wt\btheta_0) - \ol\bl(\wt\btheta_0)$. Then, $\wh\bI_0 = -\lrI{\sum\wh\bd_i\wh\bd_i^\T}^{-1}\lrI{\sum\wh\bd_i\wh\bg_i^\T}$ and $\wh \bOmega = -\lrI{\sum\wh\bg_i\wh\bg_i^\T}^{-1}\lrI{\sum\wh\bg_i\wh\bd_i^\T}$.

%For simplicity, the following discussion in this section is based on the $n$ copies $\lrm{\bX_j}_{j=1}^n$ in the $i$th center $\mathcal{M}_i$.

The first statement is an application of the multivariate mean-value theorem of the distance between $n^{-1}\sum_{j \in \mathcal{M}_i} \nabla^2 L(\btheta;\bX_j)$ and $n^{-1}\sum_{j \in \mathcal{M}_i} \nabla^2 L(\btheta_0;\bX_j)$: if the third derivative of $L(\btheta;\bX)$ exists and is $m(\bX)$-Lipschitz continuous with respect to $\btheta$, then, for any $\btheta \in U(\rho)$, 
    \begin{align*}
        & n^{-1}\sum_{j \in \mathcal{M}_i} \nabla^2 L(\btheta;\bX_j) - n^{-1}\sum_{j \in \mathcal{M}_i} \nabla^2 L(\btheta_0;\bX_j) \approx \lrm{\bI_d \otimes\bDelta^\T} \bQ, \\%\nonumber \\
        & n^{-1}\sum_{j \in \mathcal{M}_i} \int_0^1 \nabla^2 L\lrm{\btheta_0 + t\bDelta;\bX_j} \diff t - n^{-1}\sum_{j \in \mathcal{M}_i} \nabla^2 L(\btheta_0;\bX_j) 
        \approx \frac{1}{2}\lrm{\bI_d \otimes\bDelta^\T} \bQ, %\label{Heq:HhatHHalf}
    \end{align*}
    where 
    \begin{equation*}
        \bQ = \begin{pmatrix} \nabla^2 \lrO{\frac{\partial}{\partial \theta_{1}}\E_{\mathcal{P}}\lrm{L(\btheta_0;\bX)}} \\ \vdots \\ \nabla^2 \lrO{\frac{\partial}{\partial \theta_{d}}\E_{\mathcal{P}}\lrm{L(\btheta_0;\bX)}}\ \end{pmatrix}. \label{Heq:QTDef}
    \end{equation*}
    
The second property is the ``linear'' association between the $i$th sample gradient at $\btheta_0$ and the corresponding $M$-estimator $\wh\btheta_i$ at $\mathcal{M}_i$. At the $i$th center, let $\bd_i = \wh\btheta_i - \btheta_0$ and $\ol\bH_i = n^{-1}\sum_j \int_0^1 \nabla^2 L\lri{\btheta_0 + t\bd_i;\bX_{ij}}\diff t$. Then,
\begin{equation}
    \bl_i(\btheta_0) = -\ol\bH_i \bd_i = -\bI_0\bd_i + (\bI_0 - \ol\bH_i)\bd_i. \label{Heq:MEstL}
\end{equation}
Under the conditions given in \cite{Zhang2013}, $\lrn{\bd_i}_2 = \Op(n^{-1/2})$ and $\lrn{\bI_0 - \ol\bH_i}_2 = \Op(n^{-1/2})$. In \eqref{Heq:MEstL}, $(\bI_0 - \ol\bH_i)\bd_i$ becomes subordinate -- $\lrn{(\bI_0 - \ol\bH_i)\bd_i}_2 = \Op(n^{-1/2}\lrn{\bd_i}_2)$. All these facts lead to the ``linearity'' between $\bl_i(\btheta_0)$ and $\bd_i$:
\begin{equation*}
    \bl_i(\btheta_0) \approx -\bI_0\bd_i. \label{Heq:LID}
\end{equation*}
In practice, only $\bl_i(\btheta)$ is available, so we consider $\bg_i = \bl_i(\btheta) - \bI_0\bDelta$ as the practical implementation of $\bl_i(\btheta_0)$. The association between $\bg_i$ and $\bd_i$ can be quantified by expanding $\bg_i$ around $\btheta_0$. 

Let $\wt\bH_i = n^{-1}\sum_j \int_0^1 \nabla^2 L\lri{\btheta_0 + t\bDelta;\bX_{ij}}\diff t$. Applying \eqref{Heq:MEstL},
\begin{equation}
    \bg_i = \bl_i(\btheta_0) + \wt\bH_i\bDelta - \bI_0\bDelta \nonumber = -\bI_0\bd_i + \be_i, \label{Heq:GID}
\end{equation}
where the residual $\be_i = (\bI_0 - \ol\bH_i)\bd_i + (\wt\bH_i - \bI_0)\bDelta$. The form $\bg_i = -\bI_0\bd_i + \be_i$ indicates a ``regression'' estimator of $\bI_0$:
\begin{equation}
    \wt\bI_0(\btheta) = -\lrI{\sum\bd_i\bd_i^\T}^{-1}\lrI{\sum\bd_i\bg_i^\T},
\end{equation}
and its estimation error is
\begin{equation}
    \wt\bI_0(\btheta) - \bI_0 = -\lrI{\sum\bd_i\bd_i^\T}^{-1}\lrI{\sum\bd_i\be_i^\T}. \label{Heq:IEstErr}
\end{equation}
In \eqref{Heq:IEstErr}, the denominator $m^{-1}\sum\bd_i\bd_i^\T$ is less a concern: the pseudo linearity $\bl_i(\btheta_0) \approx -\bI_0\bd_i$ leads to the fact that $m^{-1}\sum n\bd_i\bd_i^\T \approx m^{-1}\bI_0^{-1}\sum n\bl_i(\btheta_0)\bl_i(\btheta_0)^\T\bI_0^{-1}  \cid \bI_0^{-1}\bQ_{11}\bI_0^{-1}$. 

On the other hand, for the numerator, by definition, 
\begin{align*}
    m^{-1}\sum\bd_i\be_i^\T 
    %= & m^{-1}\sum_i \bd_i\bd_i^\T(\bI_0 - \ol\bH_i) + m^{-1}\sum_i \bd_i\bDelta_1^\T(\wt\bH_i - \bI_0) \nonumber \\
    = & m^{-1}\sum_i \bd_i\bd_i^\T(\bI_0 - \bH_i) + m^{-1}\sum_i \bd_i\bd_i^\T(\bH_i - \ol\bH_i) + \\ %\nonumber \\
        & m^{-1}\sum_i \bd_i\bDelta^\T(\bH_i - \bI_0) + m^{-1}\sum_i \bd_i\bDelta^\T(\wt\bH_i - \bH_i). %\label{Heq:deDec}
\end{align*}
Bounding each part, we have
\begin{equation*}
    m^{-1}\sum_j \bd_i\be_i^\T \approx n^{-1} (\bI_{d} \otimes \bDelta^\T) \bQ_{12} + \bR_D',
\end{equation*}
where $\lrn{\bR_D'} = \Op(n^{-2} + n^{-3/2}m^{-1/2}) + \Op(n^{-1/2}m^{-1/2}\lrn{\bDelta}_2^2)$. Together with the denominator, %we have %$m^{-1}\sum n\bd_i\bd_i^\T \cid \bI_0^{-1}\bQ_{11}\bI_0^{-1}$,  
\begin{equation*}
    \wt\bI_0(\btheta) - \bI_0 = \lrI{\sum n\bd_i\bd_i^\T}^{-1}\lrI{\sum n\bd_i\be_i^\T} \approx \bI_0\bQ_{11}^{-1}\bI_0(\bI_{d} \otimes \bDelta^\T) \bQ_{12} + \bR_D'', 
\end{equation*}
with $\lrn{\bR_D''}_2 = \Op(n^{-1} + n^{-1/2}m^{-1/2} + n^{1/2}m^{-1/2}\lrn{\bDelta_1}_2^2)$.

For practical applications, we implement the sample-based estimators $\wh\bg_i$ and $\wh\bd_i$ to approximate the unknown $\bg_i$ and $\bd_i$, which gives us $\wh\bI_0(\btheta)$. Meanwhile, another implication of \eqref{Heq:GID} is $\bd_i = -\bI_0^{-1}\bg_i + \bI_0^{-1}\be_i$, and this brings the estimator of $\bI_0^{-1}$ --- $\wh\bOmega(\btheta)$, whose distance to $\bI_0^{-1}$ can be quantified following the similar steps.

\color{black}

\section{Main Results  \label{Hsec:Main}}
In this section, we present both iterative and one-shot Newton-like optimization algorithms using the  GM and MG Fisher information estimators.  

\subsection{The Iterative Distributed Estimators}

Based on Proposition \ref{Hpr:DGGD}, $\wh\bI_0(\btheta)$ and $\wh\bOmega(\btheta)^{-1}$ can be applied to replace the unattainable global Hessian in each iteration of Newton's method. We propose our $M$-estimator/gradient Newton-type optimization methods in Algorithms \ref{Halg:RHE} and \ref{Halg:RIHE}. 
%Let $\btheta^*$ be the ``oracle'' $M$-estimator, that is, $\btheta^* = \argmin_{\btheta \in \Theta} n^{-1}m^{-1}\sum_i\sum_j L(\btheta;\bX_{ij})$, and $\bDelta^* = \btheta^* - \btheta_0$ be its error. 

\begin{algorithm}[ht]
\begin{algorithmic}[1]
\caption{\label{Halg:RHE} Distributed Newton-Like Optimization with MG estimators}
\STATE Collect all $M$-estimators $\lrm{\wh\btheta_i}_{i = 1}^m$ and compute their average $\ol\btheta = m^{-1}\sum_i\wh\btheta_i$; 
\STATE Input the initial estimator $\wt\btheta_0$ and the total number of iterations $T$; 
\FOR{$t = 0, \dots, T-1$}
\STATE Broadcast the current estimator $\wt\btheta_t$ to all $m$ parallel centers $\lrm{\mathcal{M}_i}_{i=1}^m$;
    \FOR{$i = 1, \dots, m$}
    \STATE At center $\mathcal{M}_i$, compute the gradient $\bl_i(\wt\btheta_t) = n^{-1}\sum_{j \in \mathcal{M}_i} \nabla L(\wt\btheta_t;\bX_{j})$; 
    \STATE Return $\bl_i(\wt\btheta_t)$ to the local center $\mathcal{L}$;
    \ENDFOR
    \STATE At $\mathcal{L}$, compute the global gradient $\ol\bl(\wt\btheta_t) = m^{-1}\sum_i\bl_i(\wt\btheta_t)$; 
    \STATE Construct the MG Fisher information estimator:  
    $$ 
        \wh\bI_{0}(\wt\btheta_t) = -\lrM{\sum_i (\wh\btheta_i - \ol\btheta)(\wh\btheta_i - \ol\btheta)^\T}^{-1} \lrO{\sum_i (\wh\btheta_i - \ol\btheta)\lrM{\bl_i(\wt\btheta_t) - \ol\bl(\wt\btheta_t)}^\T};
      $$
    \STATE Update $\wt\btheta_{t + 1} = \wt\btheta_t - \wh\bI_{0}(\wt\btheta_t)^{-1}\ol\bl(\wt\btheta_t)$;
\ENDFOR
\RETURN $\wt\btheta_{T}$.
\end{algorithmic}
\end{algorithm}

\begin{algorithm}[ht]
\begin{algorithmic}[1]
\caption{\label{Halg:RIHE} Distributed Newton-Like Optimization with GM estimators}
\STATE Collect all $M$-estimators $\lrm{\wh\btheta_i}_{i = 1}^m$ and compute their average $\ol\btheta = m^{-1}\sum_i\wh\btheta_i$; 
\STATE Input the initial estimator $\wt\btheta_0'$ and the total number of iterations $T$; 
\FOR{$t = 0, \dots, T-1$}
\STATE Broadcast the current estimator $\wt\btheta_t'$ to all $m$ parallel centers $\lrm{\mathcal{M}_i}_{i=1}^m$;
    \FOR{$i = 1, \dots, m$}
    \STATE At center $\mathcal{M}_i$, compute the gradient $\bl_i(\wt\btheta_t') = n^{-1}\sum_{j \in \mathcal{M}_i} \nabla L(\wt\btheta_t';\bX_{j})$; 
    \STATE Return $\bl_i(\wt\btheta_t')$ to the local center $\mathcal{L}$;
    \ENDFOR
    \STATE At $\mathcal{L}$, compute the global gradient $\ol\bl(\wt\btheta_t') = m^{-1}\sum_i\bl_i(\wt\btheta_t')$; 
    \STATE Construct the GM Fisher information estimator:  
    $$
    \wh\bOmega(\wt\btheta_t') = -\lrO{\sum_i\lrM{\bl_i(\wt\btheta_t') - \ol\bl(\wt\btheta_t')}\lrM{\bl_i(\wt\btheta_t') - \ol\bl(\wt\btheta_t')}^\T}^{-1} \lrO{\sum_i\lrM{\bl_i(\wt\btheta_t') - \ol\bl(\wt\btheta_t')}(\wh\btheta_i - \ol\btheta)^\T};
    $$
    \STATE Update $\wt\btheta_{t + 1}' = \wt\btheta_t' - \wh\bOmega(\wt\btheta_t')\ol\bl(\wt\btheta_t')$;
\ENDFOR
\RETURN $\wt\btheta_{T}'$.
\end{algorithmic}
\end{algorithm}

Under conditions listed in Section \ref{Hsec:As}, these methods converge stochastically faster than the existing method does. 

Define $$\bQ_{12}\circ\bu = (\bI_d\otimes \bu^\T)\bQ_{12} \bu, \bQ\circ\bu = (\bI_d\otimes \bu^\T)\bQ \bu, \bu \in \mathbb{R}^d.$$ 
Let $\btheta^*$ be the ``oracle'' $M$-estimator: $$\btheta^* = \argmin_{\btheta \in \Theta} N^{-1}\sum_{j \in \mathcal{G}} L(\btheta;\bX_{j}),$$ and $\bDelta^* = \btheta^* - \btheta_0$. 
\begin{theorem} \label{Hthm:DistEstOra}  Let $\bDelta_t = \wt\btheta_t - \btheta_0$ and $\bDelta_t' = \wt\btheta_t' - \btheta_0$ for $t = 0, \dots, T$. When the initial estimators $\wt\btheta_0$ and $\wt\btheta_0' \in U(\rho)$, and $\lrn{\bDelta_0}_2 = \Op(n^{-1/2})$ and $\lrn{\bDelta_0'}_2 = \Op(n^{-1/2})$, consider the estimators $\wt\btheta_{t + 1}$ and $\wt\btheta_{t + 1}'$ defined in Algorithms \ref{Halg:RHE} and \ref{Halg:RIHE}:
\begin{align*}
    & \wt\btheta_{t + 1} - \btheta^* = \bQ_{11}^{-1}\lri{\bQ_{12}\circ\bDelta_t} -2^{-1}\bI_0^{-1} \lri{\bQ\circ\bDelta_t} + \bR^*_1 \\
    & \wt\btheta_{t + 1}' - \btheta^* = \bQ_{11}^{-1}\lri{\bQ_{12}\circ\bDelta_t'} -2^{-1}\bI_0^{-1} \lri{\bQ\circ\bDelta_t'} + \bR_2^*,
\end{align*}
where the residual terms
\begin{align*}
    & \lrn{\bR^*_1}_2 = \Op\lrm{(m^{-1} + n^{-1})(\lrn{\bDelta_t}_2 + \lrn{\bDelta^*}_2) + m^{-1}n^{-1}} \\ 
    & \lrn{\bR^*_2}_2 = \Op\lrm{(m^{-1} + n^{-1})(\lrn{\bDelta_t'}_2 + \lrn{\bDelta^*}_2) + m^{-1}n^{-1}}.
    \end{align*}
\end{theorem}

Theorem \ref{Hthm:DistEstOra} indicates that when $\lrn{\bDelta_t}_2 = \Op(n^{-1/2})$ and $\lrn{\bDelta_t'}_2 = \Op(n^{-1/2})$, in probability $\wt\btheta_{t + 1}$ and $\wt\btheta_{t + 1}'$ converge stochastically faster than the CSL method does in each iteration. That is, under the conditions in Theorem \ref{Hthm:DistEstOra}, after each iteration, $$\lrn{\bDelta_{t+1}}_2 = O(\lrn{\bDelta_t}_2^2) + n^{-1/2}\Op(n^{-1} + m^{-1}).$$
The error of the updated estimator is bounded by the square of its predecessor's estimation error, which is of the same order of magnitude as Newton's method with the global Hessian. The same conclusion can be drawn for the GM estimator $\btheta_{t+1}'$ as well.

Also, our algorithms have no concern about accidentally choosing a ``bad'' local center. The whole iteration process depends on the initial estimator and the global data quality, instead of the local one. Therefore, when $n$ and $m$ are both large enough, they tend to be more stable. 

On the other hand, when $\bDelta_t$ is the dominating error contributor rather than $n$ or $m$, the explicit forms of those $\bDelta_t$-constant terms --- $\bQ_{11}^{-1}\lri{\bQ_{12}\circ\bDelta_t}$ and $2^{-1}\bI_0^{-1} \lri{\bQ\circ\bDelta_t}$ --- give us a chance to make an adjustment for them with only the local data to achieve better accuracy. Following this idea, we propose two one-shot algorithms in the next section, using the local $M$-estimator as the initial estimator. 

For the choices of the initial estimator, a popular one is the $M$-estimator of the local center $\mathcal{L}$, whose $l_2$ error is $\Op(n^{-1/2})$. Another one is the average of $M$-estimators, $\ol\btheta$, whose error is $\lrn{\ol\btheta - \btheta_0}_2 = \Op(n^{-1} + n^{-1/2}m^{-1/2})$ (see \cite{Zhang2013}). 

%In this case, Theorem \ref{Hthm:DistEstOra} indicates, after one and two rounds of communication, the updated estimators $\wt\btheta_1$ and $\wt\btheta_2$ have the distance $\lrn{\wt\btheta_1 - \btheta^*}_2 = \Op(n^{-1} + n^{-1/2}m^{-1})$ and $\lrn{\wt\btheta_2 - \btheta^*}_2 = \Op(n^{-2} + n^{-1/2}m^{-3/2})$. So do $\wt\btheta_1'$ and $\wt\btheta_2'$. Therefore, when $m = o(n^3)$, these two proposed methods usually stabilize within two or three rounds of communication with a sufficiently large $n$. 

\subsection{Bias-Adjusted One-Step Estimators}

Often, in practice, communication among data centers is limited and one-step estimation is much more preferred. In this case, a common choice of the initial estimator is the $n^{1/2}$-consistent local $M$-estimator obtained within $\mathcal{L}$. When $m = o(n^2)$, Theorem \ref{Hthm:DistEstOra} indicates the $\bDelta_0$-constant terms $\bI_0^{-1}\lri{\bQ\circ\bDelta_0}$ and $\bQ_{11}^{-1}\lri{\bQ_{12}\circ\bDelta_0}$ are the dominating contributors of $\bDelta_1$. To adjust for these two terms, we propose to estimate $\bQ_{11}$, $\bQ$, and $\bQ_{12}$ locally in $\mathcal{L}$ with any $n^{1/2}$-consistent estimators; while estimation of $\bDelta_0$ requires some more accurate estimator $\btheta^A$, and this can be done after one iteration of collecting gradients and $M$-estimators from parallel centers. 

\begin{prop} \label{Hthm:EEst} %Consider the estimation of $\bI_0^{-1}\lri{\bQ\circ\bDelta_1}$ and $\bQ_{11}^{-1}\lri{\bQ_{12}\circ\bDelta_1}$. 
For some estimator $\btheta^A \in U(\rho)$, define $\bDelta^A = \btheta^A - \btheta_0$, and let $\wh\bDelta_0 = \wt\btheta_0 - \btheta^A$. Construct the local estimators of $\bI_0$, $\bQ_{11}$, and $\bQ_{12}$ with $\btheta^A$:
\begin{align*}
    & \wh\bH^{A} =  n^{-1}\sum_{j \in \mathcal{L}} \nabla^2 L(\btheta^A;\bX_j) \\
    & \wh\bQ =  n^{-1}\sum_{j \in \mathcal{L}} \nabla^3 L(\btheta^A;\bX_{j}) \\
    & \wh\bQ_{11} =  n^{-1}\sum_{j \in \mathcal{L}} \lrM{\nabla L(\btheta^A; \bX_j) - n^{-1}\sum_{j \in \mathcal{L}} \nabla L(\btheta^A; \bX_j)} \lrM{\nabla L(\btheta^A; \bX_j) - n^{-1}\sum_{j \in \mathcal{L}} \nabla L(\btheta^A; \bX_j)}^\T \\
    & \wh\bQ_{12} = n^{-1}\sum_{j \in \mathcal{L}} \lrM{\nabla L(\btheta^A; \bX_j) - n^{-1}\sum_{j \in \mathcal{L}} \nabla L(\btheta^A; \bX_j)} \otimes \lrM{\nabla^2 L(\btheta^A;\bX_j) - n^{-1}\sum_{j \in \mathcal{L}} \nabla^2 L(\btheta^A;\bX_j)}. 
\end{align*}
Under the conditions in Theorem \ref{Hthm:DistEstOra}, when $\lrn{\bDelta^A}_2 = \Op(n^{-1/2})$,
\begin{align*}
    & \lrN{\lri{\wh\bQ\circ\wh\bDelta_0} - \lri{\bQ\circ\bDelta_0}}_2 = \Op(\lrn{\bDelta^A}_2^2) + \Op(\lrn{\bDelta^A}_2\lrn{\bDelta_0}_2) + \Op(n^{-1/2} \lrn{\bDelta_0}_2^2) \\
    & \lrN{\lri{\wh\bQ_{12}\circ\wh\bDelta_0} - \lri{\bQ_{12}\circ\bDelta_0}}_2 = \Op(\lrn{\bDelta^A}_2^2) + \Op(\lrn{\bDelta^A}_2\lrn{\bDelta_0}_2) + \Op(n^{-1/2} \lrn{\bDelta_0}_2^2) \\
    & \lrN{(\wh\bH^{A})^{-1}\lri{\wh\bQ\circ\wh\bDelta_0} - \bI_0^{-1}\lri{\bQ\circ\bDelta_0}}_2 = \Op(\lrn{\bDelta^A}_2^2) + \Op(\lrn{\bDelta^A}_2\lrn{\bDelta_0}_2) + \Op(n^{-1/2} \lrn{\bDelta_0}_2^2) \\
    & \lrN{\wh\bQ_{11}^{-1}\lri{\wh\bQ_{12}\circ\wh\bDelta_0} - \bQ_{11}^{-1}\lri{\bQ_{12}\circ\bDelta_0}}_2 = \Op(\lrn{\bDelta^A}_2^2) + \Op(\lrn{\bDelta^A}_2\lrn{\bDelta_0}_2) + \Op(n^{-1/2} \lrn{\bDelta_0}_2^2).
\end{align*}
\end{prop}

Among other choices of $\btheta^A$, we suggest to use the average of $M$-estimators $\ol\btheta$ with the error $\lrn{\ol\btheta - \btheta_0}_2 = \Op(n^{-1/2}m^{-1/2} + n^{-1})$. It is more stable and does no depends on the local center. Otherwise, if $m$ is small or the local data quality is satisfactory, the one-step updated estimator $\wt\btheta_1$ or $\wt\btheta_1'$ can be taken as well. 

\begin{theorem} \label{Hthm:OSEst} When the conditions in Theorem \ref{Hthm:DistEstOra} hold, take the local $M$-estimator as $\wt\btheta_0$, and set $\btheta^A = \ol\btheta$. %the local estimators $\wh\bQ_{11}$ and $\wh\bQ_{12}$ defined in Proposition \ref{Hthm:EEst}
\begin{enumerate}[label={(\arabic*)}]
\item Consider the one-step bias-adjusted distributed estimator based on $\wh\bI_0(\wt\btheta_0)$:
\begin{equation*}
\wt\btheta_{os} = \wt\btheta_0 - \wh\bI_0(\wt\btheta_0)^{-1} \ol \bl(\wt\btheta_0) - \wh\bQ_{11}^{-1}\lri{\wh\bQ_{12}\circ\wh\bDelta_0} + \frac{1}{2}(\wh\bH^{A})^{-1} \lri{\wh\bQ\circ\wh\bDelta_0}.
\end{equation*}
Then, $\lrn{\wt\btheta_{os} - \btheta^*}_2 = \Op(n^{-3/2} + m^{-1}n^{-1/2})$.
\item Consider the one-step bias-adjusted distributed estimator based on $\wh\bOmega(\wt\btheta_0)$:
\begin{equation*}
\wt\btheta_{os}' = \wt\btheta_0 - \wh\bOmega(\wt\btheta_0) \ol \bl(\wt\btheta_0) - \wh\bQ_{11}^{-1}\lri{\wh\bQ_{12}\circ\wh\bDelta_0} + \frac{1}{2}(\wh\bH^{A})^{-1} \lri{\wh\bQ\circ\wh\bDelta_0}.
\end{equation*}
Then, $\lrn{\wt\btheta_{os}' - \btheta^*}_2 = \Op(n^{-3/2} + m^{-1}n^{-1/2})$.
\end{enumerate}
\end{theorem}

It is well known that the oracle estimator $\btheta^*$ converges in distribution to $\Gau(\bzero, \bI_0^{-1}\bQ_{11}\bI_0^{-1})$. Theorem \ref{Hthm:OSEst} indicates that both $\wt\btheta_{os}$ and $\wt\btheta_{os}'$ have the same limiting distributions when $m = o(n^2)$, which means they both achieve the optimal asymptotic efficiency as the oracle estimator $\btheta^*$ does. Thus, Gaussian approximation can be applied with any consistent estimators of $\bI_0$ and $\bQ_{11}$ (for example, see Proposition \ref{Hthm:EEst}) for statistical inference and construction of confidence intervals. 

\begin{corollary} \label{Hcrl:OSEstGau}Consider $\wt\btheta_{os}$ and $\wt\btheta_{os}'$ defined in Theorem \ref{Hthm:OSEst}. Under the conditions in Theorem \ref{Hthm:OSEst}, when $m = o(n^2)$,
\begin{align*}
    & (nm)^{1/2}\lrm{\wh\bOmega(\wt\btheta_0)\wh\bQ_{11}\wh\bOmega(\wt\btheta_0)^\T}^{-1/2}(\wt\btheta_{os} - \btheta_0) \cid \Gau(\bzero, \bI) \\
    & (nm)^{1/2}\lrm{\wh\bOmega(\wt\btheta_0)\wh\bQ_{11}\wh\bOmega(\wt\btheta_0)^\T}^{-1/2}(\wt\btheta_{os}' - \btheta_0) \cid \Gau(\bzero, \bI).
\end{align*}
\end{corollary}

\subsection{Technical Assumptions \label{Hsec:As}}
%We list here the necessary existence assumptions and define the thrice derivation of the loss function with its Lipschitz continuity assumption. 
In this section, we list the convexity and identifiability assumptions used in our technical analysis. 
\begin{assump} \label{as:T0Def}
The parameter space $\bTheta \subset \mathbb{R}^d$ is a compact convex set, with $\btheta \in \text{int}(\bTheta)$ and $l_2$-radius $r_0 = \max_{\btheta \in \bTheta}\lrn{\btheta - \btheta_0}_2$.
\end{assump}

\begin{assump} \label{as:T0Ind}
That $\btheta_0 \in \bTheta$ is the unique minimizer of the population risk $\E_{\mathcal{P}}\lrm{L(\btheta;\bX_i)}$. For any $\delta > 0$, there exists $\ve > 0$, such that
$$\liminf_{n \to \infty} \P\lrO{\inf_{\lrn{\btheta - \btheta_0}_2 \ge \delta} \lrM{n^{-1}\sum_{i=1}^n L(\btheta;\bX_i) - n^{-1}\sum_{i=1}^n L(\btheta_0;\bX_i)} \geq \ve} = 1.$$
\end{assump}

\begin{assump} \label{Has:LH} The population risk is twice-differentiable, and there exist finite constants $\lambda_-, \lambda_+, \lambda_1, \lambda_2 > 0$ such that $\lambda_-\bI  \preceq \bI_0  \preceq \lambda_+\bI,$
    $$\E_\mathcal{P}\lrm{\lrn{\nabla L(\btheta_0;\bX)}_2^6} \leq \lambda_1^6 \text{ and } \E_\mathcal{P}\lrm{\lrn{\nabla^2 L(\btheta_0;\bX)}_2^6} \leq \lambda_2^6.$$
Additionally, for all $\btheta', \btheta \in U$,
    $$\lrn{\nabla^2 L(\btheta';\bX) - \nabla^2 L(\btheta;\bX)}_2 \leq h(\bX)\lrn{\btheta' - \btheta}_2$$
We assume $\max\lri{\E_\mathcal{P}\lrm{h(\bX)^6}, \E_\mathcal{P}\lro{\lrA{h(\bX) - \E_\mathcal{P}\lrm{h(\bX)}}^6}} \leq \lambda_h^6$ for some finite positive constant $\lambda_h$.

\end{assump}
We further assume the loss function is three-times differentiable with enough continuity in $U(\rho)$. Recall the definition of the third derivative tensor operator of $L(\btheta;\bX)$, $\bT(\btheta, \bX) = \nabla^3 L(\btheta;\bX)$, for $\bu \in \mathbb{R}^d$:
    \begin{align*}
    & \bT(\btheta, \bX) = \nabla^3 L(\btheta;\bX) = \begin{bmatrix} \nabla^2 \lrm{\frac{\partial}{\partial \theta_{1}}L(\btheta;\bX)} \\ \vdots \\ \nabla^2 \lrm{\frac{\partial}{\partial \theta_{d}}L(\btheta;\bX)} \end{bmatrix}. 
    %& \bT(\btheta, \bX)(\bu) = (\bI_d \otimes \bu^\T) \times \nabla^3 L(\btheta;\bX) = \begin{bmatrix} \bu^\T\nabla^2 \lrM{\frac{\partial}{\partial \theta_{1}}L(\btheta;\bX)} \\ \vdots \\ \bu^\T\nabla^2 \lrM{\frac{\partial}{\partial \theta_{d}}L(\btheta;\bX)}\ \end{bmatrix}
    %& \bT(\btheta, \bX_i)(\bu^{\otimes 2}) = \nabla^3 L(\btheta;\bX_i)(\bu^{\otimes 2}) = \begin{bmatrix} \bu^\T\nabla^2 \lrm{\frac{\partial}{\partial \theta_{1}}L(\btheta;\bX_i)}\bu \\ \vdots \\ \bu^\T\nabla^2 \lrm{\frac{\partial}{\partial \theta_{d}}L(\btheta;\bX_i)}\bu \end{bmatrix},
    %& \bT(\btheta, \bX_i)(\bu^{\otimes 3}) = \nabla^3 L(\btheta;\bX_i)(\bu^{\otimes 3}) = \sum_i\sum_j\sum_k \lrM{\frac{\partial^3}{\partial \theta_i\partial \theta_j\partial \theta_k}L(\btheta;\bX_i)} u_i u_j u_k.
    \end{align*}
%\end{definition}
Note that $\bQ$ defined in \eqref{Heq:QTDef} is the expectation of $\bT(\btheta_0, \bX)$.
The definition of $\bT(\btheta, \bX)$ indicates, for $\btheta \in U(\rho)$,
\begin{equation*}
    \nabla^2 L(\btheta; \bX) - \nabla^2 L(\btheta_0; \bX) 
= \int_0^1 \lrm{\bI_d \otimes (\btheta - \btheta_0)^\T} \bT\lrm{\btheta_0 + t(\btheta - \btheta_0), \bX} \diff t. 
\end{equation*}
When $\btheta$ is close to $\btheta_0$, and $\bT(\btheta, \bX)$ is smooth enough with respect to $\btheta$, $\nabla^2 L(\btheta; \bX) - \nabla^2 L(\btheta_0; \bX) \approx \lrm{\bI_d \otimes (\btheta - \btheta_0)^\T} \bT\lri{\btheta_0, \bX}$. The condition is specified as follows. 

\begin{assump} \label{Has:L3Cont} The risk function is three-times differentiable, and $\bT(\btheta, \bX)$ is $m(\bX)$-Lipschitz. That is, for all $\btheta', \btheta \in U(\rho)$ and $\bu \in \mathbb{R}^d$, 
    \begin{equation*}
    \lrn{\lri{\bI_d \otimes \bu^\T}\lrm{\bT(\btheta', \bX) - \bT(\btheta, \bX)}}_2 
    \leq  m(\bX)\lrn{\btheta - \btheta'}_2\lrn{\bu}_2, 
    \end{equation*}
with $\max\lri{\E_\mathcal{P}\lrm{m(\bX)^4}, \E_\mathcal{P}\lro{\lrm{m(\bX) - \E\lrm{m(\bX)}}^4}} \leq \lambda_m^4$ for some finite positive constant $\lambda_m$. Additionally, at $\btheta_0$, $\lrn{\bT(\btheta_0, \bX)}_2$ is bounded by $g(\bX)$ for some function $g(\bX)$, and $\E_\mathcal{P}\lrm{g(\bX)^4} \leq \lambda_3^4$.
\end{assump}

Assumption \ref{Has:L3Cont} also implicates Assumption \ref{Has:LH}. That is, under Assumption \ref{Has:L3Cont}, $\lrn{\bT(\btheta, \bX)}_2 \leq g(\bX) + \rho m(\bX).$
We can choose $h(\bX) = g(\bX) + \rho m(\bX)$, so that for any $\btheta', \btheta \in U(\rho)$
\begin{align*}
\lrn{\nabla^2 L(\btheta'; \bX) - \nabla^2 L(\btheta; \bX)}_2 
& \leq \int_0^1 \lrn{\bT\lrm{\btheta + t(\btheta' - \btheta), \bX}(\btheta' - \btheta)}_2 \diff t \\
& \leq h(\bX)\lrn{\btheta' - \btheta}_2.
\end{align*}

\section{Simulations}\label{Hsec:Simu}

In this section, we validate and visualize our methodology with extensive synthetic examples. Our presentation starts with evaluating the accuracy of $\wh\bI_0(\btheta)$ and $\wh\bOmega(\btheta)$. In the second part, the convergence rates of $\wt\btheta_t$ and $\wt\btheta_t'$ proposed in Algorithms \ref{Halg:RHE} and \ref{Halg:RIHE} are compared with existing methods, when multiple rounds of communications among centers are eligible. Then, in the cases of one-step estimation, we demonstrate the distributions and the empirical coverage probabilities of $\wt\btheta_{os}$ and $\wt\btheta_{os}'$ proposed in Theorem \ref{Hthm:OSEst} with different $n$ and $m$. %The final part is an application to high-dimensional models. 

Herein, suppose $\bX$ has two components: $\bX = \lri{Y, \bS^\T}^\T$, where $Y$ is the response variable, and $\bS$ is the predictor. Let $\btheta_0$ be the true coefficient describing the association between $Y$ and $\bS$. We show examples of Poisson regression and logistic regression with canonical links, with $\bS$ coming in different sizes and different distributions. It is our pragmatic experience that the distribution of the design matrices makes no substantial difference, as long as those moment assumptions hold.

\subsection{Estimation Accuracy for Fisher Information}
    
The estimation accuracy of $\wh\bI_0(\btheta)$ and $\wh\bOmega(\btheta)$ is quantified by the relative distance to $\bI_0$ and $\bI_0^{-1}$, that is, the Fisher information estimation error --- $\delta_1(\bA) = \lrn{\bI_0}_2^{-1}\lrn{\bI_0 - \bA}_2$, and the inverse Fisher information estimation error --- $\delta_2(\bA) = \lrn{\bI_0^{-1}}_2^{-1}\lrn{\bI_0^{-1} - \bA^{-1}}_2$ for some non-trivial matrix $\bA$. We consider an example of logistic regression:
	\begin{equation}
	    \P(Y_i = 1\mid\bS_i, \btheta_0) = \frac{\exp(\bS_i^\T\btheta_0)}{1 + \exp(\bS_i^\T\btheta_0)}, \label{Heq:Log}
	\end{equation}
	where $d = 4$, $\btheta_0 = d^{-1/2}\bone$, and $\bS_i \sim \Gau(\bzero, \bI)$.
	
	Proposition \ref{Hpr:DGGD} indicates $\bdelta_1$ and $\bdelta_2$ are affected by the error $\bDelta$, the center-wise sample size $n$, and the total number of centers $m$. To demonstrate a comprehensive picture, we sample $\btheta$ from $\Gau(\btheta_0, \sigma^2\bI)$ with $\sigma^2 = 16^{-1}, 256^{-1}$, and $65536^{-1}$. Also, $n$ and $m$ take 100, 200, 400, and 800 separately. For each combination of $\lrm{\sigma^2, n, m}$, 10000 simulations are repeated. 

    The results are visualized in Figure \ref{Hfg:Hhat}. We compare four estimators of $\bI_0$ and their inverses for $\bI_0^{-1}$: the local estimator $\bH^{(l)}(\btheta) = n^{-1}\sum_{j \in \mathcal{L}} \nabla^2 L(\btheta;\bX_j)$ (LC, black lines), the global estimator $\bH^{(g)}(\btheta) = N^{-1}\sum_{j \in \mathcal{G}} \nabla^2 L(\btheta;\bX_{j})$ (GL, green lines), the proposed MG Fisher information estimator $\wh\bI_0(\btheta)$ (MG, red dashed lines), and the GM Fisher information estimator $\wh\bOmega(\btheta)^{-1}$ (GM, red solid lines). 
    
    From Figure \ref{Hfg:Hhat}, we observe that with fixed $\sigma^2$ and $m$, the estimation errors of both proposed methods drop as $n$ increases; while $\bH^{(l)}(\btheta)$ and $\bH^{(g)}(\btheta)$ remain almost the same, when $\btheta$ is far from $\btheta_0$. It is not surprising that our proposed estimators are closer to the truth. Recall that in generalized linear models, $\bQ_{12} = \bzero$; in this scenario, Proposition \ref{Hpr:DGGD} indicates that the impacts of $\bDelta$ on $\wh\bI_0(\btheta)$ and $\wh\bOmega(\btheta)^{-1}$ diminish, as $n$ and $m$ increase. 
    
    \begin{figure*}[!ht]
    \centering 
    \includegraphics[width=.91\textwidth]{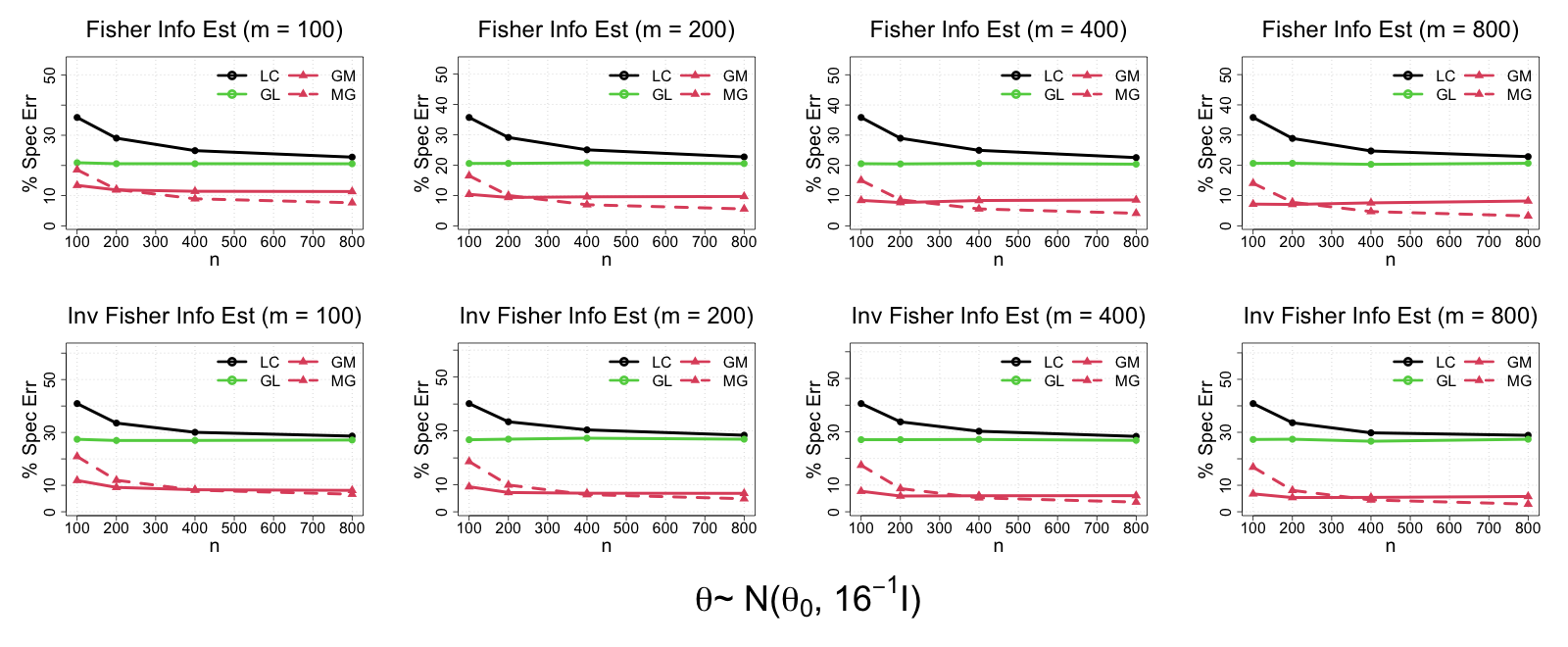} 
    \includegraphics[width=.91\textwidth]{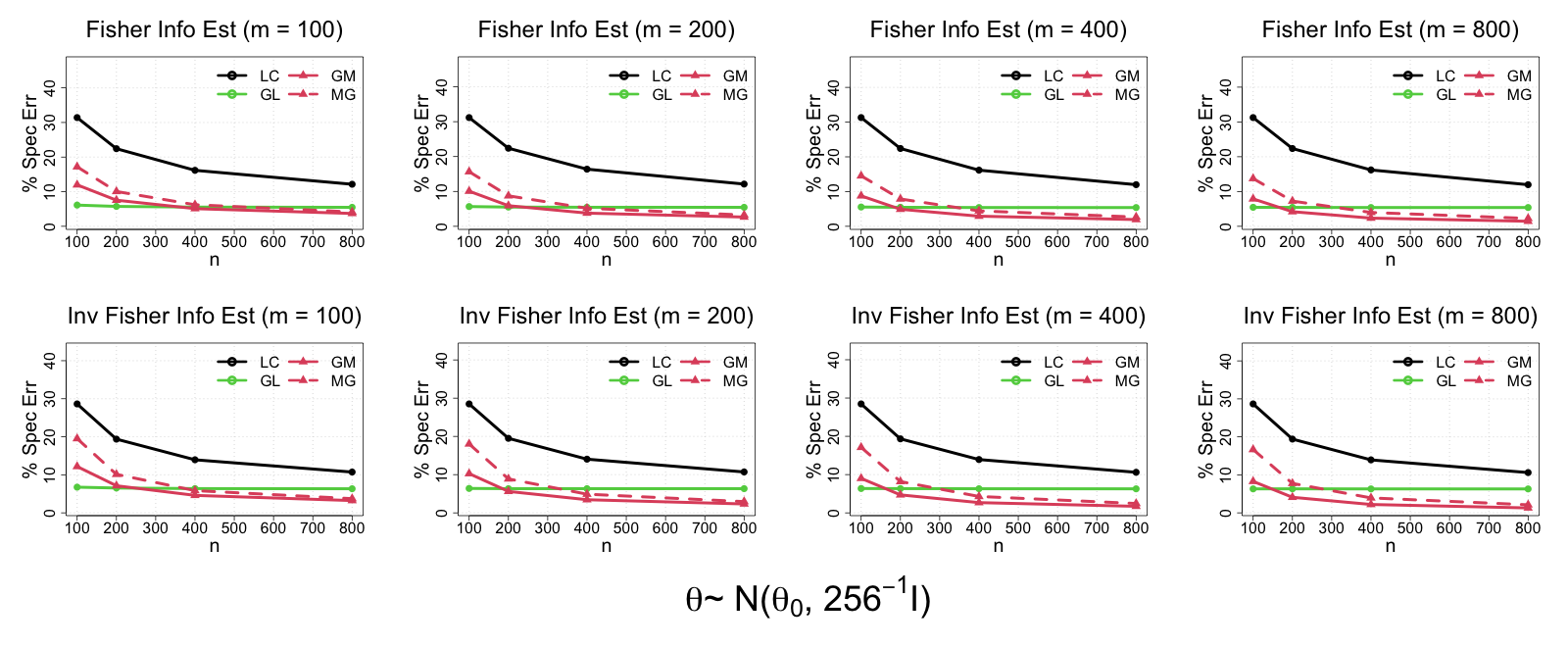} 
    \includegraphics[width=.91\textwidth]{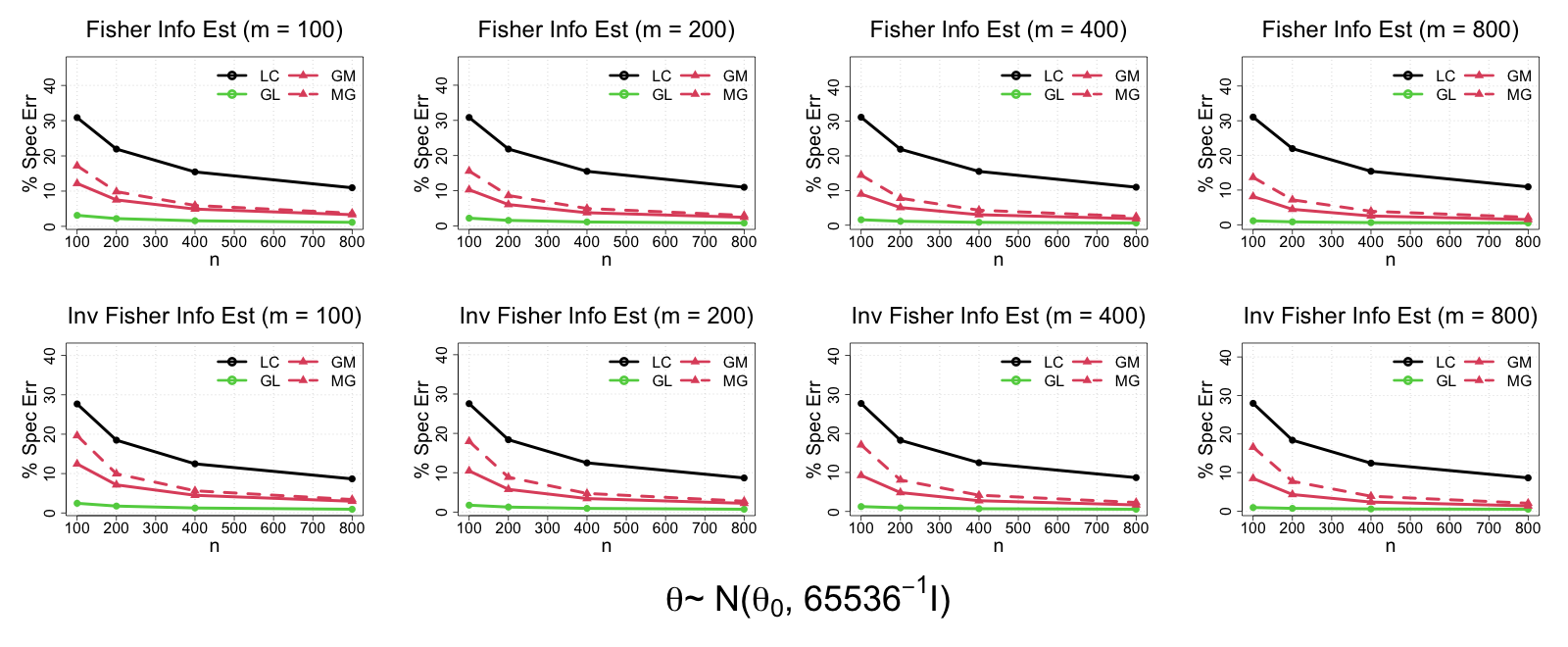} 
    \caption{ Estimation accuracy for $\bI_0$ and $\bI_0^{-1}$. In each setting of $\sigma^2$, the first row stands for the relative estimation errors regarding $\bI_0$, and the second row for $\bI_0^{-1}$. Four estimators are compared: black lines for the local Hessian $\bH^{(l)}(\btheta)$ (LC), green lines for the global Hessian $\bH^{(g)}(\btheta)$ (GL), the red solid lines for $\wh\bOmega(\btheta)^{-1}$ (GM), and the red dashed lines for $\wh\bI_0(\btheta)$ (MG). Their inverses are used to estimate $\bI_0^{-1}$ in the second row. \label{Hfg:Hhat}}
    \end{figure*}

    On the other hand, if $\lrn{\btheta - \btheta_0}_2$ is small, $\delta_1\lrm{\bH^{(l)}(\btheta)}$ and $\delta_2\lrm{\bH^{(l)}(\btheta)}$ also drop and converge to their global counterparts. Similar phenomena can be observed in $\wh\bI_0(\btheta)$ and $\wh\bOmega(\btheta)$ as well. But both red lines drop much faster as $n$ increases. This is consistent with our conclusions in Proposition \ref{Hpr:DGGD}.

\subsection{Performance of Iterative Distributed Estimation}

    Our second presentation focuses on the iterative estimators $\wt\btheta_t$ and $\wt\btheta_t'$ proposed in Algorithm \ref{Halg:RHE} and \ref{Halg:RIHE}. To verify our conclusion, we consider a Poisson model: 
    \begin{equation}
        \P(Y_i = y\mid\bS_i, \btheta_0) = \frac{1}{y!} \exp\lrm{y\bS_i^\T\btheta_0 - \exp(\bS_i^\T\btheta_0)}, \label{Heq:Poi}
    \end{equation}
    where $d = 4$, $\btheta_0 = d^{-1/2}\bone$, and $\bS_i \sim \Gau(\bzero, \bI)$. 
    
    The initial estimator $\wt\btheta_0$ is set to be the average $M$-estimator $\ol\btheta$. Three iterations are performed. We take $n$ and $m$ to be 100, 200, 400, and 800 respectively, and compare the relative $l_2$ distance to the oracle estimator $\btheta^*$, $\delta_o({\btheta}) = \lrn{\btheta^*}_2^{-1}\lrn{\btheta^* - \btheta}_2$,  at different rounds of iterations. For each combination of $n$ and $m$, 10000 simulations are performed. 
    
    Figure \ref{Hfg:AVGPoi} presents the comparisons. In each sub-figure, the back line (CSL) stands for the CSL method proposed in \cite{MJ}. The green line (GL) stands for the ideal Newton estimator when we have all the data. The red dashed line and the red solid line stand for $\wt\btheta_t'$ (MG) and $\wt\btheta_t$ (GM) proposed in Algorithm \ref{Halg:RHE} and \ref{Halg:RIHE}. 
    
    This figure tells that our methods converge to $\btheta^*$ much faster than CSL does. In each setting, regardless of the ratio between $n$ and $m$, our methods are able to keep the relative errors below 5\% after two rounds of communications, and $<1\%$ after three rounds. By contrast, the performance of CSL depends heavily on $n$; that is, the local data quality has great impact on CSL. This confirms our conclusion and discussion in Theorem \ref{Hthm:DistEstOra}. 
    %\begin{figure*}[!ht]
    %    \centering
    %    \includegraphics[width=.9\textwidth]{Fig/AVG_Log_Gau_d3} 
    %    \caption{\label{Hfg:AVGLog} Relative $l_2$ distance to the oracle estimator $\btheta^*$ with logistic regression. Back lines stand for the original method $\wt\btheta^c$, and green lines stand for the ``ideal'' step-wise Newton estimator when we have all data. Red solid lines stand for $\btheta_2$ and red dashed lines for $\btheta_2'$.}
    %\end{figure*}
    
    \begin{figure*}[!ht]
        \centering
        \includegraphics[width=1\textwidth]{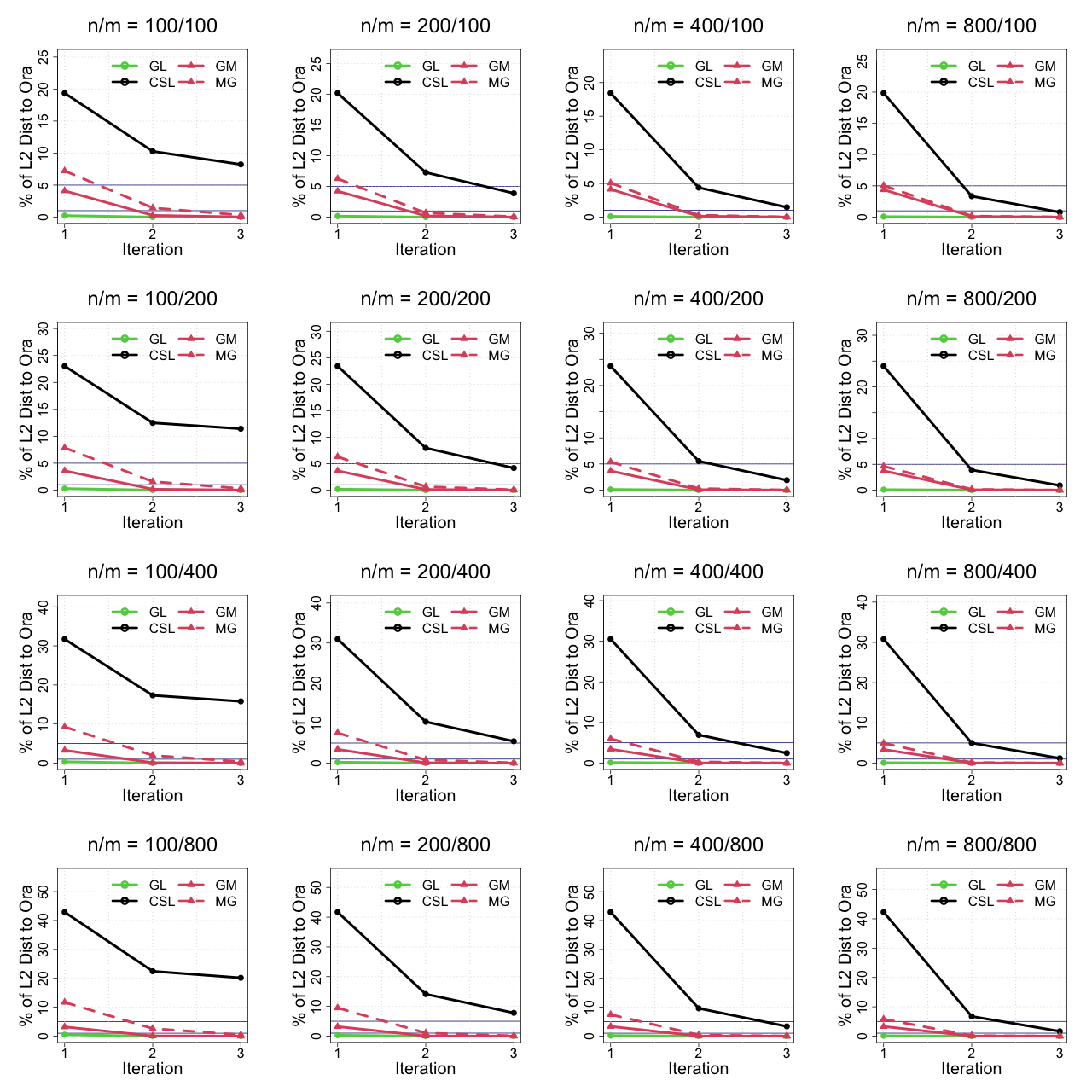} 
        \caption{ Relative $l_2$ distance to the oracle estimator. Back lines stand for the CSL method, and green lines for the Newton estimator (GL) with the global Hessian. Red solid lines stand for $\wt\btheta_t$ (GM), and red dashed lines for $\wt\btheta_t'$ (MG). \label{Hfg:AVGPoi}}
    \end{figure*}

\subsection{Performance of One-Step Distributed Estimation}
    The last presentation focuses on the the refined one-step distributed estimators $\wt\btheta_{os}$ and $\wt\btheta_{os}'$ discussed in Theorem \ref{Hthm:OSEst} and Corollary \ref{Hcrl:OSEstGau}. We state that they have the same asymptotic efficiency as $\btheta^*$ does when $m = o(n^2)$. To examine this, with different $m$ and $n$, we demonstrate plots of the proposed estimators versus the oracle estimator and the empirical confidence interval coverage. The results in this section focus on the first element of $\btheta_0$. 
	
	We take $n$ and $m$ to be 100, 200, 400, 800, and 1600 respectively. Four one-step methods are compared: the one-step CSL methods proposed in \cite{MJ}; the average $M$-estimator discussed in \cite{Zhang2013}; our two proposed methods, $\wt\btheta_{os}$ and $\wt\btheta_{os}'$. Three models are involved:
	\begin{itemize}
	\item ``Model 1'' is logistic regression \eqref{Heq:Log} with $\btheta_0 = d^{-1/2}\bone$, and $\bS_i \sim \Gau(\bzero, \bI)$;
	\item ``Model 2'' is logistic regression \eqref{Heq:Log} with $\btheta_0 = d^{-1/2}\bone$, and elements of $\bS_i$ are mutually independent and follow $\exp(1)$;
	\item ``Model 3'' is Poisson regression \eqref{Heq:Poi} with $\btheta_0 = d^{-1/2}\bone$, and $\bS_i \sim \Gau(\bzero, \bI)$.
%	\item ``Model 4'' is Poisson regression \eqref{Heq:Poi} with $\btheta_0 = -d^{-1/2}\bone$, and elements of $\bS_i$ are mutually independent following $\exp(1)$.
	\end{itemize}
	
	Figure \ref{Hfg:QQLogGaud2} examines the approximation performance of the proposed methods. The x-axis is the oracle estimator. And the estimators from the four methods are on the y-axis. We have a few observations. 
	When the ratio $nm^{-1/2}$ increases, the distance between our proposed methods and $\btheta^*$ is getting smaller and smaller. By contrast, due to the impact from the local center, the accuracy of CSL drops rapidly when $m$ increases. Also, it is interesting to mention that there exists an ``eternal'' gap between $\ol\btheta$ and the oracle estimator. This stands for the bias of $\ol\btheta$, which depends on $n$ only, and cannot be reduced by adding more data centers. 	
	
	\begin{figure*}[!ht]
        \centering
        \includegraphics[width=.93\textwidth]{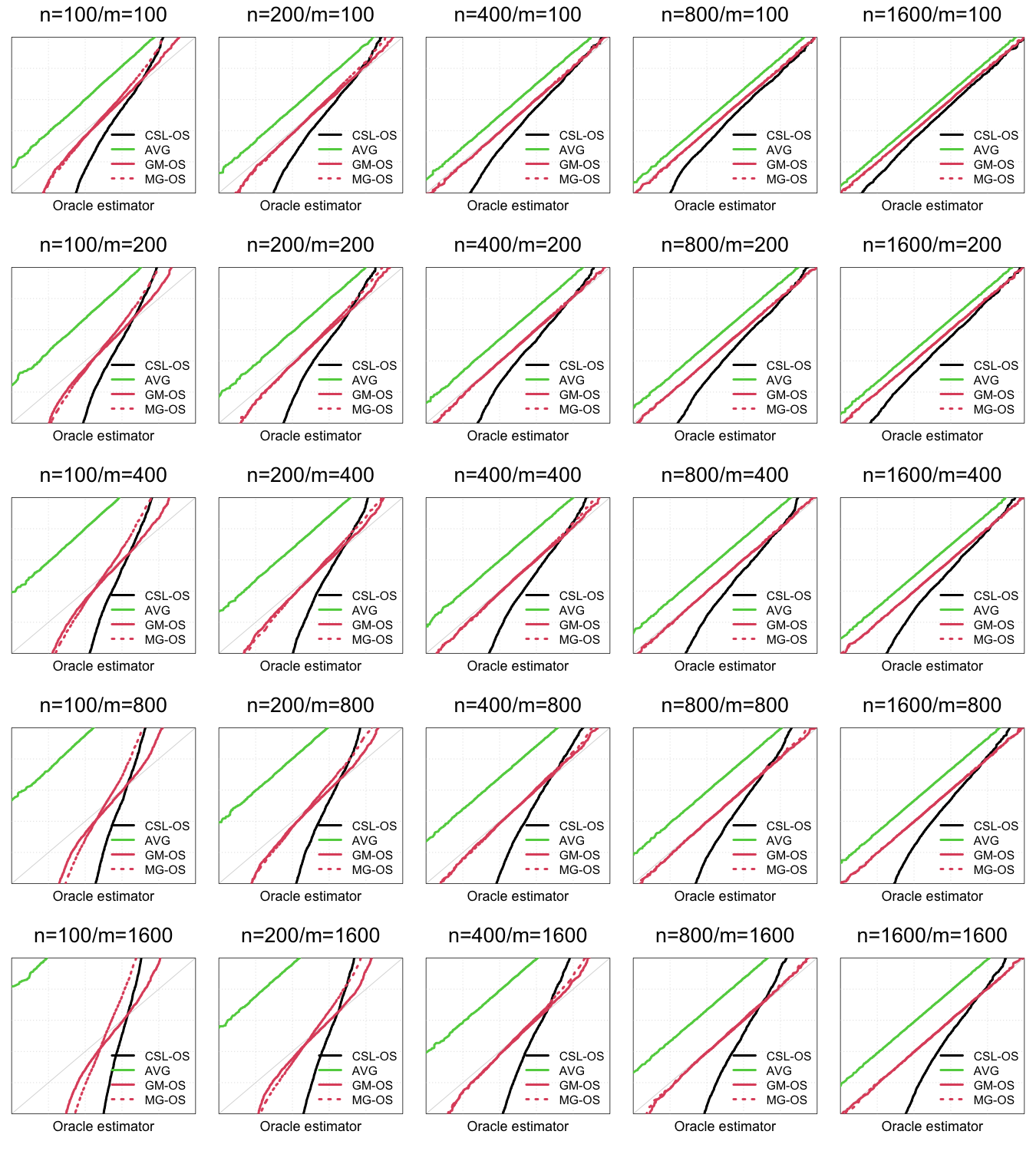} % first figure itself
        \caption{Plots of the one-step estimators versus the oracle estimator. The oracle estimator is on the x-axis. The y-axis includes four estimators: the green line stands for the average of $M$-estimators (AVG), the black line for the one-shot CSL method (CSL-OS), the red solid line for $\wt\btheta_{os}$ (GM-OS), and the red dashed line for $\wt\btheta_{os}'$ (MG-OS). Inference for the first element of $\btheta_0$ in Model 1 is presented. \label{Hfg:QQLogGaud2} }
    \end{figure*}
    
    %\begin{figure*}[!ht]
    %    \centering
    %    \includegraphics[width=1\textwidth]{Fig/QQ_OS_Poi_Exp_d2} % first figure itself
    %    \caption{\label{Hfg:QQPoiExpd2} Poisson regression with $\bS \sim \exp(1)$}
    %\end{figure*}
    
    The 95\% confidence intervals of the first element of $\btheta_0$ are constructed based on $\wt\btheta_{os}$ and $\wt\btheta_{os}'$, and their empirical coverage  is evaluated in Table \ref{Htbl:OSECP}. The normal approximation methods proposed in Corollary \ref{Hcrl:OSEstGau} are applied. These results are consistent with Figure \ref{Hfg:QQLogGaud2}: when $n$ is large and $m = o(n^{2})$, our methods provide close-to-nominal coverage. It is justifiable to request a large $n$, since both proposed methods are based on the large-sample theory, 

	\begin{table*}
        \begin{center}
        \caption{Empirical coverage probabilities for the first element of $\btheta_0$. Methods given in Corollary \ref{Hcrl:OSEstGau} are used to construct 95\% confidence intervals. \label{Htbl:OSECP}}
        \vspace{1em}
        \begin{tabular}{ c c c c c c c c c c c c c c } \hline \multirow{3}{*}{$m$}
            & \multirow{3}{*}{$n$} & \multicolumn{6}{c} {$d = 2$} & \multicolumn{6}{c} {$d = 4$}  \\
            \cline{3-14} & & \multicolumn{2}{c}{Model 1} & \multicolumn{2}{c}{Model 2} & \multicolumn{2}{c} {Model 3} & \multicolumn{2}{c} {Model 1} & \multicolumn{2}{c} {Model 2} & \multicolumn{2}{c} {Model 3}\\ 
            \cline{3-14} & & $\wt\btheta_{os}$ & $\wt\btheta_{os}'$ & $\wt\btheta_{os}$ & $\wt\btheta_{os}'$ & $\wt\btheta_{os}$ & $\wt\btheta_{os}'$ & $\wt\btheta_{os}$ & $\wt\btheta_{os}'$ & $\wt\btheta_{os}$ & $\wt\btheta_{os}'$ & $\wt\btheta_{os}$ & $\wt\btheta_{os}'$ \\ \hline 
            
            \multirow{5}{*}{$100$}
            &  100 & 90 & 86 & 87 & 77 & 83 & 86 & 82 & 76 & 74 & 64 & 73 & 75 \\ 
            &  200 & 93 & 92 & 92 & 88 & 90 & 91 & 89 & 88 & 85 & 82 & 85 & 86 \\ 
            &  400 & 94 & 94 & 94 & 92 & 93 & 93 & 92 & 92 & 91 & 89 & 90 & 91 \\ 
            &  800 & 95 & 94 & 94 & 94 & 94 & 94 & 94 & 93 & 93 & 92 & 93 & 93 \\ 
            & 1600 & 95 & 95 & 95 & 95 & 95 & 95 & 95 & 94 & 94 & 94 & 95 & 94 \\ \hline 
            \multirow{5}{*}{$200$}
            &  100 & 88 & 81 & 83 & 69 & 79 & 83 & 78 & 69 & 71 & 58 & 68 & 71 \\ 
            &  200 & 93 & 90 & 91 & 85 & 87 & 90 & 87 & 85 & 82 & 77 & 79 & 85 \\ 
            &  400 & 94 & 94 & 93 & 90 & 93 & 93 & 92 & 91 & 89 & 87 & 88 & 90 \\ 
            &  800 & 95 & 95 & 95 & 94 & 94 & 94 & 94 & 93 & 93 & 92 & 92 & 93 \\ 
            & 1600 & 94 & 94 & 95 & 94 & 94 & 95 & 94 & 94 & 94 & 94 & 94 & 94 \\ \hline 
            \multirow{5}{*}{$400$}
            &  100 & 84 & 74 & 81 & 60 & 71 & 79 & 73 & 60 & 63 & 47 & 57 & 64 \\ 
            &  200 & 91 & 87 & 89 & 78 & 82 & 88 & 84 & 81 & 78 & 70 & 74 & 82 \\ 
            &  400 & 93 & 92 & 93 & 89 & 90 & 93 & 90 & 91 & 87 & 85 & 85 & 90 \\ 
            &  800 & 95 & 94 & 94 & 93 & 93 & 94 & 94 & 94 & 91 & 90 & 91 & 93 \\ 
            & 1600 & 95 & 95 & 95 & 95 & 95 & 95 & 94 & 94 & 94 & 94 & 93 & 94 \\ \hline 
            \multirow{5}{*}{$800$}
            &  100 & 81 & 65 & 77 & 50 & 66 & 71 & 67 & 50 & 57 & 38 & 50 & 56 \\ 
            &  200 & 88 & 84 & 86 & 71 & 78 & 85 & 79 & 74 & 73 & 61 & 67 & 76 \\ 
            &  400 & 93 & 92 & 92 & 85 & 86 & 91 & 87 & 87 & 84 & 80 & 80 & 87 \\ 
            &  800 & 94 & 94 & 94 & 92 & 92 & 94 & 91 & 92 & 90 & 90 & 89 & 92 \\ 
            & 1600 & 95 & 95 & 95 & 94 & 94 & 94 & 94 & 94 & 93 & 92 & 93 & 94 \\ \hline 
            \multirow{5}{*}{$1600$}
            &  100 & 76 & 53 & 71 & 40 & 58 & 62 & 59 & 38 & 49 & 27 & 42 & 46 \\ 
            &  200 & 86 & 76 & 82 & 61 & 72 & 81 & 74 & 66 & 65 & 50 & 57 & 68 \\ 
            &  400 & 91 & 89 & 89 & 80 & 82 & 89 & 84 & 84 & 79 & 74 & 74 & 83 \\ 
            &  800 & 93 & 93 & 93 & 89 & 89 & 93 & 91 & 91 & 87 & 86 & 86 & 92 \\ 
            & 1600 & 94 & 94 & 95 & 93 & 93 & 94 & 93 & 94 & 92 & 93 & 91 & 93 \\ \hline 
            
            \end{tabular}
        \end{center}
        \normalsize
    \end{table*}
    
\clearpage
\nocite{*}
\bibliographystyle{apalike}
\bibliography{Ref}

\clearpage
\section*{Appendix}
 
\section*{Proof of Proposition \ref{Hpr:DGGD}}
First,  define $\bQ_{12}(\bu) = (\bI \otimes \bu^\T) \bQ_{12}$, $\bQ_{12}(\bu^{\otimes 2}) = \bQ_{12} \circ \bu = (\bI \otimes \bu^\T) \bQ_{12}\bu$,  and $\bQ(\bu^{\otimes 2}) = \bQ \circ \bu =  (\bI \otimes \bu^\T) \bQ\bu$.
Let $\wh\bd_i = \wh\btheta_i - \ol\btheta$ and $\wh\bg_i = \bl_i(\btheta) - \ol\bl(\btheta)$. 
\subsection*{The First Statement}
By definition, 
\begin{align*}
    \bl_i(\btheta_0) & = -\bI_0\wh\bd_i + \bI_0\wh\bd_i + \bl_i(\btheta_0) = -\bI_0\wh\bd_i + \bI_0(\bd_i + \btheta_0 - \bar\btheta) + \bl_i(\btheta_0) \\
        & = -\bI_0\wh\bd_i + \bI_0(\btheta_0 - \bar\btheta) + \lrm{\bI_0 \bd_i + \bl_i(\btheta_0)} \\
    \wh\bg_i & = \bl_i(\btheta_0) + (\wt\bH_i - \wt\bH)\bDelta - \bar\bl(\btheta_0) \\
        & = -\bI_0\wh\bd_i + \bI_0(\btheta_0 - \bar\btheta) + \lrm{\bI_0 \bd_i + \bl_i(\btheta_0)} + (\wt\bH_i - \wt\bH)\bDelta - \bar\bl(\btheta_0).
\end{align*}
Let $\be_i = \bI_0(\btheta_0 - \bar\btheta) + \lrm{\bI_0 \bd_i + \bl_i(\btheta_0)} + (\wt\bH_i - \wt\bH)\bDelta - \bar\bl(\btheta_0)$, then
\begin{equation*}
    \wh\bI_0 - \bI_0 = \lrI{m^{-1}\sum\wh\bd_i\wh\bd_i^\T}^{-1}\lrI{m^{-1}\sum\wh\bd_i\be_i^\T}.
\end{equation*}

Recall the definition $\bd_i = \wh\btheta_i - \btheta_0$ and $\wh\bd_i = \bd_i + \btheta_0 - \ol\btheta$. Note that $\sum \wh\bd_i = \sum \be_i = \bzero$, $m^{-1}\sum\bd_i = \bar\btheta - \btheta_0$, and
\begin{align*}
    m^{-1}\sum\wh\bd_i\be_i^\T & = m^{-1}\sum \lri{\bd_i + \btheta_0 - \ol\btheta}\lrM{\bI_0(\btheta_0 - \bar\btheta) + \lrm{\bI_0 \bd_i + \bl_i(\btheta_0)} + (\wt\bH_i - \wt\bH)\bDelta - \bar\bl(\btheta_0)}^\T \\
        & = m^{-1}\sum \bd_i\lrm{\bl_i(\btheta_0) + \bI_0\bd_i}^\T - \lri{\bar\btheta - \btheta_0}\lri{\bar\btheta - \btheta_0}^\T\bI_0 + \\
        & \blkeq m^{-1}\sum \bd_i \bDelta^\T(\wt\bH_i - \wt\bH) - (\bar\btheta - \btheta_0)\bar\bl(\btheta_0)^\T.
\end{align*}

Lemma \ref{Hlm:Zhang} indicates $\lrn{\bar\bl(\btheta_0)}_2) = \Op(n^{-1/2}m^{-1/2})$, and Lemma \ref{Hlm:dbar} states $\lrn{\bar\btheta - \btheta_0}_2 = \Op\lrm{(nm)^{-1/2} + n^{-1}}$. Also, under Assumptions \ref{as:T0Def}-\ref{Has:LH}, Lemma \ref{Hlm:DDD} states
$\lrN{m^{-1}\sum \bd_i\lrm{\bl_i(\btheta_0) + \bI_0\bd_i}^\T}_2 = \Op(m^{-1/2}n^{-3/2} + n^{-2}).$
Under Assumptions \ref{as:T0Def}-\ref{Has:L3Cont}, Lemma \ref{Hlm:P3s} states     
\begin{equation*}
    m^{-1}\sum_i (\wt\bH_i - \wt\bH)\bDelta\bd_i^\T = -n^{-1} \bQ_{12}^\T \times (\bI_{d\times d} \otimes \bDelta) \times \bI_0^{-1} + \bR_{d11},
\end{equation*}
where $\lrn{\bR_{d11}}_2 = \Op(n^{-1/2})\lrn{\bDelta}_2^3 + \Op\lrm{n^{-1} + (nm)^{-1/2}}\lrn{\bDelta}_2^2 + \Op\lrm{n^{-3/2} + n^{-1}m^{-1/2}}\lrn{\bDelta}_2$.
Putting them together, we have
\begin{align*}
    & \lrN{m^{-1}\sum\wh\bd_i\be_i^\T + n^{-1} \bI_0^{-1}\bQ_{12}(\bDelta)}_2 \\  \leq & \lrN{m^{-1}\sum \bd_i\lrm{\bl_i(\btheta_0) + \bI_0\bd_i}^\T}_2 + \lrN{\lri{\bar\btheta - \btheta_0}\lri{\bar\btheta - \btheta_0}^\T\bI_0}_2 + \\
        & \lrN{m^{-1}\sum \bd_i \bDelta^\T(\wt\bH_i - \wt\bH) + n^{-1} \bI_0^{-1}\bQ_{12}(\bDelta)}_2 + \lrN{(\bar\btheta - \btheta_0)\bar\bl(\btheta_0)^\T}_2 \\
    = & \Op(n^{-1/2})\lrn{\bDelta}_2^3 + \Op\lrm{n^{-1} + (nm)^{-1/2}}\lrn{\bDelta}_2^2 + \Op\lrm{n^{-3/2} + n^{-1}m^{-1/2}}\lrn{\bDelta}_2 + \\
    & \Op(n^{-1}m^{-1} + n^{-2}).
\end{align*}

On the other hand, Lemma \ref{Hlm:NDD} indicates $\lrN{m^{-1}\sum n \wh\bd_i\wh\bd_i^\T - \bI_0^{-1}\bQ_{11}\bI_0^{-1}}_2 = \Op(n^{-1/2} + m^{-1/2}).$ The matrix inverse transformation in continuous. Therefore, by the continuous mapping theorem, 
\begin{align*}
    & \lrN{\lrI{m^{-1}\sum n \wh\bd_i\wh\bd_i^\T}^{-1} - \lri{\bI_0^{-1}\bQ_{11}\bI_0^{-1}}^{-1}}_2 \\
    \leq & \lrN{\lrI{m^{-1}\sum n \wh\bd_i\wh\bd_i^\T}^{-1}}_2\lrN{m^{-1}\sum n \wh\bd_i\wh\bd_i^\T - \bI_0^{-1}\bQ_{11}\bI_0^{-1}}_2  O(1) \\
    = & \Op(n^{-1/2} + m^{-1/2}).
\end{align*}
Note that this part does not depend on $\bDelta$. 
\begin{align*}
        & \wh\bI_0 - \bI_0\\
        = & \lrI{m^{-1}\sum\wh\bd_i\wh\bd_i^\T}^{-1}\lrI{m^{-1}\sum\wh\bd_i\be_i^\T}\\
        = & \lrM{\lrI{m^{-1}\sum n\wh\bd_i\wh\bd_i^\T}^{-1} - \bI_0\bQ_{11}^{-1}\bI_0 + \bI_0\bQ_{11}^{-1}\bI_0}\lrM{m^{-1}\sum n\wh\bd_i\be_i^\T - \bI_0^{-1}\bQ_{12}(\bDelta) + \bI_0^{-1}\bQ_{12}(\bDelta)} \\
        = & \lrM{\lrI{m^{-1}\sum n\wh\bd_i\wh\bd_i^\T}^{-1} - \bI_0\bQ_{11}^{-1}\bI_0}\lrI{m^{-1}\sum n\wh\bd_i\be_i^\T} + \\
        & \lri{\bI_0\bQ_{11}^{-1}\bI_0}\lrM{m^{-1}\sum n\wh\bd_i\be_i^\T - \bI_0^{-1}\bQ_{12}(\bDelta)} + \bI_0\bQ_{11}^{-1}\bQ_{12}(\bDelta).
\end{align*}
Let 
\begin{align*}
    & \bR_{D1} = \lrM{\lrI{m^{-1}\sum n\wh\bd_i\wh\bd_i^\T}^{-1} - \bI_0\bQ_{11}^{-1}\bI_0}\lrI{m^{-1}\sum n\wh\bd_i\be_i^\T} \\
    & \bR_{D2} = \lri{\bI_0\bQ_{11}^{-1}\bI_0}\lrM{m^{-1}\sum n\wh\bd_i\be_i^\T - \bI_0^{-1}\bQ_{12}(\bDelta)},
\end{align*} with
\begin{align*}
        & \lrn{\bR_{D1}}_2 \\
    = & \Op(1 + n^{1/2}m^{-1/2})\lrn{\bDelta}_2^3 + \Op\lrm{n^{-1/2} + n^{1/2}m^{-1}}\lrn{\bDelta}_2^2 + \Op\lri{n^{-1} + m^{-1}}\lrn{\bDelta}_2 + \\
        & \Op(m^{-3/2} + n^{-3/2}) + \Op(n^{-1/2} + m^{-1/2})\lrn{\bQ_{12}(\bDelta)}_2\\
        & \lrn{\bR_{D2}}_2 \\
    = & \Op(n^{1/2})\lrn{\bDelta}_2^3 + \Op\lri{1 + n^{1/2}m^{-1/2}}\lrn{\bDelta}_2^2 + \Op\lri{n^{-1/2} + m^{-1/2}}\lrn{\bDelta}_2 + \\
        & \Op(m^{-1} + n^{-1}).
\end{align*}
Define $\bR_D = \bR_{D1} + \bR_{D2}$, and we complete the proof. 

\subsection*{The Second Statement of Proposition \ref{Hpr:DGGD}}
By the definition, 
\begin{align*}
    \wh\bg_i & = \bl_i(\btheta_0) + (\wt\bH_i - \wt\bH)\bDelta - \bar\bl(\btheta_0) \\
        & = -\bI_0\wh\bd_i + \bI_0(\btheta_0 - \bar\btheta) + \lrm{\bI_0 \bd_i + \bl_i(\btheta_0)} + (\wt\bH_i - \wt\bH)\bDelta - \bar\bl(\btheta_0) \\
    \wh\bd_i & = -\bI_0^{-1}\wh\bg_i + (\btheta_0 - \bar\btheta) + \lrm{\bd_i + \bI_0^{-1}\bl_i(\btheta_0)} + \bI_0^{-1}(\wt\bH_i - \wt\bH)\bDelta - \bI_0^{-1}\bar\bl(\btheta_0).
\end{align*}
Let $\be_i = (\btheta_0 - \bar\btheta) + \lrm{\bd_i + \bI_0\bl_i(\btheta_0)} + \bI_0^{-1}(\wt\bH_i - \wt\bH)\bDelta - \bI_0^{-1}\bar\bl(\btheta_0)$. Then,
\begin{equation*}
    \wh\bOmega - \bI_0^{-1} = \lrI{m^{-1}\sum\wh\bg_i\wh\bg_i^\T}^{-1}\lrI{m^{-1}\sum\wh\bg_i\be_i^\T}.
\end{equation*}
First, Lemma \ref{Hlm:GU} states, under Assumption \ref{as:T0Def}-\ref{Has:LH}
\begin{equation*}
    m^{-1}\sum \wh\bg_i\wh\bg_i^\T = n^{-1}\bQ_{11} + \bR_{\bG},
\end{equation*}
where $\lrN{\bR_{\bG}}_2 = \Op(1)\lrn{\bDelta}_2^4 + \Op(n^{-1/2})\lrn{\bDelta}_2^2 + \Op(n^{-1})\lrn{\bDelta}_2 + \Op(n^{-1}m^{-1/2})$. Therefore, by the continuous mapping theorem, when $\lrn{\bDelta}_2 = \op(n^{-1/4})$, $$\lri{m^{-1}\sum n \wh\bg_i\wh\bg_i^\T}^{-1} \cid \bQ_{11}^{-1}.$$ 

Recall the definition $\bd_i = \wh\btheta_i - \btheta_0$ and $\wh\bd_i = \bd_i + \btheta_0 - \ol\btheta$. Note that $\sum \wh\bg_i = \sum \be_i = \bzero$, $m^{-1}\sum\bd_i = \bar\btheta - \btheta_0$, and
\begin{align*}
    m^{-1}\sum\wh\bg_i\be_i^\T & = m^{-1}\sum \lrm{\bl_i(\btheta_0) + (\wt\bH_i - \wt\bH)\bDelta - \bar\bl(\btheta_0)} \times \\
        & \blkeq \lrO{(\btheta_0 - \bar\btheta) + \lrm{ \bd_i + \bI_0^{-1}\bl_i(\btheta_0)} + \bI_0^{-1}(\wt\bH_i - \wt\bH)\bDelta - \bI_0^{-1}\bar\bl(\btheta_0)}^\T \\
        & = \bar\bl(\btheta_0)(\btheta_0 - \bar\btheta) + m^{-1}\sum \bl_i(\btheta_0)\lrm{\bd_i + \bI_0^{-1}\bl_i(\btheta_0)}^\T + \\
        & \blkeq m^{-1}\sum \bl_i(\btheta_0)\bDelta^\T(\wt\bH_i - \wt\bH)\bI_0^{-1} - \bar\bl(\btheta_0)\bar\bl(\btheta_0)^\T\bI_0^{-1}
        \\
        & \blkeq m^{-1}\sum (\wt\bH_i - \wt\bH)\bDelta\lrm{\bd_i + \bI_0^{-1}\bl_i(\btheta_0)}^\T + \\
        & \blkeq m^{-1}\sum (\wt\bH_i - \wt\bH)\bDelta\bDelta^\T(\wt\bH_i - \wt\bH)\bI_0^{-1}.
\end{align*}

Lemma \ref{Hlm:Zhang} indicates $\lrn{\bar\bl(\btheta_0)}_2) = \Op(n^{-1/2}m^{-1/2})$, and Lemma \ref{Hlm:dbar} states $\lrn{\bar\btheta - \btheta_0}_2 = \Op\lrm{(nm)^{-1/2} + n^{-1}}$. 
Lemma \ref{Hlm:P2} indicates, under Assumptions \ref{as:T0Def}-\ref{Has:L3Cont}, 
        \begin{equation*}
            m^{-1}\sum_i \bl_i(\btheta_0)\bDelta^\T(\wt\bH_i - \wt\bH) 
                = n^{-1}(\bI_{d\times d}\otimes \bDelta^\T) \times\bQ_{12} + \bR_{22},
        \end{equation*}
        where $\lrn{\bR_{22}}_2 = \Op(n^{-1/2})\lrn{\bDelta}_2^3 + \Op\lrm{n^{-1} + (nm)^{-1/2}} \lrn{\bDelta}_2^2 + \Op(n^{-1}m^{-1/2})\lrn{\bDelta}_2$.
        
        Lemma \ref{Hlm:P4} states when Assumptions \ref{as:T0Def}-\ref{Has:LH} hold, 
    $$\lrN{m^{-1}\sum_i (\wt\bH_i - \wt\bH)\bDelta\bDelta^\T(\wt\bH_i - \wt\bH)}_2
        = \Op(1)\lrn{\bDelta}_2^4 + \Op(n^{-1})\lrn{\bDelta}_2^2.$$
        
    Lemma \ref{Hlm:EDD} indicates $\E\lri{\lrN{\bd_i - \bd_{0, i}}_2^{2}}  = O(n^{-2})$. 
    Under Assumptions \ref{as:T0Def}-\ref{Has:LH}, Lemma \ref{Hlm:LDD} states
    \begin{equation*}
        \bl_i(\btheta_0)(\bd_i - \bd_{0, i})^\T = \bW_{i} + \bW_i',
    \end{equation*}
    with $\E\lri{\lrn{\bW_i}_2^2} = O(n^{-3})$, $\lrn{\E\lri{\bW_i}}_2 = O(n^{-2})$, and 
        $\E\lri{\lrn{\bW_i'}_2} = O(n^{-2})$.
    Also, by Lemma \ref{Hlm:R1MC}
    \begin{equation*}
        \E\lrI{\lrN{m^{-1}\sum \bW_i}_2^2} = O(m^{-1}n^{-3} + n^{-4}).
    \end{equation*}
    Hence,
    \begin{align*}
        \lrN{m^{-1}\sum \bl_i(\btheta_0)\lrm{\bl_i(\btheta_0) + \bI_0\bd_i}^\T}_2 & \leq \lrN{m^{-1}\sum \bW_i}_2 + \lrN{m^{-1}\sum \bW_i'}_2 \\
            & = \Op(m^{-1/2}n^{-3/2} + n^{-2}).
    \end{align*}
    
    Also, consider $m^{-1}\sum (\wt\bH_i - \wt\bH)\bDelta\lrm{\bd_i + \bI_0^{-1}\bl_i(\btheta_0)}^\T$. By the Cauchy–Schwarz inequality
    \begin{align*}
        & \lrn{m^{-1}\sum (\wt\bH_i - \wt\bH)\bDelta\lrm{\bd_i + \bI_0^{-1}\bl_i(\btheta_0)}^\T}_2 \\
        \leq & m^{-1}\sum \lrn{(\wt\bH_i - \wt\bH)\bDelta}_2\lrn{\bd_i - \bd_{0, i}}_2 \\
        \leq & \lrM{m^{-1}\sum \lrn{(\wt\bH_i - \wt\bH)\bDelta}_2^2}^{1/2}\lrM{m^{-1}\sum \lrn{\bd_i - \bd_{0, i}}_2^2}^{1/2} \\
        = & \Op(n^{-1})\lrn{\bDelta}_2^2 + \Op(n^{-3/2})\lrn{\bDelta}_2.
    \end{align*}
    The last inequality comes from Lemmas \ref{Hlm:P4} and \ref{Hlm:EDD}. 
    
Putting them together, we have
\begin{align*}
    & \lrN{m^{-1}\sum\wh\bg_i\be_i^\T - n^{-1}\bQ_{12}(\bDelta)\bI_0^{-1}}_2 \\  
    = & \Op(1)\lrn{\bDelta}_2^4 + \Op(n^{-1/2})\lrn{\bDelta}_2^3 + \Op\lrm{n^{-1} + (nm)^{-1/2}}\lrn{\bDelta}_2^2 + \\
    & \Op\lri{n^{-3/2} + n^{-1}m^{-1/2}}\lrn{\bDelta}_2 + \Op(n^{-1}m^{-1} + n^{-2}).
\end{align*}
We assume $\lrn{\bDelta}_2 = \op(n^{-1/4})$, so that $m^{-1}\sum n \wh\bg_i\wh\bg_i^\T \cid \bQ_{11}$, 
 $\lrn{\lri{m^{-1}\sum n \wh\bg_i\wh\bg_i^\T}^{-1} - \bQ_{11}^{-1}} = \op(1)$, and $\lrn{m^{-1}\sum n\wh\bg_i\be_i^\T - \bQ_{12}(\bDelta)\bI_0^{-1}}_2 = \op(1)$. 
\begin{align*}
        & \wh\bOmega - \bI_0^{-1} \\
    = & \lrI{m^{-1}\sum\wh\bg_i\wh\bg_i^\T}^{-1}\lrI{m^{-1}\sum\wh\bg_i\be_i^\T}\\
    = & \lrM{\lrI{m^{-1}\sum n\wh\bg_i\wh\bg_i^\T}^{-1} - \bQ_{11}^{-1} + \bQ_{11}^{-1}}\lrI{m^{-1}\sum n\wh\bg_i\be_i^\T - \bQ_{12}(\bDelta)\bI_0^{-1} + \bQ_{12}(\bDelta)\bI_0^{-1}} \\
    = & \lrM{\lrI{m^{-1}\sum n\wh\bg_i\wh\bg_i^\T}^{-1} - \bQ_{11}^{-1}}\lrI{m^{-1}\sum n\wh\bg_i\be_i^\T} + \\
      & \bQ_{11}^{-1}\lrM{m^{-1}\sum n\wh\bg_i\be_i^\T - \bQ_{12}(\bDelta)\bI_0^{-1}} + \bQ_{11}^{-1}\bQ_{12}(\bDelta)\bI_0^{-1}.
\end{align*}
Let 
\begin{align*}
    & \bR_{G1} = \bQ_{11}^{-1}\lrM{m^{-1}\sum n\wh\bg_i\be_i^\T - \bQ_{12}(\bDelta)\bI_0^{-1}}, \\
    & \bR_{G2} = \lrM{\lrI{m^{-1}\sum n\wh\bg_i\wh\bg_i^\T}^{-1} - \bQ_{11}^{-1}}\lrI{m^{-1}\sum n\wh\bg_i\be_i^\T}, 
\end{align*}
with
\begin{align*}
    \lrn{\bR_{G1}}_2 & = \Op(n)\lrn{\bDelta}_2^4 + \Op(n^{1/2})\lrn{\bDelta}_2^3 + \Op\lri{1 + n^{1/2}m^{-1/2}}\lrn{\bDelta}_2^2 + \\
        & \blkeq \Op\lri{n^{-1/2} + m^{-1/2}}\lrn{\bDelta}_2 + \Op(m^{-1} + n^{-1}), \\
    \lrn{\bR_{G2}}_2 & = \op(\lrn{\bR_{G1}}_2) + \lrm{\Op(n^{1/2})\lrn{\bDelta}_2^2 + \Op(1)\lrn{\bDelta}_2 + \Op(m^{-1/2})}\lrn{\bQ_{12}(\bDelta)}_2.
\end{align*}
Define $\bR_G = \bR_{G1} + \bR_{G2}$, and we complete the proof. 

\section*{Proof of Theorem \ref{Hthm:DistEstOra}}
\subsection*{The first statement}
%We note that $\wh\bI_0$ is a consistent estimator of $\bI_0$, then $\lrn{\wh\bI_0 - \bI_0}_2 = \op(1)$ and $\lrn{\wh\bI_0^{-1} - \bI_0^{-1}}_2 = \op(1)$. 
The first estimator has the decomposition:
\begin{align*}
& \wt\btheta_{t+1} - \btheta_0
         = \wh\bI_0^{-1}(\wh\bI_0 - \bI_0) \bDelta_{t} + \wh\bI_0^{-1}(\bI_0 - \ol\bH^{(g)})\bDelta_{t} - \wh\bI_0^{-1}\bl^{(g)}(\btheta_0) \\
& \wt\btheta_{t+1} - \btheta^* 
         = \wh\bI_0^{-1}(\wh\bI_0 - \bI_0) \bDelta_{t} + \wh\bI_0^{-1}(\bI_0 - \ol\bH^{(g)})\bDelta_{t} + \wh\bI_0^{-1}\lri{\wh\bI_0 - \ol\bH^*}(\ol\bH^*)^{-1}\bl^{(g)}(\btheta_0),
\end{align*}
where $\ol\bH^* = \int_0^1 \sum_{j \in \mathcal{G}} \nabla^2 L\lri{\btheta_0 + t\bDelta^*;\bX_{j}}\diff t$.
Proposition \ref{Hpr:DGGD} states when Assumptions \ref{as:T0Def}-\ref{Has:L3Cont} hold and $\wt\btheta_t \in U(\rho)$, 
$
    \wh\bI_0 - \bI_0 = \bI_0\bQ_{11}^{-1}\bQ_{12}(\bDelta_{t}) + \bR_{D},
$
where $\lrn{\bR_{D}}_2 = \Op(n^{1/2})\lrn{\bDelta_{t}}_2^3 + \Op\lri{1 + n^{1/2}m^{-1/2}}\lrn{\bDelta_{t}}_2^2 + \Op\lri{n^{-1/2} + m^{-1/2}}\lrn{\bDelta_{t}}_2 + \Op(m^{-1} + n^{-1})$. When $\lrn{\bDelta_{t}}_2 = \Op(n^{-1/2})$, $\lrn{\bR_{D}}_2 = \Op(m^{-1} + n^{-1})$, and $\lrn{\wh\bI_0 - \bI_0}_2 = \Op(\lrn{\bDelta_{t}}_2 + m^{-1} + n^{-1})$. Consequently,
\begin{align*}
    \wh\bI_0^{-1}(\wh\bI_0 - \bI_0) \bDelta_{t} & = (\wh\bI_0^{-1} - \bI_0^{-1} + \bI_0^{-1})\lrm{\bI_0\bQ_{11}^{-1}\bQ_{12}(\bDelta_{t}) + \bR_{D}} \bDelta_{t} \\
        & = \bQ_{11}^{-1}\bQ_{12}(\bDelta_{t}^{\otimes 2}) + \wh\bI_0^{-1}(\bI_0 - \wh\bI_0)\bI_0^{-1}\lri{\wh\bI_0 - \bI_0}\bDelta_{t} + \bI_0^{-1}\bR_D\bDelta_{t}.
\end{align*}
Let $\bR_{11} = \wh\bI_0^{-1}(\bI_0 - \wh\bI_0)\bI_0^{-1}\lri{\wh\bI_0 - \bI_0}\bDelta_{t} + \bI_0^{-1}\bR_D\bDelta_{t}$, then $\lrn{\bR_{11}}_2 = \Op(\lrn{\bDelta_{t}}_2^3) + \Op(m^{-1} + n^{-1})\lrn{\bDelta_{t}}_2$.

On the other hand, Lemma \ref{Hlm:IHD} gives $(\bI_0 - \ol\bH^{(g)})\bDelta_{t} = -2^{-1} \bQ(\bDelta_{t}^{\otimes 2}) + \bR_{H}$, 
 where $\lrn{\bR_H} = \Op(n^{-1/2}\lrn{\bDelta_{t}}_2^2 + n^{-1/2}m^{-1/2}\lrn{\bDelta_{t}}_2)$. Hence,
 \begin{align*}
    \wh\bI_0^{-1}(\bI_0 - \ol\bH^{(g)})\bDelta_{t} & = (\wh\bI_0^{-1} - \bI_0^{-1} + \bI_0^{-1})\lrm{-2^{-1} \bQ(\bDelta_{t}^{\otimes 2}) + \bR_{H}} \\
        & = -2^{-1}\bI_0^{-1} \bQ(\bDelta_{t}^{\otimes 2}) + \bI_0^{-1}\bR_{H} + \wh\bI_0^{-1}(\bI_0 - \wh\bI_0)\bI_0^{-1}(\bI_0 - \ol\bH^{(g)})\bDelta_{t}.
        \end{align*}
Let $\bR_{22} = \bI_0^{-1}\bR_{H} + \wh\bI_0^{-1}(\bI_0 - \wh\bI_0)\bI_0^{-1}(\bI_0 - \ol\bH^{(g)})\bDelta_{t}$. We have $\lrn{(\bI_0 - \ol\bH^{(g)})\bDelta_{t}} = \Op(\lrn{\bDelta_{t}}_2^2 + n^{-1/2}m^{-1/2}\lrn{\bDelta_{t}}_2)$ and 
$\lrn{\bR_{22}}_2 = \Op(n^{-1/2}\lrn{\bDelta_{t}}_2^2 + n^{-1/2}m^{-1/2}\lrn{\bDelta_{t}}_2)$.

Consider the last term of $\bDelta_{t+1}$, $-\wh\bI_0^{-1}\bl^{(g)}(\btheta_0)$. 
Note that $\bDelta^* = -(\ol\bH^*)^{-1}\bl^{(g)}(\btheta_0)$, and we state without proof that $\lrn{\bDelta^*}_2 = \Op(n^{-1/2}m^{-1/2})$. Let $\bH^{(g)}_0 = n^{-1}m^{-1}\sum_{j \in \mathcal{G}} \nabla^2 L(\btheta_0;\bX_{j})$. Lemma \ref{Hlm:HHD} and \ref{Hlm:Zhang} indicate
\begin{equation*}
    \lrn{(\ol\bH^* - \bI_0)\bDelta^*}_2 \leq \lrn{\lri{\ol\bH^* - \bH^{(g)}_0}\bDelta^*}_2 + \lrn{(\bH^{(g)}_0 - \bI_0)\bDelta^*}_2 = \Op(n^{-1}m^{-1} + n^{-1/2}m^{-1/2}\lrn{\bDelta^*}_2).
\end{equation*}
Therefore, $\lrn{\bl^{(g)}(\btheta_0)}_2 = \lrn{\ol\bH^*\bDelta^*}_2 = \Op(\lrn{\bDelta^*}_2 + n^{-1}m^{-1})$, and 
\begin{align*}
    \lrn{\wh\bI_0^{-1}\lri{\wh\bI_0 - \ol\bH^*}(\ol\bH^*)^{-1}\bl^{(g)}(\btheta_0)}_2 & \leq  \lrn{\wh\bI_0^{-1}}_2\lrM{\lrn{\wh\bI_0 - \bI_0}_2\lrn{\bDelta^*}_2  + \lrn{(\bI_0 - \ol\bH^*)\bDelta^*}_2} \\
        & = \Op\lrm{(m^{-1} + n^{-1} + \lrn{\bDelta_{t}}_2)\lrn{\bDelta^*}_2 + n^{-1}m^{-1}}.
\end{align*}
Hence, 
\begin{align*}  
    \wt\btheta_{t+1} - \btheta^* & = \bQ_{11}^{-1}\bQ_{12}(\bDelta_{t}^{\otimes 2}) -2^{-1}\bI_0^{-1} \bQ(\bDelta_{t}^{\otimes 2}) + \bR_{11} + \bR_{22} + \wh\bI_0^{-1}\lri{\wh\bI_0 - \ol\bH^*}(\ol\bH^*)^{-1}\bl^{(g)}(\btheta_0) \\
        & = \bQ_{11}^{-1}\bQ_{12}(\bDelta_{t}^{\otimes 2}) -2^{-1}\bI_0^{-1} \bQ(\bDelta_{t}^{\otimes 2}) + \bR^*_1, 
\end{align*}
where $\bR^*_1 = \bR_{11} + \bR_{22} - \wh\bI_0^{-1}\lri{\wh\bI_0 - \ol\bH^*}\bDelta^*$ with
\begin{align*}
    \lrn{\bR^*_1}_2 & \leq \lrn{\bR_{11}}_2 + \lrn{\bR_{12}}_2 + \lrn{\wh\bI_0^{-1}}_2\lrM{\lrn{\wh\bI_0 - \bI_0}_2\lrn{\bDelta^*}_2  + \lrn{(\bI_0 - \ol\bH^*)\bDelta^*}_2} \\
        & = \Op(m^{-1} + n^{-1})(\lrn{\bDelta_{t}}_2 + \lrn{\bDelta^*}_2) + \Op(m^{-1}n^{-1}).
\end{align*}

\subsection*{The second statement}
Similarly, let $\bDelta_{t+1}' =\wt\btheta_{t+1}' - \btheta_0$, and
\begin{align*}
    \wt\btheta_{t+1}' - \btheta_0 & = (\bI - \wh\bOmega\ol\bH^{(g)})\bDelta_{t} - \wh\bOmega\bl^{(g)}(\btheta_0) \\
        & = \bI_0^{-1}(\bI_0 - \ol\bH^{(g)})\bDelta_{t} - (\wh\bOmega - \bI_0^{-1})\ol\bH^{(g)}\bDelta_{t} - \wh\bOmega\bl^{(g)}(\btheta_0),\\
    \wt\btheta_{t+1}' - \btheta^* & = \bI_0^{-1}(\bI_0 - \ol\bH^{(g)})\bDelta_{t} - (\wh\bOmega - \bI_0^{-1})\ol\bH^{(g)}\bDelta_{t} - \lrm{\wh\bOmega - (\ol\bH^*)^{-1}}\bl^{(g)}(\btheta_0).
\end{align*}
For the first component, Lemma \ref{Hlm:IHD} gives $(\bI_0 - \ol\bH^{(g)})\bDelta_{t} = -2^{-1} \bQ(\bDelta_{t}^{\otimes 2}) + \bR_{H}$, 
 where $\lrn{\bR_H}_2 = \Op(n^{-1/2})\lrn{\bDelta_{t}}_2^2 + \Op(n^{-1/2}m^{-1/2})\lrn{\bDelta_{t}}_2$. When $\lrn{\bDelta_{t}}_2 = \Op(n^{-1/2})$, $\lrn{(\bI_0 - \ol\bH^{(g)})\bDelta_{t}} = \Op(1)\lrn{\bDelta_{t}}_2^2 + \Op(n^{-1/2}m^{-1/2})\lrn{\bDelta_{t}}_2$. Therefore,
$$
    \bI_0^{-1}(\bI_0 - \ol\bH^{(g)})\bDelta_{t} = -2^{-1} \bI_0^{-1}\bQ(\bDelta_{t}^{\otimes 2}) + \bI_0^{-1}\bR_{H}.
$$
 
Next, by Proposition \ref{Hpr:DGGD}, $
    \wh\bOmega - \bI_0^{-1} = \bQ_{11}^{-1}\bQ_{12}(\bDelta_{t})\bI_0^{-1} + \bR_{G},
$
where $\lrn{\bR_{G}}_2 = \Op\lri{n^{-1/2} + m^{-1/2}}\lrn{\bDelta_{t}}_2 + \Op(m^{-1} + n^{-1})$. 
\begin{align*}
    (\wh\bOmega - \bI_0^{-1})\ol\bH^{(g)}\bDelta_{t} & = \lrm{\bQ_{11}^{-1}\bQ_{12}(\bDelta_{t})\bI_0^{-1} + \bR_{G}}(\ol\bH^{(g)} - \bI_0 + \bI_0)\bDelta_{t} \\
        & = \bQ_{11}^{-1}\bQ_{12}(\bDelta_{t}^{\otimes 2}) + \bQ_{11}^{-1}\bQ_{12}(\bDelta_{t})\bI_0^{-1}(\ol\bH^{(g)} - \bI_0)\bDelta_{t} + \bR_G\ol\bH^{(g)}\bDelta_{t},
\end{align*}
where
\begin{align*}
    & \lrn{\bQ_{11}^{-1}\bQ_{12}(\bDelta_{t})\bI_0^{-1}(\ol\bH^{(g)} - \bI_0)\bDelta_{t}}_2 = \Op(1)\lrn{\bDelta_{t}}_2^3 + \Op(n^{-1/2}m^{-1/2})\lrn{\bDelta_{t}}_2^2 \\
    & \lrn{\bR_G\ol\bH^{(g)}\bDelta_{t}}_2 = \Op\lri{n^{-1/2} + m^{-1/2}}\lrn{\bDelta_{t}}_2^2 + \Op(m^{-1} + n^{-1})\lrn{\bDelta_{t}}_2.
\end{align*}
Putting them together, we have
\begin{equation*}
    \bDelta_{t+1}' = -2^{-1} \bI_0^{-1}\bQ(\bDelta_{t}^{\otimes 2}) + \bQ_{11}^{-1}\bQ_{12}(\bDelta_{t}^{\otimes 2}) + \bR_1' -\wh\bOmega\bl^{(g)}(\btheta_0),
\end{equation*}
where $\bR_1' = \bI_0^{-1}\bR_H +\bQ_{11}^{-1}\bQ_{12}(\bDelta_{t})\bI_0^{-1}(\ol\bH^{(g)} - \bI_0)\bDelta_{t} + \bR_G\ol\bH^{(g)}\bDelta_{t}$, with $\lrn{\bR'_1}_2 = \Op(n^{-1/2} + m^{-1/2})\lrn{\bDelta_{t}}_2^2 + \Op(n^{-1} + m^{-1})\lrn{\bDelta_{t}}_2$.

Consider $- \lrm{\wh\bOmega - (\ol\bH^*)^{-1}}\bl^{(g)}(\btheta_0)$. Recall that $\lrn{\wh\bOmega - \bI_0^{-1}}_2 = \Op(\lrn{\bDelta_{t}}_2 + m^{-1} + n^{-1})$, and $\lrn{\bl^{(g)}(\btheta_0) + \bI_0\bDelta^*}_2 = \Op(n^{-1}m^{-1} + n^{-1/2}m^{-1/2}\lrn{\bDelta^*}_2)$. 
\begin{align*}
    \lrn{\lrm{\wh\bOmega - (\ol\bH^*)^{-1}}\bl^{(g)}(\btheta_0)}_2 & \leq \lrn{\wh\bOmega - \bI_0^{-1}}_2\lrn{\bl^{(g)}(\btheta_0)}_2 + \lrn{\bI_0^{-1}\lri{\ol\bH^* - \bI_0}\bDelta^*}_2 \\
        & = \Op\lrm{(\lrn{\bDelta_{t}}_2 + m^{-1} + n^{-1})\lrn{\bDelta^*}_2} + \Op(m^{-1}n^{-1}).
\end{align*}
Let $\bR_2^* = \bR'_1 - \lrm{\wh\bOmega - (\ol\bH^*)^{-1}}\bl^{(g)}(\btheta_0)$. Then,
$$\lrn{\bR_2^*} = \Op(n^{-1} + m^{-1})(\lrn{\bDelta_{t}}_2 + \lrn{\bDelta^*}_2) + \Op(m^{-1}n^{-1}).$$

\section*{Proof of Proposition \ref{Hthm:EEst}} 
Recall that $\bQ_{11} = \E\lrM{\bl(\btheta_0;\bX_{i})\bl(\btheta_0;\bX_{i})^\T}$ and $\bQ_{12} = \E\lrO{\bl(\btheta_0;\bX_{i}) \otimes \lrM{\bH(\btheta_0;\bX_{i}) - \bI_0}}$. In this section, since all estimators are based on the local data, for parsimony, we drop the superscript $(l)$. Also, let $\bl = \nabla L$ and $\bH = \nabla^2 L$. 

\subsection*{Estimation of $\bI_0^{-1}\bQ(\bDelta_0^{\otimes 2})$}

Consider $\bQ(\bDelta_0^{\otimes 2})$. By definition and Lemma \ref{Hlm:L3bd} and \ref{Hlm:L3Dec}, 
\begin{align*}
    & \lrn{\lrm{\nabla^3 L(\btheta^A;\bX_i) - \nabla^3 L(\btheta_0;\bX_i)}(\wh\bDelta_0^{\otimes 2})}_2 \leq m(\bX_i)\lrn{\wh\bDelta_0}_2^2\lrn{\bDelta^A}_2 \\
    & \lrM{n^{-1}\sum \nabla^3 L(\btheta_0;\bX_{i}) - \bQ}(\wh\bDelta_0^{\otimes 2}) = 2\bR_2 \\
    & \lrN{\bQ(\wh\bDelta_0^{\otimes 2} - \bDelta_0^{\otimes 2})}_2 \leq d\lambda_h\lrn{\wh\bDelta_0 + \bDelta_0}_2\lrn{\wh\bDelta_0 - \bDelta_0}_2,
\end{align*}
where $\lrn{\bR_2}_2 \leq R_2 \lrn{\wh\bDelta_0}_2^2$ with $\E(R_2^k) = O(n^{-k/2})$. Recall the definition $\wh\bDelta_0 = \wt\btheta_0 - \btheta^A = \bDelta_0 - \bDelta^A$.
\begin{align*}
        & \lrN{\wh\bQ(\wh\bDelta_0^{\otimes 2}) - \bQ(\bDelta_0^{\otimes 2})}_2 \\
        \leq & n^{-1}\sum \lrn{\lrm{\nabla^3 L(\btheta^A;\bX_i) - \nabla^3 L(\btheta_0;\bX_i)}(\wh\bDelta_0^{\otimes 2})}_2 +  \\
        & \lrN{\lrM{n^{-1}\sum \nabla^3 L(\btheta_0;\bX_{i}) - \bQ}(\wh\bDelta_0^{\otimes 2})} +  \lrN{\bQ(\wh\bDelta_0^{\otimes 2} - \bDelta_0^{\otimes 2})}_2 \\
        = & \Op(1)\lrn{\wh\bDelta_0}_2^2\lrn{\bDelta^A}_2 + \Op(n^{-1/2}) \lrn{\wh\bDelta_0}_2^2 + d\lambda_h\lrn{\wh\bDelta_0 - \bDelta_0}_2\lrn{\wh\bDelta_0 + \bDelta_0}_2 \\
        = & \Op(1)\lrn{\bDelta_0 - \bDelta^A}_2^2\lrn{\bDelta^A}_2 + \Op(n^{-1/2}) \lrn{\bDelta_0 - \bDelta^A}_2^2 + d\lambda_h\lrn{\bDelta^A}_2\lrn{2\bDelta_0 - \bDelta^A}_2 \\
        = & \Op(1)\lrn{\bDelta^A}_2^3 + \Op(\lrn{\bDelta_0}_2 + 1)\lrn{\bDelta^A}_2^2 + \Op(1)\lrn{\bDelta_0}_2\lrn{\bDelta^A}_2 + \Op(n^{-1/2})\lrn{\bDelta_0}_2^2.
\end{align*}
When $\lrn{\bDelta^A}_2 = \Op(n^{-1/2})$ and $\lrn{\bDelta_0}_2 = \Op(n^{-1/2})$,
$$
        \lrN{\wh\bQ(\wh\bDelta_0^{\otimes 2}) - \bQ(\bDelta_0^{\otimes 2})}_2 = \Op(1) \lrn{\bDelta^A}_2^2 + \Op(1)\lrn{\bDelta^A}_2\lrn{\bDelta_0}_2 + \Op(n^{-1/2}) \lrn{\bDelta_0}_2^2. 
$$
Consider the local estimator of $\bI_0^{-1}$, $\wh\bH = n^{-1}\sum_i \wh\bH_i $ where $\wh\bH_i = \nabla^2 L(\btheta^A;\bX_i)$.
\begin{align}
    & \lrn{\wh\bH - \bI_0}_2 \leq \lrn{\wh\bH - \bH_0}_2 + \lrn{\bI_0 - \bH_0}_2 = \Op(1)\lrn{\bDelta^A} + \Op(n^{-1/2}) \\
    & \lrn{\wh\bH^{-1} - \bI_0^{-1}}_2 \leq \lrn{\wh\bH^{-1}}_2\lrn{\bI_0^{-1}}_2\lrn{\wh\bH - \bI_0}_2 = \Op(1)\lrn{\bDelta^A} + \Op(n^{-1/2}).
\end{align}

Therefore, 
\begin{align}
        & \lrN{\wh\bH^{-1}\wh\bQ(\wh\bDelta_0^{\otimes 2}) - \bI_0^{-1}\bQ(\bDelta_0^{\otimes 2})}_2 \nonumber \\
    \leq & \lrN{\wh\bH^{-1}\wh\bQ(\wh\bDelta_0^{\otimes 2}) - \bI_0^{-1}\wh\bQ(\wh\bDelta_0^{\otimes 2})}_2 + \lrN{\bI_0^{-1}\wh\bQ(\wh\bDelta_0^{\otimes 2}) - \bI_0^{-1}\bQ(\bDelta_0^{\otimes 2})}_2 \nonumber \\
    = & \Op(1) \lrn{\bDelta^A}_2^2 + \Op(1)\lrn{\bDelta^A}_2\lrn{\bDelta_0}_2 + \Op(n^{-1/2})\lrn{\bDelta_0}_2^2.
\end{align}

\subsection*{Estimation of $\bQ_{11}^{-1}\bQ_{12}(\bDelta_0^{\otimes 2})$}
Recall that $\bl_i$ denotes $\bl_i^{(l)}$ and $\ol\bl$ denotes $\ol\bl^{(l)}$. By definition 
\begin{align*} 
    \wh\bQ_{11} & = n^{-1}\sum \lrm{\bl_i(\btheta^A) - \ol\bl(\btheta^A)}\lrm{\bl_i(\btheta^A) - \ol\bl(\btheta^A)}^\T \\
        & = n^{-1}\sum \lrm{\bl_i(\btheta^A) - \bl_i + \bl_i}\lrm{\bl_i(\btheta^A) - \bl_i + \bl_i}^\T - \ol\bl(\btheta^A)\ol\bl(\btheta^A)^\T \\
        & = n^{-1}\sum \bl_i\bl_i^\T + n^{-1}\sum \lrm{\bl_i(\btheta^A) - \bl_i}\bl_i^\T + n^{-1}\sum \bl_i\lrm{\bl_i(\btheta^A) - \bl_i}^\T + \\
        & \blkeq n^{-1}\sum \lrm{\bl_i(\btheta^A) - \bl_i}\lrm{\bl_i(\btheta^A) - \bl_i}^\T - \ol\bl(\btheta^A)\ol\bl(\btheta^A)^\T.
\end{align*}
By Lemma \ref{Hlm:EZK}, we have
\begin{align*}
    & \E\lrI{\lrN{n^{-1}\sum \bl_i\bl_i^\T - \bQ_{11}}_2^2} = O(n^{-1}) \\
    & \lrN{n^{-1}\sum \bl_i\bl_i^\T - \bQ_{11}}_2 = \Op(n^{-1/2}).
\end{align*}
Let $\ol\bH_i = \int_0^1 \nabla^2 L(\btheta_0 + t\bDelta^A;\bX_i^{(l)})\diff t$, then $\bl_i(\btheta^A) = \bl_i(\btheta_0) + \ol\bH_i\bDelta^A$ and 
\begin{align*} 
    & \lrn{\bl_i(\btheta^A) - \bl_i(\btheta_0)}_2 = \lrn{(\ol\bH_i - \bH_i)\bDelta^A + \bH_i\bDelta^A}_2 \leq h(\bX_i)\lrn{\bDelta^A}_2^2 + \lrN{\bH_i}_2\lrn{\bDelta^A}_2 \\
    & \lrn{\ol\bl(\btheta^A) - \ol\bl(\btheta_0)}_2 \leq n^{-1}\sum \lrn{\bl_i(\btheta^A) - \bl_i(\btheta_0)}_2 = \Op(1)\lrn{\bDelta^A}_2^2 + \Op(1) \lrn{\bDelta^A}_2.
\end{align*}
Therefore,
\begin{align*}
        & \lrN{n^{-1}\sum \lrm{\bl_i(\btheta^A) - \bl_i}\bl_i^\T}_2 \\
    \leq & n^{-1}\sum h(\bX_i)\lrn{\bl_i}_2\lrn{\bDelta^A}_2^2 + n^{-1}\sum \lrn{\bl_i}_2\lrn{\bH_i}_2\lrn{\bDelta^A}_2 \\
    = & \Op(1)\lrn{\bDelta^A}_2^2 + \Op(1) \lrn{\bDelta^A}_2 \\
        & \lrN{n^{-1}\sum \lrm{\bl_i(\btheta^A) - \bl_i}\lrm{\bl_i(\btheta^A) - \bl_i}^\T}_2 \\
    \leq & n^{-1}\sum \lrn{\bl_i(\btheta^A) - \bl_i(\btheta_0)}_2^2 \\
    \leq & n^{-1}\sum 2h(\bX_i)^2\lrn{\bDelta^A}_2^4 + n^{-1}\sum 2\lrN{\bH_i}_2^2\lrn{\bDelta^A}_2^2 \\
    = & \Op(1)\lrn{\bDelta^A}_2^4 + \Op(1)\lrn{\bDelta^A}_2^2 \\
        & \lrn{\ol\bl(\btheta^A)\ol\bl(\btheta^A)^\T}_2 \\
    = & \lrn{\ol\bl(\btheta^A) - \ol\bl + \ol\bl}_2^2 \\
    = & \Op(1)\lrn{\bDelta^A}_2^4 + \Op(1) \lrn{\bDelta^A}_2^2 + \Op(n^{-1}).
\end{align*}
Putting them together, we have
\begin{equation}
    \lrn{\wh\bQ_{11} - \bQ_{11}}_2 = \Op(1)\lrn{\bDelta^A}_2^4 + \Op(1) \lrn{\bDelta^A}_2^2 + \Op(1) \lrn{\bDelta^A}_2 + \Op(n^{-1/2}). \label{Heq:E11Hat}
\end{equation}
Consider $\bQ_{12}$. 
By Lemma \ref{Hlm:lcH}, $\lrN{n^{-1}\sum \bl_i \otimes (\bH_i - \bI_0) - \bQ_{12}} = \Op(n^{-1/2})$. By definition, 
\begin{align*}
\wh\bQ_{12} & = n^{-1}\sum \lrm{\bl_i(\btheta^A) - \bl_i(\btheta_0) + \bl_i(\btheta_0)} \otimes \lrm{\wh\bH_i - \wh\bH} \\
    & = n^{-1}\sum \lrm{\bl_i(\btheta^A) - \bl_i(\btheta_0)} \otimes \lrm{\wh\bH_i - \bI_0 + \bI_0 - \wh\bH} + \\
    & \blkeq n^{-1}\sum \bl_i(\btheta_0) \otimes \lrm{\wh\bH_i - \bH_i + \bH_i - \bI_0 + \bI_0 -  \wh\bH} \\
    & = n^{-1}\sum \lrm{\bl_i(\btheta^A) - \bl_i(\btheta_0)} \otimes \lrm{\wh\bH_i - \bI_0} +  \bar\bl(\btheta^A) \otimes \lrm{\bI_0 - \wh\bH} + \\
    & \blkeq n^{-1}\sum \bl_i(\btheta_0) \otimes \lrm{\wh\bH_i - \bH_i} + n^{-1}\sum \bl_i(\btheta_0) \otimes \lri{\bH_i - \bI_0}.
\end{align*}
Recall that 
\begin{align*}
    & \lrn{\bl_i(\btheta^A) - \bl_i(\btheta_0)}_2 = \lrn{(\ol\bH_i - \bH_i)\bDelta^A + \bH_i\bDelta^A}_2 \leq h(\bX_i)\lrn{\bDelta^A}_2^2 + \lrN{\bH_i}_2\lrn{\bDelta^A}_2 \\
    & \lrn{\wh\bH_i - \bI_0}_2 \leq \lrn{\wh\bH_i - \bH_i}_2 + \lrn{\bH_i - \bI_0}_2 \leq h(\bX_i)\lrn{\bDelta^A}_2 + \lrn{\bH_i - \bI_0}_2. 
\end{align*}
Then,
\begin{align*}
    & \lrN{\bar\bl(\btheta^A) \otimes \lrm{\bI_0 - \wh\bH}} \leq \lrN{\bar\bl(\btheta^A) - \bar\bl(\btheta_0) + \bar\bl(\btheta_0)}_2 \lrn{\bI_0 - \bH + \bH - \wh\bH}_2 \\
        \leq & \lrM{\lrN{\bar\bl(\btheta_0)}_2 + n^{-1}\sum h\lri{\bX_i} \lrn{\bDelta^A}_2^2 + n^{-1}\sum \lrn{\bH_i}_2\lrn{\bDelta^A}_2}\lrM{\lrN{\bI_0 - \bH}_2 + n^{-1}\sum h(\bX_i)\lrn{\bDelta^A}_2}\\
        = & \lrM{\Op(n^{-1/2}) + \Op(1)\lrn{\bDelta^A}_2^2 + \Op(1)\lrn{\bDelta^A}_2}\lrM{\Op(n^{-1/2}) + \Op(1)\lrn{\bDelta^A}_2} \\
        = & \Op(1)\lrn{\bDelta^A}_2^3 + \Op(1)\lrn{\bDelta^A}_2^2 + \Op(n^{-1/2})\lrn{\bDelta^A}_2 + \Op(n^{-1})\\
    & \lrN{n^{-1}\sum \bl_i(\btheta_0) \otimes \lrm{\wh\bH_i - \bH_i}}_2 \leq n^{-1} \sum \lrN{\bl_i(\btheta_0)}_2 h(\bX_i)\lrN{\bDelta^A}_2 = \Op(1)\lrn{\bDelta^A}_2\\
    & \lrN{n^{-1}\sum \lrm{\bl_i(\btheta^A) - \bl_i(\btheta_0)} \otimes \lrm{\wh\bH_i - \bI_0} }_2 \leq n^{-1}\sum \lrN{\bl_i(\btheta^A) - \bl_i(\btheta_0)}_2\lrN{\wh\bH_i - \bI_0}_2 \\
        \leq & n^{-1}\sum \lrM{h(\bX_i)\lrn{\bDelta^A}_2^2 + \lrN{\bH_i}_2\lrn{\bDelta^A}_2}\lrM{h(\bX_i)\lrn{\bDelta^A}_2 + \lrn{\bH_i - \bI_0}_2} \\
        = & n^{-1}\sum h(\bX_i)^2\lrn{\bDelta^A}_2^3 + n^{-1}\sum h(\bX_i)\lri{\lrn{\bH_i - \bI_0}_2 + \lrN{\bH_i}_2}\lrn{\bDelta^A}_2^2 + \\
            & n^{-1}\sum \lrN{\bH_i}_2\lrn{\bH_i - \bI_0}_2\lrn{\bDelta^A}_2 \\
        = & \Op(1) \lrn{\bDelta^A}_2^3 + \Op(1)\lrn{\bDelta^A}_2^2 + \Op(1)\lrn{\bDelta^A}_2
\end{align*}
Therefore,
\begin{align*}
    \lrN{\wh\bQ_{12} - \bQ_{12}}_2 & \leq \lrN{\wh\bQ_{12} - n^{-1}\sum \bl_i \otimes \lri{\bH_i - \bI_0}}_2 + \lrN{n^{-1}\sum \bl_i \otimes \lri{\bH_i - \bI_0} - \bQ_{12}}_2 \\
        & = \Op(1) \lrn{\bDelta^A}_2^3 + \Op(1)\lrn{\bDelta^A}_2^2 + \Op(1)\lrn{\bDelta^A}_2 + \Op(n^{-1/2})
\end{align*}
Let $\bE_{12, j} = \E\lrO{l_j(\btheta_0;\bX_i)\lrM{\bH(\btheta_0;\bX_i) - \bI_0}}$ and $\bQ_{12} = (\bE_{12, 1}, \dots, \bE_{12, d})^\T$. By definition, 
\begin{align*}
    (\bI_{d\times d}\otimes \bu^\T) \times\bQ_{12} \times \bu & = \begin{pmatrix}
        \bu^\T \bE_{12, 1}  \bu \\ \vdots \\
        \bu^\T \bE_{12, d}  \bu
    \end{pmatrix}, \forall \bu \in \mathbb{R}^d.
\end{align*}
Then, $\wh\bQ_{12}(\wh\bDelta_0^{\otimes 2}) - \bQ_{12}(\bDelta_0^{\otimes 2}) = (\wh\bQ_{12} - \bQ_{12})(\wh\bDelta_0^{\otimes 2}) + \bQ_{12}(\wh\bDelta_0^{\otimes 2}) - \bQ_{12}(\bDelta_0^{\otimes 2})$, that is
\begin{align*}
     & (\bI_{d\times d}\otimes \wh\bDelta_0^\T) \times \wh\bQ_{12} \times \wh\bDelta_0 - (\bI_{d\times d}\otimes \bDelta_0^\T) \times \bQ_{12} \times \bDelta_0 \\
    = & (\bI_{d\times d}\otimes \wh\bDelta_0^\T) \times (\wh\bQ_{12} - \bQ_{12}) \times \wh\bDelta_0 +  (\bI_{d\times d}\otimes \wh\bDelta_0^\T) \times \bQ_{12} \times \wh\bDelta_0 - \\
     & (\bI_{d\times d}\otimes \bDelta_0^\T) \times \bQ_{12} \times \bDelta_0.
\end{align*}
Note that 
\begin{align*}
        & \lrn{(\wh\bQ_{12} - \bQ_{12})(\wh\bDelta_0^{\otimes 2})}_2 
    \leq \lrn{\wh\bDelta_0}_2^2\lrn{\wh\bQ_{12} - \bQ_{12}}_2 \\
    = &  \Op(1) \lrn{\bDelta^A}_2^3\lrn{\wh\bDelta_0}_2^2 + \Op(1)\lrn{\bDelta^A}_2^2\lrn{\wh\bDelta_0}_2^2 + \Op(1)\lrn{\bDelta^A}_2\lrn{\wh\bDelta_0}_2^2 + \Op(n^{-1/2})\lrn{\wh\bDelta_0}_2^2 \\
        & \lrn{\bQ_{12}(\wh\bDelta_0^{\otimes 2}) - \bQ_{12}(\bDelta_0^{\otimes 2})}_2 \\
    = & \lrN{\begin{Bmatrix}
        (\bDelta_0 + \wh\bDelta_0)^\T \bE_{12, 1}  (\bDelta_0 - \wh\bDelta_0) \\ \vdots \\
        (\bDelta_0 + \wh\bDelta_0)^\T \bE_{12, d}  (\bDelta_0 - \wh\bDelta_0)
    \end{Bmatrix}}_2 \\
    \leq & O(1)\lrn{\bDelta_0 + \wh\bDelta_0}_2\lrn{\bDelta_0 - \wh\bDelta_0}_2.
\end{align*}
Applying the conditions that $\wh\bDelta_0 = \wt\btheta_0 - \btheta^A = \bDelta_0 - \bDelta^A$, $\lrn{\bDelta^A}_2 = \Op(n^{-1/2})$, and $\lrn{\bDelta_0}_2 = \Op(n^{-1/2})$,
\begin{align*}
    & \lrn{\wh\bQ_{12}(\wh\bDelta_0^{\otimes 2}) - \bQ_{12}(\bDelta_0^{\otimes 2})}_2 \\
    = & \Op(1) \lrn{\bDelta^A}_2^3\lrn{\wh\bDelta_0}_2^2 + \Op(1)\lrn{\bDelta^A}_2^2\lrn{\wh\bDelta_0}_2^2 + \Op(1)\lrn{\bDelta^A}_2\lrn{\wh\bDelta_0}_2^2 + \Op(n^{-1/2})\lrn{\wh\bDelta_0}_2^2 +  \\
    & O(1)\lrn{\bDelta_0 - \wh\bDelta_0}_2\lrn{\bDelta_0 + \wh\bDelta_0}_2  \\
    = & \Op(n^{-1/2})\lrn{\wh\bDelta_0}_2^2 + O(1)\lrn{\bDelta^A}_2\lrn{2\bDelta_0 - \bDelta^A}_2 \\
    = & \Op(1)\lrn{\bDelta^A}_2^2 + O(1)\lrn{\bDelta^A}_2\lrn{\bDelta_0}_2 + \Op(n^{-1/2})\lrn{\bDelta_0}_2^2.
\end{align*}
Also, in this case, \eqref{Heq:E11Hat} indicates $\lrn{\wh\bQ_{11} - \bQ_{11}}_2 = \Op(n^{-1/2})$ and $\lrn{\wh\bQ_{11}^{-1} - \bQ_{11}^{-1}}_2 \leq \lrn{\wh\bQ_{11}^{-1}}_2\lrn{\wh\bQ_{11} - \bQ_{11}}_2\lrn{\bQ_{11}^{-1}}_2 = \Op(n^{-1/2})$. Then
\begin{align*}
    & \lrn{\wh\bQ_{11}^{-1}\wh\bQ_{12}(\wh\bDelta_0^{\otimes 2}) - \bQ_{11}^{-1}\bQ_{12}(\bDelta_0^{\otimes 2})}_2 \\
    \leq & \lrn{\wh\bQ_{11}^{-1}\wh\bQ_{12}(\wh\bDelta_0^{\otimes 2}) - \bQ_{11}^{-1}\wh\bQ_{12}(\wh\bDelta_0^{\otimes 2})}_2 + \lrn{\bQ_{11}^{-1}\wh\bQ_{12}(\wh\bDelta_0^{\otimes 2}) - \bQ_{11}^{-1}\bQ_{12}(\bDelta_0^{\otimes 2})}_2 \\
    \leq & \Op(1)\lrn{\bDelta^A}_2^2 + O(1)\lrn{\bDelta^A}_2\lrn{\bDelta_0}_2 + \Op(n^{-1/2})\lrn{\bDelta_0}_2^2.
\end{align*}

\section*{Proof of Theorem \ref{Hthm:OSEst}}
In either case, Lemma \ref{Hlm:dbar} indicates $\lrn{\bDelta^A}_2 = \Op\lrm{(nm)^{-1/2} + n^{-1}}$. As a result, Proposition \ref{Hthm:EEst} indicates
    \begin{align*}
        & \lrN{\wh\bH^{-1}\wh\bQ(\wh\bDelta_0^{\otimes 2}) - \bI_0^{-1}\bQ(\bDelta_0^{\otimes 2})}_2 = \Op(n^{-1}m^{-1/2} + n^{-3/2}) \\
        & \lrN{\wh\bQ_{11}^{-1}\wh\bQ_{12}(\wh\bDelta_0^{\otimes 2}) - \bQ_{11}^{-1}\bQ_{12}(\bDelta_0^{\otimes 2})}_2 = \Op(n^{-1}m^{-1/2} + n^{-3/2}). 
    \end{align*}
Theorem \ref{Hthm:DistEstOra} indicates
    \begin{equation*}
    \wt\btheta_{1} - \btheta^* = \bQ_{11}^{-1}\bQ_{12}(\bDelta_0^{\otimes 2}) -2^{-1}\bI_0^{-1} \bQ(\bDelta_0^{\otimes 2}) + \bR^*_1,
\end{equation*}
with $\lrn{\bR^*_1}_2 = \Op\lrm{(m^{-1} + n^{-1})(\lrn{\bDelta_0}_2 + \lrn{\bDelta^*}_2) + m^{-1}n^{-1}}$. Also, $\lrn{\bDelta_0}_2 = \Op(n^{-1/2})$ and $\lrn{\bDelta^*} =  \Op\lrm{(nm)^{-1/2} + n^{-1}}$. Then,
    \begin{align*}
        \lrn{\wt\btheta_{os} - \btheta^*} & \leq \lrn{\bR^*_1}_2 + \lrN{\wh\bH^{-1}\wh\bQ(\wh\bDelta_0^{\otimes 2}) - \bI_0^{-1}\bQ(\bDelta_0^{\otimes 2})}_2 + \lrN{\wh\bQ_{11}^{-1}\wh\bQ_{12}(\wh\bDelta_0^{\otimes 2}) - \bQ_{11}^{-1}\bQ_{12}(\bDelta_0^{\otimes 2})}_2 \\
            & = \Op(n^{-3/2} + m^{-1}n^{-1/2}).
    \end{align*}
    
    The second statements can be verified using the similar logic. 
    
\section*{Technical Lemmas}

\subsection*{Lemmas of multi-node}
\begin{lemma} \label{Hlm:dbar} Under Assumptions \ref{as:T0Def}-\ref{Has:LH}, consider the mean of the M-estimators from $m$ centers:
$$\lrn{\bar\btheta - \btheta_0}_2 = \Op\lrm{(nm)^{-1/2}} + O(n^{-1}).$$
\end{lemma}
\begin{proof} Let $\bd_i = \wh\btheta_i - \btheta_0$. 
By Lemma \ref{Hlm:EDD}, $\lrn{\E(\bd_i)}_2 = O(n^{-1})$ and $\E(\lrn{\bd_i}_2^{2}) = O(n^{-1})$.
\begin{align*}
    \E\lrm{\lrn{\bar\bd - \E(\bd_1)}_2^2} & = m^{-2} \sum_i\sum_{i'} \E\lro{\lrm{\bd_i - \E(\bd_1)}^\T\lrm{\bd_{i'} - \E(\bd_1)}} = m^{-1} \E(\lrn{\bd_1 - \E(\bd_1)}_2^2) \\
        & = m^{-1}\E(\lrn{\bd_1}_2^2) +  m^{-1}\lrn{\E(\bd_1)}_2^2 = O\lrm{(nm)^{-1}} \\
    \lrn{\bar\bd}_2 & = \Op\lrm{(nm)^{-1/2}} + O(n^{-1}).
\end{align*}
\end{proof}

\begin{lemma} \label{Hlm:DDD} Under Assumptions \ref{as:T0Def}-\ref{Has:LH}
$$    \lrN{m^{-1}\sum \bd_i\lrm{\bl_i(\btheta_0) + \bI_0\bd_i}^\T}_2 = \Op(m^{-1/2}n^{-3/2} + n^{-2}).
$$
\end{lemma}
\begin{proof}
We have the decomposition:
    \begin{equation*}
        \bd_i\lrm{\bl_i(\btheta_0) + \bI_0\bd_i}^\T = \bd_{0, i}\lri{\bd_i - \bd_{0, i}}^\T\bI_0 + \lri{\bd_i - \bd_{0, i}}\lri{\bd_i - \bd_{0, i}}^\T\bI_0
    \end{equation*}
    Lemma \ref{Hlm:EDD} indicates $\E\lri{\lrN{\bd_i - \bd_{0, i}}_2^{2}}  = O(n^{-2})$. 
    Under Assumptions \ref{as:T0Def}-\ref{Has:LH}, Lemma \ref{Hlm:LDD} states
    \begin{equation*}
        \bl_i(\btheta_0)(\bd_i - \bd_{0, i})^\T = \bW_{i} + \bW_i' 
    \end{equation*}
    with $\E\lri{\lrn{\bW_i}_2^2} = O(n^{-3})$, $\lrn{\E\lri{\bW_i}}_2 = O(n^{-2})$, and 
        $\E\lri{\lrn{\bW_i'}_2} = O(n^{-2})$.
    Also, by Lemma \ref{Hlm:R1MC}
    \begin{equation*}
        \E\lrI{\lrN{m^{-1}\sum \bW_i}_2^2} = O(m^{-1}n^{-3} + n^{-4}).
    \end{equation*}
    Hence
    \begin{align*}
        \lrN{m^{-1}\sum \bd_i\lrm{\bl_i(\btheta_0) + \bI_0\bd_i}^\T}_2 & \leq \lrN{m^{-1}\sum \bW_i}_2 + \lrN{m^{-1}\sum \bW_i'}_2 + \\
            & \blkeq \lrN{m^{-1}\sum\lri{\bd_i - \bd_{0, i}}\lri{\bd_i - \bd_{0, i}}^\T\bI_0}_2 \\
            & = \Op(m^{-1/2}n^{-3/2} + n^{-2}).
    \end{align*}
\end{proof}

\begin{lemma} \label{Hlm:NDD}
    Under Assumptions \ref{as:T0Def}-\ref{Has:LH}, $$\lrN{m^{-1}\sum n \wh\bd_i\wh\bd_i^\T - \bI_0^{-1}\bQ_{11}\bI_0^{-1}}_2 = \Op(n^{-1/2} + m^{-1/2}).$$
\end{lemma}
\begin{proof}
First, we have $m^{-1}\sum \wh\bd_i\wh\bd_i^\T = m^{-1}\sum \bd_i\bd_i^\T - (\bar\btheta - \btheta_0)(\bar\btheta - \btheta_0)^\T$. 
Let $\bd_{0, i} = -\bI_0^{-1}\bl_i(\btheta_0)$, and Lemma \ref{Hlm:EDD} indicates $\E\lri{\lrN{\bd_i - \bd_{0, i}}_2^{2}}  = O(n^{-2})$ under Assumption \ref{Has:LH}. 
\begin{align*}
    & \bd_i\bd_i^\T \\
    = & (\bd_{0, i} + \bd_i - \bd_{0, i})(\bd_{0, i} + \bd_i - \bd_{0, i})^\T \\
     = & \bd_{0, i}\bd_{0, i} + \bd_{0, i}(\bd_i - \bd_{0, i})^\T + (\bd_i - \bd_{0, i})\bd_{0, i}^\T + (\bd_i - \bd_{0, i})(\bd_i - \bd_{0, i})^\T \\
    & \E\lri{\lrn{\bd_i\bd_i^\T - \bd_{0, i}\bd_{0, i}^\T}_2} \\
    \leq & 2 \lrM{\E\lri{\lrn{\bd_i - \bd_{0, i}}_2^2}\E\lri{\lrn{\bd_{0, i}}_2^2}}^{\frac{1}{2}} + \E\lri{\lrn{\bd_i - \bd_{0, i}}_2^2} = O(n^{-\frac{3}{2}}).
\end{align*}
On the other hand, consider $\bd_{0, i}\bd_{0, i}^\T - n^{-1}\bI_0^{-1}\bQ_{11}\bI_0^{-1}$ where $\bQ_{11} = \E\lrM{\bl(\btheta_0;\bX_{ij})\bl(\btheta_0;\bX_{ij})^\T}$. Under Assumption \ref{Has:LH}, Lemma \ref{Hlm:Zhang} indicates $\E\lrm{\lrn{\bl_i(\btheta_0)}_2^4} = O(n^{-2})$, and
\begin{equation*}
    \E\lrI{\lrn{\bd_{0, i}\bd_{0, i}^\T - n^{-1}\bI_0^{-1}\bQ_{11}\bI_0^{-1}}_2^2} = O(n^{-2}).
\end{equation*}
When it comes to $m$ centers, by Lemma \ref{Hlm:EZK},
    \begin{equation*}
        \E\lrI{\lrN{m^{-1}\sum  {\bd_{0, i}\bd_{0, i}^\T - n^{-1}\bI_0^{-1}\bQ_{11}\bI_0^{-1}}}_2^2} = O(n^{-2}m^{-1}).
    \end{equation*}
Also, Lemma \ref{Hlm:dbar} states $\lrn{\bar\btheta - \btheta_0}_2 = \Op\lrm{(nm)^{-1/2}} + O(n^{-1})$. 
All together, we have
\begin{align*}
        & \lrN{m^{-1}\sum\wh\bd_i\wh\bd_i^\T -n^{-1}\bI_0^{-1}\bQ_{11}\bI_0^{-1}}_2 \\
    \leq & \lrN{m^{-1}\sum\bd_i\bd_i^\T -n^{-1}\bI_0^{-1}\bQ_{11}\bI_0^{-1}}_2 + \lrn{\bar\btheta - \btheta_0}_2^2 \\
    \leq & \lrN{m^{-1}\sum\lri{\bd_i\bd_i^\T - \bd_{0, i}\bd_{0, i}^\T}}_2 + \lrN{m^{-1}\sum\bd_{0, i}\bd_{0, i}^\T - n^{-1}\bI_0^{-1}\bQ_{11}\bI_0^{-1}}_2^2 + \lrn{\bar\btheta - \btheta_0}_2^2 \\
    = & \Op(n^{-3/2}) + \Op(n^{-1}m^{-1/2}).
\end{align*}
Hence, $$\lrN{m^{-1}\sum n \wh\bd_i\wh\bd_i^\T - \bI_0^{-1}\bQ_{11}\bI_0^{-1}}_2 = \Op(n^{-1/2} + m^{-1/2}).$$
\end{proof}

\begin{lemma} \label{Hlm:GU} 
Define $\bG_i = \wh\bg_i\wh\bg_i^\T$, $\bU_i = -\wh\bg_i\wh\bd^\T$, $\bG = m^{-1}\sum\bG_i$, and $\bU = m^{-1}\sum\bU_i$. 
Under Assumption \ref{as:T0Def}-\ref{Has:LH}, when $\btheta \in U(\rho)$, 
\begin{equation*}
    \bG = n^{-1}\bQ_{11} + \bR_{\bG},
\end{equation*}
where $\lrN{\bR_{\bG}}_2 = \Op(1)\lrn{\bDelta}_2^4 + \Op(n^{-1/2})\lrn{\bDelta}_2^2 + \Op(n^{-1})\lrn{\bDelta}_2 + \Op(n^{-1}m^{-1/2})$;
\begin{equation*}
    \bU = -n^{-1}\bQ_{11}\bI_0^{-1} + \bR_{\bU}, 
\end{equation*}
where $\lrn{\bR_{\bU}}_2 = \Op(n^{-1/2})\lrn{\bDelta}_2^2 + \Op(n^{-1})\lrn{\bDelta}_2 + \Op(n^{-1}m^{-1/2} + n^{-3/2}).$

Specifically, when $\lrn{\bDelta}_2 = \Op(n^{-1/2})$, 
\begin{align*}
    & \lrN{n\bG - \bQ_{11}}_2 = \Op(\lrn{\bDelta}_2) + \Op(m^{-\frac{1}{2}}), \\
    & \lrN{n\bU + \bQ_{11}\bI_0^{-1}}_2 = \Op(n^{-\frac{1}{2}} + m^{-\frac{1}{2}}). 
\end{align*}
\end{lemma}

\begin{proof} We can expend $\wh\bg_i$ such that $\wh\bg_i = \bl_i(\btheta_0) + \wt\bH_i\bDelta - \lrm{\ol\bl(\btheta_0) + \wt\bH\bDelta}$. For parsimony, let $\bl_i = \bl_i(\btheta_0)$ and $\ol\bl = \ol\bl(\btheta_0)$.
\begin{align*}
    \sum \wh\bg_i\wh\bg_i^\T & = \sum\bl_i\bl_i^\T + \sum\bl_i\bDelta^\T(\wt\bH_i - \wt\bH) - \sum\bl_i\bar\bl^\T + \\
        & \blkeq \sum(\wt\bH_i - \wt\bH)\bDelta\bl_i^\T + \sum(\wt\bH_i - \wt\bH)\bDelta\bDelta^\T(\wt\bH_i - \wt\bH) - \sum(\wt\bH_i - \wt\bH)\bDelta\bar\bl^\T \\ 
        & \blkeq - \sum\bar\bl\wh\bg_i^\T \\
        & = \sum\bl_i\bl_i^\T - m\bar\bl\bar\bl^\T + \sum\bl_i\bDelta^\T(\wt\bH_i - \wt\bH) + \sum(\wt\bH_i - \wt\bH)\bDelta\bl_i^\T + \\
        & \blkeq \sum(\wt\bH_i - \wt\bH)\bDelta\bDelta^\T(\wt\bH_i - \wt\bH).
\end{align*}
Under Assumptions \ref{as:T0Def}-\ref{Has:LH}, by Lemma \ref{Hlm:P2} and \ref{Hlm:P4}, we have 
\begin{align*}
    & \lrN{m^{-1}\sum_i \bl_i(\btheta_0)\bDelta^\T(\wt\bH_i - \wt\bH)}_2 = \Op(n^{-1/2})\lrn{\bDelta}_2^2 + \Op(n^{-1})\lrn{\bDelta}_2 \\
    & \lrN{m^{-1}\sum_i (\wt\bH_i - \wt\bH)\bDelta\bDelta^\T(\wt\bH_i - \wt\bH)}_2
        = \Op(1)\lrn{\bDelta}_2^4 + \Op(n^{-1})\lrn{\bDelta}_2^2.
\end{align*}

Also, $\lrN{\bar\bl}_2^2 = \Op\lrM{(nm)^{-1}}$ and by Lemma \ref{Hlm:EZK}
\begin{align}
    & \E(\lrn{\bl_i}_2^4) = O(n^{-2}) \nonumber \\
    & \E(\bl_i\bl_i^{\T}) = n^{-1}\sum_j\sum_{j'} \E\lrM{\bl(\btheta_0;\bX_{ij})\bl(\btheta_0;\bX_{ij'})^\T} = n^{-1}\bQ_{11} \nonumber \\
    & \E\lrI{\lrN{m^{-1}\sum_i \bl_i\bl_i^{\T} - n^{-1}\bQ_{11}}_2^2} = O(m^{-1})\E(\lrn{\bl_i\bl_i^\T}_2^2) = O(n^{-2}m^{-1}), \label{Heq:ELL}
\end{align}
where $\bQ_{11} = \E\lrM{\bl(\btheta_0;\bX_{11})\bl(\btheta_0;\bX_{11})^\T}$. Therefore, 
\begin{equation}
    \bG = n^{-1}\bQ_{11} + \bR_{\bG}, \label{Heq:G}
\end{equation}
where
\begin{align*}
    \bR_{\bG} & = m^{-1}\sum_i \bl_i\bl_i^{\T} - n^{-1}\bQ_{11} - \bar\bl\bar\bl^\T + m^{-1}\sum\bl_i\bDelta^\T(\wt\bH_i - \wt\bH) +  \\
    & \blkeq m^{-1}\sum(\wt\bH_i - \wt\bH)\bDelta\bl_i^\T + m^{-1}\sum(\wt\bH_i - \wt\bH)\bDelta\bDelta^\T(\wt\bH_i - \wt\bH) \\
    \lrN{\bR_{\bG}}_2 & = \Op(1)\lrn{\bDelta}_2^4 + \Op(n^{-\frac{1}{2}})\lrn{\bDelta}_2^2 + \Op(n^{-1})\lrn{\bDelta}_2 + \Op(n^{-1}m^{-\frac{1}{2}}).
\end{align*}

Similarly, consider $\bU = -m^{-1}\sum \wh\bg_i\wh\bd_i$. 
\begin{align*}
-\bU & = m^{-1} \sum \bg_i(\bd_i - \bar\bd) = m^{-1} \sum \bg_i\bd_i \\
    & = m^{-1} \sum \bl_i\bd_i^\T + m^{-1}\sum (\wt\bH_i - \wt\bH)\bDelta\bd_i^\T - \bar\bl(\btheta_0)\bar\bd^\T \\
    & = m^{-1} \sum \bl_i\bd_{0, i}^\T + m^{-1} \sum \bl_i(\bd_i - \bd_{0, i})^\T + m^{-1}\sum (\wt\bH_i - \wt\bH)\bDelta\bd_i^\T - \bar\bl(\btheta_0)\bar\bd^\T \\
    & = -n^{-1}\bQ_{11}\bI_0^{-1} + \lrI{-m^{-1} \sum \bl_i\bl_{i}^\T + n^{-1}\bQ_{11}}\bI_0^{-1} + \\
    & \blkeq m^{-1} \sum \bl_i(\bd_i - \bd_{0, i})^\T + m^{-1}\sum (\wt\bH_i - \wt\bH)\bDelta\bd_i^\T - \bar\bl(\btheta_0)\bar\bd^\T.
\end{align*}
Recall that by Lemmas \ref{Hlm:dbar}, \ref{Hlm:P3s}, and \ref{Hlm:EDD}
\begin{align*}
    & \lrn{\bar\bd}_2 = \Op\lrm{(nm)^{-1/2}} + O(n^{-1}) \\
    & \lrN{m^{-1}\sum_i (\wt\bH_i - \wt\bH)\bDelta\bd_i^\T}_2 = \Op(n^{-1/2})\lrn{\bDelta}_2^2 + \Op(n^{-1})\lrn{\bDelta}_2 \\
    & \E\lrM{\lrN{\bl_i(\btheta_0)(\bd_i - \bd_{0, i})^\T}_2} \leq \lrO{\E\lrm{\lrN{\bl_i(\btheta_0)}_2^2}\E\lri{\lrN{\bd_i - \bd_{0, i}}_2^2}}^{1/2} = O(n^{-3/2}).
\end{align*}
Together with \eqref{Heq:ELL} that $\lrN{m^{-1}\sum_i \bl_i\bl_i^{\T} - n^{-1}\bQ_{11}}_2= \Op(n^{-1}m^{-1/2})$, 
\begin{equation}
    \bU = -n^{-1}\bQ_{11}\bI_0^{-1} + \bR_{\bU}, \label{Heq:U}
\end{equation}
where 
\begin{align*}
    \bR_{\bU} & = \lrI{n^{-1}\bQ_{11} - m^{-1} \sum \bl_i\bl_{i}^\T}\bI_0^{-1} + m^{-1} \sum \bl_i(\bd_i - \bd_{0, i})^\T + \\
        & \blkeq m^{-1}\sum (\wt\bH_i - \wt\bH)\bDelta\bd_i^\T - \bar\bl(\btheta_0)\bar\bd^\T \\
    \lrn{\bR_{\bU}}_2 & = \Op(n^{-\frac{1}{2}})\lrn{\bDelta}_2^2 + \Op(n^{-1})\lrn{\bDelta}_2 + \Op(n^{-1}m^{-\frac{1}{2}} + n^{-\frac{3}{2}}).
\end{align*}
When $\btheta \in U(\rho)$ and $\lrn{\bDelta} = \Op(n^{-1/2})$, \eqref{Heq:G} and \eqref{Heq:U} can be simplified: 
\begin{align*}
    & \lrN{n\bG - \bQ_{11}}_2 = \Op(1)\lrn{\bDelta}_2 + \Op(m^{-\frac{1}{2}}), \\
    & \lrN{n\bU + \bQ_{11}\bI_0^{-1}}_2 = \Op(n^{-\frac{1}{2}} + m^{-\frac{1}{2}}). 
\end{align*}
\end{proof}        

\begin{lemma} \label{Hlm:P2} Consider $m^{-1}\sum_i\bl_i(\btheta_0)\bDelta^\T(\wt\bH_i - \wt\bH)$ with $\btheta \in U(\rho)$. 
    \begin{enumerate}[label={(\arabic*)}]
        \item If Assumptions \ref{as:T0Def}-\ref{Has:LH} hold, 
            $$\lrN{m^{-1}\sum_i \bl_i(\btheta_0)\bDelta^\T(\wt\bH_i - \wt\bH)}_2 = \Op(n^{-1/2})\lrn{\bDelta}_2^2 + \Op(n^{-1})\lrn{\bDelta}_2;$$
        \item If Assumptions \ref{as:T0Def}-\ref{Has:L3Cont} hold, 
        \begin{equation*}
            m^{-1}\sum_i \bl_i(\btheta_0)\bDelta^\T(\wt\bH_i - \wt\bH) 
                = n^{-1}(\bI_{d\times d}\otimes \bDelta^\T) \times\bQ_{12} + \bR_{22},
        \end{equation*}
        where $\lrn{\bR_{22}}_2 = \Op(n^{-1/2})\lrn{\bDelta}_2^3 + \Op\lrm{n^{-1} + (nm)^{-1/2}} \lrn{\bDelta}_2^2 + \Op(n^{-1}m^{-1/2})\lrn{\bDelta}_2$.
    \end{enumerate}
\end{lemma}
\begin{proof}
First, 
$$    m^{-1}\sum_i \bl_i(\btheta_0)\bDelta^\T(\wt\bH_i - \wt\bH) = m^{-1}\sum_i\bl_i(\btheta_0)\bDelta^\T(\wt\bH_i - \bI_0) + \bar\bl(\btheta_0)\bDelta^\T(\bI_0 - \wt\bH).
$$
By Lemmas \ref{Hlm:Zhang} and \ref{Hlm:L3Dec}, for $k \in [1, 4]$
\begin{align}
    {\lrn{\bar\bl(\btheta_0)}_2^k} & = \Op\lrm{(nm)^{-k/2}} \nonumber \\
    \lrn{\wt\bH - \bI_0}_2^k & \leq 2^{k-1}\lrM{(nm)^{-1}\sum\sum h(\bX_{ij})}^k\lrn{\bDelta}_2^k + 2^{k-1}\lrn{\bH - \bI_0}_2^k \nonumber \\
        & = \Op(1)\lrn{\bDelta}_2^k + \Op\lrm{(nm)^{-k/2}} \label{Heq:HTI} \\
    \lrn{\bar\bl(\btheta_0)\bDelta^\T(\bI_0 - \wt\bH)}_2^k & = \Op\lrm{(nm)^{-k/2}}\lrn{\bDelta}_2^{2k} + \Op\lrm{(nm)^{-k}}\lrn{\bDelta}_2^k. \label{Heq:LDelIH}
\end{align}
\subsubsection*{The first statement}
On the other hand, under Assumptions \ref{as:T0Def}-\ref{Has:LH}, for ${k'} \in [1, 2]$, at the $i$th center
\begin{align*}
    \lrn{\bl_i(\btheta_0)\bDelta^\T(\wt\bH_i - \bI_0)}_2^{k'} & \leq 2^{{k'}-1}\lrn{\bl_i(\btheta_0)\bDelta^\T(\wt\bH_i - \bH_i)}_2^{k'} + 2^{{k'}-1}\lrn{\bl_i(\btheta_0)\bDelta^\T(\bH_i - \bI_0)}_2^{k'} \\
        & \leq 2^{{k'}-1}\lrn{\bl_i(\btheta_0)}_2^{k'}\lrM{n^{-1}\sum_j h(\bX_{ij})}^{k'}\lrn{\bDelta}_2^{2k} + \\
        & \blkeq 2^{{k'}-1}\lrn{\bl_i(\btheta_0)}_2^{k'}\lrn{\bH_i - \bI_0}_2^{k'}\lrn{\bDelta}_2^{k'}.
\end{align*}
By Jensen's inequality and Holder's inequality, 
\begin{align*}
& \E\lrO{\lrn{\bl_i(\btheta_0)}_2^{k'}{\lrM{n^{-1}\sum_j h(\bX_{ij})}^{k'}}} \leq \lrO{\E\lrM{\lrn{\bl_i(\btheta_0)}_2^{2k}}\E\lrm{{h(\bX_{ij})}^{2k}}}^{1/2} = O(n^{-{k'}/2}) \\
& \E\lri{\lrn{\bl_i(\btheta_0)}_2^{k'}\lrn{\bH_i - \bI_0}_2^{k'}} \leq \lrM{\E\lri{\lrn{\bl_i(\btheta_0)}_2^{2k}}\E\lri{\lrn{\bH_i - \bI_0}_2^{2k}}}^{1/2} = O(n^{-{k'}}).
\end{align*}
Let ${k'}=1$, 
\begin{align}    
    \lrn{\bar\bl(\btheta_0)\bDelta^\T(\bI_0 - \wt\bH)}_2 & = \Op\lrm{(nm)^{-1/2}}\lrn{\bDelta}_2^{2} + \Op\lrm{(nm)^{-1}}\lrn{\bDelta}_2 \label{Heq:LDII} \\
    \lrN{m^{-1}\sum_i\bl_i(\btheta_0)\bDelta^\T(\wt\bH_i - \bI_0)}_2 & \leq m^{-1}\sum_i\lrn{\bl_i(\btheta_0)}_2\lrM{n^{-1}\sum_j h(\bX_{ij})}\lrn{\bDelta}_2^{2} + \nonumber \\
        & \blkeq m^{-1}\sum_i\lrn{\bl_i(\btheta_0)}_2\lrn{\bH_i - \bI_0}_2\lrn{\bDelta}_{2} \nonumber \\
        & = \Op(n^{-1/2})\lrn{\bDelta}_2^2 + \Op(n^{-1})\lrn{\bDelta}_2.
\end{align}
Hence,
\begin{align}
    \lrN{m^{-1}\sum_i \bl_i(\btheta_0)\bDelta^\T(\wt\bH_i - \wt\bH)}_2 & \leq \lrN{m^{-1}\sum_i\bl_i(\btheta_0)\bDelta^\T(\wt\bH_i - \bI_0)}_2 + \lrn{\bar\bl(\btheta_0)\bDelta^\T(\bI_0 - \wt\bH)}_2 \nonumber \\ 
        & = \Op(n^{-1/2})\lrn{\bDelta}_2^2 + \Op(n^{-1})\lrn{\bDelta}_2. 
\end{align}
\subsubsection*{The second statement}
Note that the major contribution in this term comes from $m^{-1}\sum_i\bl_i(\btheta_0)\bDelta^\T(\wt\bH_i - \bI_0)$. When Assumption \ref{Has:L3Cont} holds, by Lemma \ref{Hlm:L3Dec}, for ${k'} \in [1, 2]$,
$$
    (\wt\bH_i - \bH_i)\bDelta = 2^{-1}\bQ(\bDelta^{\otimes 2}) + \bR_{1i} + \bR_{2i},
$$
where $\lrn{\bR_{1i}}_2 \leq (6n)^{-1}\sum_j m(\bX_{ij})\lrn{\bDelta}_2^3$ and $\lrn{\bR_{2i}}_2 \leq R_{2i} \lrn{\bDelta}_2^2$ with $\E(R_{2i}^{k'}) = O(n^{-{k'}/2})$. Then,
\begin{align*} 
    & m^{-1}\sum_i\bl_i(\btheta_0)\bDelta^\T(\wt\bH_i - \bH_i) = \bar\bl(\btheta_0)\lro{2^{-1}\bQ(\bDelta^{\otimes 2})}^\T + m^{-1}\sum_i\bl_i(\bR_{1i} + \bR_{2i})^\T \\
    & \lrn{\bl_i(\bR_{1i} + \bR_{2i})^\T}_2^{{k'}} \leq 2^{{k'}-1}\lrn{\bl_i}_2^{k'}\lrM{(6n)^{-1}\sum_j m(\bX_{ij})}^{k'}\lrn{\bDelta}_2^{3{k'}} + 2^{{k'}-1}\lrn{\bl_i}_2^{k'} R_{2i}^{k'}\lrn{\bDelta}_2^{2{k'}}.
\end{align*}
Additionally, by definition
$$\lrN{\bQ(\bDelta^{\otimes 2})}_2 \leq \lambda_3 \lrn{\bDelta}_2^2.$$
Let ${k'}=1$, then $m^{-1}\sum\lrn{\bl_i(\bR_{1i} + \bR_{2i})^\T}_2 = \Op(n^{-1/2})\lrn{\bDelta}_2^3 + \Op(n^{-1})\lrn{\bDelta}_2^2$. 
\begin{align}
    \lrN{m^{-1}\sum_i\bl_i(\btheta_0)\bDelta^\T(\wt\bH_i - \bH_i)}_2 = \Op(n^{-1/2})\lrn{\bDelta}_2^3 + \Op\lrm{n^{-1} + (nm)^{-1/2}} \lrn{\bDelta}_2^2 \label{Heq:P2L31}.
\end{align}    
    
On the other hand, consider $m^{-1}\sum\bl_i(\btheta_0)\bDelta^\T(\bH_i - \bI_0)$. By Lemma \ref{Hlm:lcH}, with the existence of the fourth moments of $\bl(\btheta_0;\bX_{ij})$ and $\bH(\btheta_0;\bX_{ij})$, 
$$
m^{-1}\sum_i \bl_i\bDelta^\T(\bH_i - \bI_0) = (\bI_{d\times d}\otimes \bDelta^\T) \times \lrm{\bl_i \otimes (\bH_i - \bI_0)} = (\bI_{d\times d}\otimes \bDelta^\T) \times \lrI{n^{-1}\bQ_{12} + \bR_{23}}, 
$$
where $\E\lri{\lrn{\bR_{23}}_2^{2}} = O(m^{-1}n^{-2})$. 
\begin{equation}
    \lrN{n^{-1}(\bI_{d\times d}\otimes \bDelta^\T) \times\bR_{23}} \leq \lrn{\bI_{d\times d}\otimes \bDelta^\T}_2\lrN{\bR_{23}}_2 = \lrn{\bDelta}_2\Op(m^{-1/2}n^{-1}). \label{Heq:P2L32}
\end{equation}
Combining \eqref{Heq:LDII}, \eqref{Heq:P2L31}, and \eqref{Heq:P2L32} leads to
\begin{align*}
    m^{-1}\sum_i \bl_i(\btheta_0)\bDelta^\T(\wt\bH_i - \wt\bH) 
        & = m^{-1}\sum_i\bl_i(\btheta_0)\bDelta^\T(\wt\bH_i - \bH_i) + m^{-1}\sum_i \bl_i(\btheta_0)\bDelta^\T(\bH_i - \bI_0) + \\
            & \blkeq \bar\bl(\btheta_0)\bDelta^\T(\bI_0 - \wt\bH) \\
        & = n^{-1}(\bI_{d\times d}\otimes \bDelta^\T) \times\bQ_{12} + \bR_{22},
\end{align*}
where $\lrn{\bR_{22}}_2 = \Op(n^{-1/2})\lrn{\bDelta}_2^3 + \Op\lrm{n^{-1} + (nm)^{-1/2}} \lrn{\bDelta}_2^2 + \Op(n^{-1}m^{-1/2})\lrn{\bDelta}_2$. 
\end{proof}

\begin{lemma} \label{Hlm:P3s} Consider $m^{-1}\sum(\wt\bH_i - \wt\bH)\bDelta\bd_i^\T$ with $\btheta \in U(\rho)$. 
    \begin{enumerate}[label={(\arabic*)}]
        \item If Assumptions \ref{as:T0Def}-\ref{Has:LH} hold, 
            $$\lrN{m^{-1}\sum_i (\wt\bH_i - \wt\bH)\bDelta\bd_i^\T}_2 = \Op(n^{-1/2})\lrn{\bDelta}_2^2 + \Op(n^{-1})\lrn{\bDelta}_2.$$
        \item If Assumptions \ref{as:T0Def}-\ref{Has:L3Cont} hold, 
        \begin{equation*}
            m^{-1}\sum_i (\wt\bH_i - \wt\bH)\bDelta\bd_i^\T = -n^{-1} \bQ_{12}^\T \times (\bI_{d\times d} \otimes \bDelta) \times \bI_0^{-1} - \bR_{22}^\T \bI_0^{-1} + \bR_{42},
        \end{equation*}
        where $\lrn{\bR_{22}}_2 = \Op(n^{-1/2})\lrn{\bDelta}_2^3 + \Op\lrm{n^{-1} + (nm)^{-1/2}} \lrn{\bDelta}_2^2 + \Op(n^{-1}m^{-1/2})\lrn{\bDelta}_2$, and $\lrn{\bR_{42}}_2 = \Op\lrm{n^{-1} + (nm)^{-1/2}}\lrn{\bDelta}_2^2 + \Op\lrm{n^{-3/2} + (nm)^{-1}}\lrn{\bDelta}_2$.
    \end{enumerate}
\end{lemma}

\begin{proof}
First $$    m^{-1}\sum_i (\wt\bH_i - \wt\bH)\bDelta\bd_i^\T = m^{-1}\sum_i (\wt\bH_i - \bI_0)\bDelta\bd_i^\T + (\bI_0 - \wt\bH)\bDelta\bar\bd^\T.$$
Consider $(\bI_0 - \wt\bH)\bDelta\bar\bd^\T$. By \eqref{Heq:HTI}, $\lrn{\wt\bH - \bI_0}_2^k = \Op(1)\lrn{\bDelta}_2^k + \Op\lrm{(nm)^{-k/2}}$. By Lemma \ref{Hlm:EDD} and \ref{Hlm:dbar},  $\lrn{\E(\bd_i)}_2 = O(n^{-1})$, $\E(\lrn{\bd_i}_2^{k}) = O(n^{-k/2})$, and $\lrn{\bar\bd}_2 = \Op\lrm{(nm)^{-1/2}} + O(n^{-1})$.
Hence,
\begin{equation}
    \lrn{(\bI_0 - \wt\bH)\bDelta\bar\bd^\T}_2 = \Op\lrm{(nm)^{-1/2} + n^{-1}}\lrn{\bDelta}_2^2 + \Op\lrm{n^{-3/2}m^{-1/2} + (nm)^{-1}}\lrn{\bDelta}_2 \label{Heq:P4L21}.
\end{equation}
\subsubsection*{The first statement}
\begin{align*}
    \lrn{(\wt\bH_i - \bI_0)\bDelta\bd_i^\T}_2^k & \leq 2^{k-1}\lrn{(\wt\bH_i - \bH_i)\bDelta\bd_i^\T}_2^k + 2^{k-1}\lrn{(\bH_i - \bI_0)\bDelta\bd_i^\T}_2^k \\
        & \leq 2^{k-1}\lrn{\bd_i}_2^k\lrM{n^{-1}\sum_j h(\bX_{ij})}^k\lrn{\bDelta}_2^{2k} + 2^{k-1}\lrn{\bd_i}_2^k\lrn{\bH_i - \bI_0}_2^k\lrn{\bDelta}_2^k.
\end{align*}
For $k \in [1, 2]$, using Holder's inequality and Jensen's inequality, at the $i$th node
\begin{align*}
    & \E\lrO{\lrn{\bd_i}_2^k\lrM{n^{-1}\sum_j h(\bX_{ij})}^k} \leq \lrO{\E(\lrn{\bd_i}_2^{2k})\E\lrM{ h(\bX_{ij})^{2k}}}^{1/2} = O(n^{-k/2}) \\
    & \E(\lrn{\bd_i}_2^k\lrn{\bH_i - \bI_0}_2^k ) \leq \lrm{\E(\lrn{\bd_i}_2^{2k})\E\lri{\lrn{\bH_i - \bI_0}_2^{2k}}}^{1/2} = O(n^{-k}).
\end{align*}
Let $k=1$. Then,
\begin{align}
    \lrN{m^{-1}\sum_i (\wt\bH_i - \bI_0)\bDelta\bd_i^\T}_2 & \leq m^{-1}\sum \lrn{\bd_i}_2\lrM{n^{-1}\sum_j h(\bX_{ij})}\lrn{\bDelta}_2^{2} + \nonumber \\
        & \blkeq m^{-1}\sum_i \lrn{\bd_i}_2\lrn{\bH_i - \bI_0}_2\lrn{\bDelta}_2 \nonumber \\
        & = \Op(n^{-1/2})\lrn{\bDelta}_2^2 + \Op(n^{-1})\lrn{\bDelta}_2 \nonumber  \\
    \lrN{m^{-1}\sum_i (\wt\bH_i - \wt\bH)\bDelta\bd_i^\T}_2 & \leq \lrN{m^{-1}\sum_i (\wt\bH_i - \bI_0)\bDelta\bd_i^\T}_2 + \lrn{(\bI_0 - \wt\bH)\bDelta\bar\bd^\T}_2 \nonumber \\
        & = \Op(n^{-1/2})\lrn{\bDelta}_2^2 + \Op(n^{-1})\lrn{\bDelta}_2.
\end{align}

\subsubsection*{The second statement}
Define $\bd_{0, i} = -\bI_0^{-1}\bl_i(\btheta_0)$, and
\begin{align*}
m^{-1}\sum_i (\wt\bH_i - \wt\bH)\bDelta\bd_i^\T & =  m^{-1}\sum (\wt\bH_i - \wt\bH)\bDelta\bd_{0, i}^\T + m^{-1}\sum(\wt\bH - \bI_0)\bDelta\bd_{0, i}^\T \nonumber \\
        & \blkeq + m^{-1}\sum_i (\wt\bH_i - \bI_0)\bDelta(\bd_i - \bd_{0, i})^\T + (\bI_0 - \wt\bH)\bDelta\bar\bd^\T.
\end{align*}

Consider $m^{-1}\sum (\wt\bH_i - \wt\bH)\bDelta\bd_{0, i}^\T$. By the second part of Lemma \ref{Hlm:P2}, 
\begin{align}
    m^{-1}\sum (\wt\bH_i - \wt\bH)\bDelta\bd_{0, i}^\T & = -m^{-1}\sum (\wt\bH_i - \wt\bH)\bDelta\bl_{i}(\btheta_0)^\T\bI_0^{-1} \nonumber \\  
        & = -n^{-1} \bQ_{12}^\T \times (\bI_{d\times d} \otimes \bDelta) \times \bI_0^{-1} - \bR_{22}^\T \bI_0^{-1}, \label{Heq:P4L34}
\end{align}
    where $\lrn{\bR_{22}}_2 = \Op(n^{-1/2})\lrn{\bDelta}_2^3 + \Op\lrm{n^{-1} + (nm)^{-1/2}} \lrn{\bDelta}_2^2 + \Op(n^{-1}m^{-1/2})\lrn{\bDelta}_2$.

Consider $m^{-1}\sum(\wt\bH - \bI_0)\bDelta\bd_{0, i}^\T$. Note that by definition, $m^{-1}\sum(\wt\bH - \bI_0)\bDelta\bd_{0, i}^\T = -(\wt\bH - \bI_0)\bDelta\bar\bl(\btheta_0)^\T\bI_0^{-1}$. By \eqref{Heq:HTI}, 
\begin{align}   
    %m^{-1}\sum (\wt\bH_i - \bI_0)\bDelta\bd_{0, i}^\T & =  m^{-1}\sum (\wt\bH_i - \wt\bH)\bDelta\bd_{0, i}^\T - (\wt\bH - \bI_0)\bDelta\bar\bl(\btheta_0)^\T\bI_0^{-1} \nonumber \\
    \lrN{(\wt\bH - \bI_0)\bDelta\bar\bl(\btheta_0)^\T\bI_0^{-1}}_2 & \leq  \lambda_-^{-1}\lrN{\wt\bH - \bI_0}_2\lrN{\bar\bl(\btheta_0)}_2\lrN{\bDelta}_2 \nonumber \\
        & = \Op\lrm{(nm)^{-1/2}}\lrn{\bDelta}_2^2 + \Op\lrm{(nm)^{-1}}\lrn{\bDelta}_2. \label{Heq:HIDLI}
\end{align}

Consider $m^{-1}\sum(\wt\bH_i - \bI_0)\bDelta(\bd_i - \bd_{0, i})^\T$.  
\begin{align*}
    %(\wt\bH_i - \bI_0)\bDelta\bd_i^\T & = (\wt\bH_i - \bI_0)\bDelta\bd_{0, i}^\T + (\wt\bH_i - \bI_0)\bDelta(\bd_i - \bd_{0, i})^\T \\
        & \bd_i - \bd_{0, i} \\
    = & (-\ol\bH_i^{-1} + \bI_0^{-1})\bl_i(\btheta_0) = \bI_0^{-1}(\ol\bH_i - \bI_0)\ol\bH_i^{-1}\bl_i(\btheta_0) = \bI_0^{-1}(\bI_0 - \ol\bH_i)\bd_i\\
        & (\wt\bH_i - \bI_0)\bDelta(\bd_i - \bd_{0, i})^\T \\
    = & (\wt\bH_i - \bH_i + \bH_i -  \bI_0)\bDelta\bd_i^\T(\bI_0 - \bH_i + \bH_i - \ol\bH_i)\bI_0^{-1} \\
    = & (\wt\bH_i - \bH_i)\bDelta\bd_i^\T(\bI_0 - \bH_i)\bI_0^{-1} + (\wt\bH_i - \bH_i)\bDelta\bd_i^\T(\bH_i - \ol\bH_i)\bI_0^{-1} \\
        & (\bH_i -  \bI_0)\bDelta\bd_i^\T(\bI_0 - \bH_i)\bI_0^{-1} + (\bH_i -  \bI_0)\bDelta\bd_i^\T(\bH_i - \ol\bH_i)\bI_0^{-1}.
\end{align*}
Consider each component, under Assumption \ref{Has:LH},
\begin{align*}
    & \lrN{(\wt\bH_i - \bH_i)\bDelta\bd_i^\T(\bI_0 - \bH_i)\bI_0^{-1}}_2 \leq 
        \lrM{n^{-1}\sum_j h(\bX_{ij})}\lrn{\bd_i}_2\lrn{\bH_i - \bI_0}_2\lambda_-^{-1}\lrn{\bDelta}_2^2 \\
    & \lrN{(\wt\bH_i - \bH_i)\bDelta\bd_i^\T(\bH_i - \ol\bH_i)\bI_0^{-1}}_2 \leq 
        \lrM{n^{-1}\sum_j h(\bX_{ij})}\lrn{\bd_i^\T(\ol\bH_i - \bH_i)}_2\lambda_-^{-1}\lrn{\bDelta}_2^2 \\
    & \lrN{(\bH_i - \bI_0)\bDelta\bd_i^\T(\bI_0 - \bH_i)\bI_0^{-1}}_2 \leq 
        \lrn{\bH_i - \bI_0}_2^2\lrn{\bd_i}_2\lambda_-^{-1}\lrn{\bDelta}_2 \\
    & \lrN{(\bH_i - \bI_0)\bDelta\bd_i^\T(\bH_i - \ol\bH_i)\bI_0^{-1}}_2 \leq 
        \lrn{\bH_i - \bI_0}_2\lrn{\bd_i^\T(\ol\bH_i - \bH_i)}_2\lambda_-^{-1}\lrn{\bDelta}_2.
\end{align*}
By Lemma \ref{Hlm:Zhang} and Lemma \ref{Hlm:HHD}, 
\begin{align*}
        & \E\lrO{\lrM{n^{-1}\sum_j h(\bX_{ij})}\lrn{\bd_i}_2\lrn{\bH_i - \bI_0}_2} \\
    \leq & \lrI{\E\lrO{\lrM{n^{-1}\sum_j h(\bX_{ij})}^2}}^{1/2}\lrM{\E\lri{\lrn{\bd_i}_2^4}\E\lri{\lrn{\bH_i - \bI_0}_2^4}}^{1/4} = O(n^{-1}) \\
        & \E\lrO{\lrM{n^{-1}\sum_j h(\bX_{ij})}\lrn{\bd_i^\T(\ol\bH_i - \bH_i)}_2} \\
    \leq & \lrI{\E\lrO{\lrM{n^{-1}\sum_j h(\bX_{ij})}^2} \E\lrM{\lrn{\bd_i^\T(\ol\bH_i - \bH_i)}_2^2}}^{1/2} = O(n^{-1}) \\
        &\E\lrM{\lrn{\bH_i - \bI_0}_2^2\lrn{\bd_i}_2} \\
    \leq & \lrM{\E\lri{\lrn{\bd_i}_2^2}\E\lri{\lrn{\bH_i - \bI_0}_2^4}}^{1/2} = O(n^{-3/2}) \\
        & \E\lrM{\lrn{\bH_i - \bI_0}_2\lrn{\bd_i^\T(\ol\bH_i - \bH_i)}_2} \\
    \leq & \lrO{\E\lri{\lrn{\bH_i - \bI_0}_2^2} \E\lrM{\lrn{\bd_i^\T(\ol\bH_i - \bH_i)}_2^2}}^{1/2}
    = O(n^{-3/2}).
\end{align*}
Therefore, 
\begin{equation}
    \lrN{m^{-1}\sum_i (\wt\bH_i - \bI_0)\bDelta(\bd_i - \bd_{0, i})^\T}_2 = \Op(n^{-1})\lrn{\bDelta}_2^2 + \Op(n^{-3/2})\lrn{\bDelta}_2. \label{Heq:P4L31}
\end{equation}

Combining \eqref{Heq:P4L21}, \eqref{Heq:P4L34}, \eqref{Heq:HIDLI}, and \eqref{Heq:P4L31}
\begin{equation}
    m^{-1}\sum_i (\wt\bH_i - \wt\bH)\bDelta\bd_i^\T 
        = -n^{-1} \bQ_{12}^\T \times (\bI_{d\times d} \otimes \bDelta) \times \bI_0^{-1} - \bR_{22}^\T \bI_0^{-1} + \bR_{42},
\end{equation}
where 
\begin{align*}
    & \bR_{42} = - (\wt\bH - \bI_0)\bDelta\bar\bl(\btheta_0)^\T\bI_0^{-1} + m^{-1}\sum_i (\wt\bH_i - \bI_0)\bDelta(\bd_i - \bd_{0, i})^\T + (\bI_0 - \wt\bH)\bDelta\bar\bd^\T \\
    & \lrn{\bR_{42}}_2 = \Op\lrm{n^{-1} + (nm)^{-1/2}}\lrn{\bDelta}_2^2 + \Op\lrm{n^{-3/2} + (nm)^{-1}}\lrn{\bDelta}_2.
\end{align*}
\end{proof}

\begin{lemma} \label{Hlm:P4} Consider $m^{-1}\sum_i(\wt\bH_i - \wt\bH)\bDelta\bDelta^\T(\wt\bH_i - \wt\bH)$ with $\btheta \in U(\rho)$. If Assumptions \ref{as:T0Def}-\ref{Has:LH} hold, 
    $$\lrN{m^{-1}\sum_i (\wt\bH_i - \wt\bH)\bDelta\bDelta^\T(\wt\bH_i - \wt\bH)}_2
        = \Op(1)\lrn{\bDelta}_2^4 + \Op(n^{-1})\lrn{\bDelta}_2^2.$$
\end{lemma}
\begin{proof} 
By definition $\sum_i \wt\bH_i = m \wt\bH$.
\begin{align*}
    \sum_i (\wt\bH_i - \wt\bH)\bDelta\bDelta^\T(\wt\bH_i - \wt\bH) & = \sum_i\wt\bH_i\bDelta\bDelta^\T(\wt\bH_i - \wt\bH) = \sum (\wt\bH_i - \bI_0)\bDelta\bDelta^\T(\wt\bH_i - \wt\bH) \\
        & = \sum (\wt\bH_i - \bI_0)\bDelta\bDelta^\T(\wt\bH_i - \bI_0 + \bI_0 - \wt\bH) \\
        & = \sum (\wt\bH_i - \bI_0)\bDelta\bDelta^\T(\wt\bH_i - \bI_0) + m (\wt\bH - \bI_0)\bDelta\bDelta^\T(\bI_0 - \wt\bH) \\
    \lrn{(\wt\bH_i - \bI_0)\bDelta}_2^k & \leq \lrn{\bDelta}_2^{k}\lrn{\wt\bH_i - \bH_i + \bH_i - \bI_0}_2^{k} \\
        & \leq 2^{k - 1}\lrn{\bDelta}_2^{2k}\lrM{n^{-1}\sum_j h(\bX_{ij})}^{k} + 2^{k - 1}\lrn{\bDelta}_2^{k}\lrn{\bH_i - \bI_0}_2^{k}. 
\end{align*}
When $k=2$, 
\begin{align*}
        & \lrN{m^{-1}\sum_i (\wt\bH_i - \wt\bH)\bDelta\bDelta^\T(\wt\bH_i - \wt\bH)}_2 \\
    \leq & m^{-1}\sum \lrn{(\wt\bH_i - \bI_0)\bDelta}_2^2 + \lrn{(\wt\bH - \bI_0)\bDelta}_2^2 \\
    = & 2\lrn{\bDelta}_2^{4}\lrO{m^{-1}\sum_i\lrM{n^{-1}\sum_j h(\bX_{ij})}^{2} + \lrM{(nm)^{-1}\sum\sum h(\bX_{ij})}^{2}} +\\
        & 2\lrn{\bDelta}_2^{2}\lrI{m^{-1}\sum_i\lrn{\bH_i - \bI_0}_2^{2} + \lrn{\bH - \bI_0}_2^{2}} \\
    = & \Op(1)\lrn{\bDelta}_2^4 + \Op(n^{-1})\lrn{\bDelta}_2^2.
\end{align*}
\end{proof}

\begin{lemma} \label{Hlm:IHD} Under Assumption \ref{as:T0Def}-\ref{Has:L3Cont}, for $\btheta \in U(\rho)$, 
\begin{equation*} 
    (\bI_0 - \ol\bH^{(g)})\bDelta = -\frac{1}{2} \bQ(\bDelta^{\otimes 2}) + \bR_{H},
 \end{equation*}
 where $\lrn{\bR_H} = \Op(\lrn{\bDelta}_2^3) + \Op(n^{-1/2}\lrn{\bDelta}_2^2) + \Op(n^{-1/2}m^{-1/2})\lrn{\bDelta}_2$.
\end{lemma}

\begin{proof}
Note that 
$$
    (\bI_0 - \ol\bH^{(g)})\bDelta = (\bI_0 - \bH_0^{(g)})\bDelta + (\bH_0^{(g)} - \ol\bH^{(g)})\bDelta.
$$
The first component, by Lemma \ref{Hlm:Zhang}, $\lrn{(\bI_0 - \bH_0^{(g)})\bDelta} = \Op(n^{-1/2}m^{-1/2})\lrn{\bDelta}_2$. Lemma \ref{Hlm:L3Dec} indicates
$    (\ol\bH^{(g)} - \bH^{(g)}_0)\bDelta = 2^{-1} \bQ(\bDelta^{\otimes 2}) + \bR_{H1},
$
 where $\lrn{\bR_{H1}} = \Op(n^{-1/2}\lrn{\bDelta}_2^2) + \Op(\lrn{\bDelta}_2^3)$. Therefore, 
 \begin{align*} 
    (\bI_0 - \ol\bH^{(g)})\bDelta = -\frac{1}{2} \bQ(\bDelta^{\otimes 2}) + \bR_{H},
 \end{align*}
 where $\lrn{\bR_H} = \Op(\lrn{\bDelta}_2^3) + \Op(n^{-1/2}\lrn{\bDelta}_2^2) + \Op(n^{-1/2}m^{-1/2})\lrn{\bDelta}_2$.
\end{proof}

\subsection*{Lemmas of single-node}
In this section, we present lemmas that apply to single-node or non-distributed data. The index of nodes is therefore dropped for parsimony. First, define the M-estimator
$$\wh\btheta = \text{argmin}_{\btheta \in \bTheta} n^{-1}\sum L(\btheta; \bX_i),
$$
and $\bd = \wh\btheta - \btheta_0$. Additionally, 
\begin{align*}
& \bl_i(\btheta) = \nabla L(\btheta;\bX_i), \bH_i(\btheta) = \nabla^2 L(\btheta;\bX_i) \\
& \bl_i = \nabla L(\btheta_0;\bX_i), \bH_i = \nabla^2 L(\btheta_0;\bX_i), \bI_0 = \E(\bH_1) \\
& \ol\bH_i = \int_0^1 \nabla^2 L\lrm{\btheta_0 + t(\wh\btheta - \btheta_0);\bX_i}\diff t \\
& \bar\bl = n^{-1}\sum \nabla L(\btheta_0;\bX_i), \bH = n^{-1}\sum \nabla^2 L(\btheta_0;\bX_i), \ol\bH = n^{-1}\sum \ol\bH_i.
\end{align*}

\begin{lemma} \label{Hlm:L3bd} Define $$\bG_j = \nabla^2 \lrM{\frac{\partial}{\partial \theta_{j}}L(\btheta;\bX)}, j = 1, \dots, d.$$
Under Assumption \ref{as:T0Def}-\ref{Has:LH}, when $\bG_j(\btheta;\bX)$ is continuous with respect to $\btheta$, for $j = 1, \dots, d$
$$    \max_j \lrN{\bG_j(\btheta;\bX)}_2 \leq h(\bX), \forall \btheta \in U(\rho).
$$
\end{lemma}
\begin{proof}
Recall the definition 
$$\nabla^3 L(\btheta;\bX_i)(\bu) = \begin{Bmatrix} \bu^\T\bG_1(\btheta;\bX_i) \\ \vdots \\ \bu^\T\bG_d(\btheta;\bX_i) \end{Bmatrix}.$$
Assumption \ref{Has:LH} indicates that $\forall \btheta', \btheta \in U(\rho), \lrn{\nabla^2 L(\btheta';\bX_i) - \nabla^2 L(\btheta;\bX_i)}_2 \leq h(\bX_i)\lrn{\bd}_2$, where $\bd = \btheta' - \btheta$. By Assumption \ref{Has:LH},
\begin{align}
    & \lrn{\nabla^2 L(\btheta';\bX_i) - \nabla^2 L(\btheta;\bX_i)}_2 = \lrN{\begin{Bmatrix}\bd^\T \int_0^1 \bG_1(\btheta + t\bd;\bX_i) \diff t \\ \vdots \\ \bd^\T \int_0^1 \bG_d(\btheta + t\bd;\bX_i) \diff t \end{Bmatrix}}_2 \leq h(\bX_i)\lrn{\bd}_2 \nonumber \\
    & \lrN{\bd^\T\int_0^1 \bG_j(\btheta + t\bd;\bX_i) \diff t}_2 \leq h(\bX_i)\lrn{\bd}_2. \label{Heq:LML3int}
\end{align}
Assumption \ref{Has:L3Cont} states that $\lrn{\bG_j (\btheta;\bX_i)}_2$ is continuous with respect to $\btheta$. Consequently, $\lrn{\bG_j(\btheta;\bX_i)}_2 \leq h(\bX_i)$ for $j=1, \dots, d,$ almost everywhere. 

To see this, suppose there exists $\btheta \in U(\rho)$ and $\lrn{\bG_j(\btheta;\bX)}_2 - h(\bX) > 2 \delta > 0$. By continuity, there exists a small ball around $\btheta$ such that 
$$U(\delta) = \lrM{\btheta' \in U(\rho); \lrN{\bG_j(\btheta';\bX) - \bG_j(\btheta;\bX)}_2 \leq \delta}.$$ 
Since $\bd = \btheta' - \btheta$ can be in any direction, we can choose $\btheta' \in U(\delta)$ such that 
\begin{align*}
    \lrn{\bd}_2^{-1}\lrN{\bd^\T \bG_j(\btheta;\bX)}_2 = \lrN{\bG_j(\btheta;\bX)}_2.
\end{align*}
Then by the triangular inequality 
\begin{align*}
    & \lrn{\bd}_2^{-1}\lrN{\bd^\T\int_0^1 \bG_j(\btheta + t\bd;\bX) \diff t}_2 \\
    \geq &  
    \lrn{\bd}_2^{-1}\lrN{\bd^\T \bG_j(\btheta;\bX)}_2 - \lrn{\bd}_2^{-1}\lrN{\bd^\T\int_0^1 \lrM{\bG_j(\btheta;\bX) - \bG_j(\btheta + t\bd;\bX)} \diff t}_2 \\
    \geq & \lrM{h(\bX)+ 2\delta} - \delta > h(\bX),
\end{align*}
which contradicts \eqref{Heq:LML3int}. 

\end{proof}

\begin{lemma} \label{Hlm:L3Dec} Let $\ol\bH = n^{-1}\sum_i\int_0^1 \nabla^2 L\lrm{\btheta_0 + t(\btheta - \btheta_0);\bX_i}\diff t$.
Under Assumption \ref{as:T0Def}-\ref{Has:L3Cont}, for all $\btheta \in U(\rho)$ and $\bDelta = \btheta - \btheta_0$, 
\begin{align*}
    & \lrn{\ol\bH - \bH}_2 \leq \lrM{n^{-1}\sum h(\bX_{i})}\lrn{\bDelta}, \\
    & (\ol\bH - \bH)\bDelta = (2n)^{-1}\sum \nabla^3 L(\btheta_0;\bX_{i})(\bDelta^{\otimes 2}) + \bR_1, \\
    & (\ol\bH - \bH)\bDelta = 2^{-1}\bQ(\bDelta^{\otimes 2}) + \bR_1 + \bR_2,
\end{align*}
where $\lrn{\bR_1}_2 \leq (6n)^{-1}\sum m(\bX_{i})\lrn{\bDelta}_2^3$, and $\lrn{\bR_2}_2 \leq R_2 \lrn{\bDelta}_2^2$ with $\E(R_2^k) = O(n^{-k/2})$ and $$k = \max_j\lrO{j \geq 2; \E\lrm{h(\bX_1)^j)} < \infty}.$$
\end{lemma}
\begin{proof} The first equation is a direct result from Assumption \ref{Has:LH}. For parsimony, we drop the center index $i$ in the proof and consider any $\btheta \in U(\rho)$. Let $\bDelta = \btheta - \btheta_0$. Under Assumption \ref{Has:L3Cont}, recall that $\ol\bH = n^{-1}\sum_{i} \int_0^1 \nabla^2 L\lrm{\btheta_0 + t(\btheta - \btheta_0);\bX_{i}}\diff t$, then
\begin{align*}
    \bl(\btheta) & = \bl(\btheta_0) + \ol\bH \bDelta = \bl(\btheta_0) + \bH \bDelta + n^{-1}\sum_i \int_0^1 (1 - t) \nabla^3 L\lrm{\btheta_0 + t(\btheta - \btheta_0);\bX_i}(\bDelta^{\otimes 2}) \diff t \\
        & = \bl(\btheta_0) + \bH \bDelta + (2n)^{-1}\sum_i \nabla^3 L(\btheta_0;\bX_i)(\bDelta^{\otimes 2}) + \bR_{1},
\end{align*}
where
\begin{align*}
    & \bR_{1} = n^{-1}\sum_i \int_0^1 (1 - t)\lrO{\nabla^3 L\lrm{\btheta_0 + t(\btheta - \btheta_0);\bX_i} - \nabla^3 L(\btheta_0;\bX_i)}(\bDelta^{\otimes 2}) \diff t \\
    & \lrn{\bR_{1}}_2 \leq (6n)^{-1}\sum m(\bX_i)\lrn{\bDelta}_2^3.
\end{align*}
Further decomposition gives
\begin{align*}
& n^{-1}\sum_i \nabla^3 L(\btheta_0;\bX_i)(\bDelta^{\otimes 2}) = \bQ(\bDelta^{\otimes 2}) + \bR_2\\
& \bR_2 = \lrO{n^{-1}\sum_i \nabla^3 L(\btheta_0;\bX_i) - \bQ}(\bDelta^{\otimes 2}) \\
    & \blkeq = \begin{pmatrix} \bDelta^\T\lrO{n^{-1}\sum \bG_1(\btheta_0;\bX_i) - \E\lrM{\bG_1(\btheta_0;\bX_i)}}\bDelta \\ \vdots \\ \bDelta^\T\lrO{n^{-1}\sum \bG_d(\btheta_0;\bX_i) - \E\lrM{\bG_d(\btheta_0;\bX_i)}}\bDelta \end{pmatrix}.
\end{align*}
For $j=1,\dots,d$, consider $n^{-1}\sum \bG_j(\btheta_0;\bX_i)$. By definition, $\bG_j(\btheta_0;\bX_i)$ is Hermitian, and Lemma \ref{Hlm:L3bd} indicates 
$\lrN{\bG_j(\btheta_0;\bX_i)}_2 \leq h(\bX_i)$. Let $\bZ_{j, i} = \bG_j(\btheta_0;\bX_i) - \E\lrM{\bG_j(\btheta_0;\bX_i)}$, and
Lemma \ref{Hlm:EZK} indicates 
    \begin{equation*}
        \E\lrI{\lrN{n^{-1}\sum_i  \bZ_{j, i}}_2^k} = O(n^{-k/2}),
    \end{equation*}
for some $k \geq 2$ and $\bE(\lrn{\bZ_{j, i}}_2^k)$ exists. In our case, we have
$$    \E\lrI{\lrn{\bZ_{j, i}}^k} \leq 2^{k-1}\E\lrM{h(\bX_i)^k} + 2^{k-1} \lambda_h^k.
$$
Hence, the maximum of $k$ depends on moments of $h(\bX_i)$.
By definition,
\begin{align*}
& \lrN{\bR_2}_2 = \lrN{\begin{pmatrix} \bDelta^\T n^{-1}\sum_i\bZ_{1, i} \bDelta \\ \vdots \\ \bDelta^\T n^{-1}\sum_i\bZ_{d, i} \bDelta \end{pmatrix}}_2 \leq \lrn{\bDelta}_2^2\lrI{\sum_j\lrN{n^{-1}\sum_i  \bZ_{j, i}}_2^2}^{1/2} \\
& \E\lrM{\lrI{\sum_j\lrN{n^{-1}\sum_i  \bZ_{j, i}}_2^2}^{k/2}} \leq d^{k/2-1}\sum_j \E\lrI{\lrN{n^{-1}\sum_i  \bZ_{j, i}}_2^{k}} = O(n^{-k/2}).
\end{align*}
Let $R_1 = \lrI{\sum_j\lrN{n^{-1}\sum_i  \bZ_{j, i}}_2^2}^{1/2}$, and we complete the proof. 
\end{proof}

\begin{lemma} \label{Hlm:E0}
 Let $\delta_{\rho} = \min\lrM{\rho, \rho\lambda_- / (4\lambda_h)}$, and define the following four events:
\begin{align*}
    {\mathcal{E}}_1 & = \lrM{n^{-1} \sum h(\bX_i) \leq 2\lambda_h} \\
    {\mathcal{E}}_2 & = \lrM{\lrn{\bH - \bI_0}_2 \leq \frac{\rho\lambda_-}{2}} \\
    {\mathcal{E}}_3 & = \lrM{\lrn{\bl(\btheta_0)}_2 \leq \frac{(1 - \rho) \lambda_-\delta_{\rho}}{2}} \\
    {\mathcal{E}}_4 & = \lrM{n^{-1} \sum m(\bX_i) \leq 2\lambda_m}.
\end{align*} Let ${\mathcal{E}}_0 = {\mathcal{E}}_1 \cap {\mathcal{E}}_2 \cap {\mathcal{E}}_3$ and $\wt{\mathcal{E}}_0 = {\mathcal{E}}_0\cap {\mathcal{E}}_4$. 
\begin{enumerate}[label={(\arabic*)}]
\item 
Suppose for some $k \geq 2$
\begin{align*}
& \E\lrm{\lrn{\nabla L(\btheta_0;\bX_i)}_2^k} \leq G^k \text{ and } \E\lrO{\lrn{\nabla^2 L(\btheta_0;\bX_i) - \E\lrm{\lrn{\nabla^2 L(\btheta_0;\bX_i)}}}_2^k} \leq H^k \\
& E\lrm{h(\bX_i)^k} \leq \lambda_h^k \text{ and } E\lrO{\lrA{h(\bX_i) - \E\lrm{h(\bX_i)}}^k} \leq \lambda_h^k.
\end{align*}
Then,
$$\P({\mathcal{E}}_0^c) = O(n^{-k/2}).$$ Specifically, under ${\mathcal{E}}_0$, $$\lrn{\bd}_2 \leq \frac{2\lrn{\bl(\btheta_0)}_2}{(1 - \rho)\lambda_-}, \text{ and } \bH(\btheta') \succeq (1 - \rho)\lambda_-\bI, \forall \btheta' \in U(\rho),$$ and $\lrn{\ol\bH}_2 \geq (1 - \rho)\lambda_-$.
\item With the additional assumption $\max\lrI{E\lrm{m(\bX_i)^k}, E\lrO{\lrA{m(\bX_i) - \E\lrm{m(\bX_i)}}^k}} \leq \lambda_m^k,$ $$\P(\wt{\mathcal{E}}_0^c) = O(n^{-k/2}).$$ 
\end{enumerate}
\end{lemma}
\begin{proof}
By Lemma 6 and Lemma 7 in \cite{Zhang2013}, $\P({\mathcal{E}}_0^c) = O(n^{-k/2}).$
The probability of ${\mathcal{E}}_4^c$ can be derived in a similar way. By Lemma \ref{Hlm:EZK},
\begin{align*}
    \P({\mathcal{E}}_4^c) & = \E\lrO{\mone\lrM{n^{-1} \sum m(\bX_i) > 2 \lambda_m }} \\
        & \leq \E\lrO{\mone\lrM{n^{-1} \sum m(\bX_i) - \E\lrM{m(\bX_i)} > \lambda_m }} \\
        & \leq \lambda_m^{-k}\E\lrO{\lrA{n^{-1} \sum m(\bX_i) - \E\lrm{m(\bX_i)}}^k} \\
        & = O(n^{-k/2}).
\end{align*}
Hence, $\P(\wt{\mathcal{E}}_0^c) \leq \P({\mathcal{E}}_0^c) + \P({\mathcal{E}}_4^c) = O(n^{-2})$. 
\end{proof}

\begin{lemma} \label{Hlm:Zhang} Under the conditions of Lemma \ref{Hlm:E0} (1), %there exist constants $C$, $C'$, 
    \begin{align*}
    & \E(\lrn{\bar\bl(\btheta_0)}_2^k) \leq O(n^{-k/2}) \\
    & \E(\lrn{\bH - \bI_0}_2^k) \leq O(n^{-k/2}) \\
    & \E(\lrn{\bd}_2^k) = O(n^{-k/2}).
    \end{align*}
\end{lemma}
See Theorem 1 and Lemmas 7-9 given in \cite{Zhang2013}. 

\begin{lemma} \label{Hlm:HHD} Under the conditions of Lemma \ref{Hlm:E0} (1),  $$\E\lrM{\lrn{(\bH - \ol\bH)\bd}_2^{k/2}} = O(n^{-{k/2}}).$$
\end{lemma}
\begin{proof}
\begin{align*}
    \lrn{(\bH - \ol\bH)\bd}_2 & \leq \mone({\mathcal{E}}_0)\lrn{(\bH - \ol\bH)\bd}_2 + \mone({\mathcal{E}}_0^c)\lrN{(\bH - \bI_0)\bd}_2 + \mone({\mathcal{E}}_0^c)\lrN{\bI_0\bd}_2 + \mone({\mathcal{E}}_0^c)\lrn{\bl} \\
    \E\lrM{\mone({\mathcal{E}}_0)\lrn{(\bH - \ol\bH)\bd}_2^{k/2}} & \leq (2\lambda_h)^{k/2} \E\lrI{\lrn{\bd}_2^{k}} = O(n^{-{k/2}}) \\
    \E\lrM{\mone({\mathcal{E}}_0^c)\lrn{(\bH - \bI_0)\bd}_2^{k/2}} & \leq \E\lrI{\lrn{\bH - \bI_0}_2^{k/2}\lrn{\bd}_2^{k/2}} \leq \lrM{\E\lrI{\lrn{\bH - \bI_0}_2^{k}}\E\lrI{\lrn{\bd}_2^{k}}}^{1/2} = O(n^{-{k/2}}) \\
    \E\lrM{\mone({\mathcal{E}}_0^c)\lrn{\bI_0\bd}_2^{k/2}} & \leq \lambda_+^{k/2} \E\lrM{\mone({\mathcal{E}}_0^c)\lrn{\bd}_2^{k/2}} \leq \lrM{\P({\mathcal{E}}_0^c)\E\lrI{\lrn{\bd}_2^{k}}}^{1/2} = O(n^{-k/2}) \\
    \E\lrM{\mone({\mathcal{E}}_0^c)\lrn{\bl}_2^{k/2}} & \leq \E\lrM{\mone({\mathcal{E}}_0^c)\lrn{\bl}_2^{k/2}} \leq \lrM{\P({\mathcal{E}}_0^c)\E\lrI{\lrn{\bl}_2^{k}}}^{1/2} = O(n^{-k/2}) \\
    \E\lrI{\lrn{(\bH - \ol\bH)\bd}_2^{k/2}} & \leq 4^{{k/2}-1}\E\lrM{\mone({\mathcal{E}}_0)\lrn{(\bH - \ol\bH)\bd}_2^{k/2}} + 4^{{k/2}-1}\E\lrM{\mone({\mathcal{E}}_0^c)\lrn{(\bH - \bI_0)\bd}_2^{k/2}} +\\
        & \blkeq 4^{{k/2}-1}\E\lrM{\mone({\mathcal{E}}_0^c)\lrn{\bI_0\bd}_2^{k/2}} + 4^{{k/2}-1}\E\lrM{\mone({\mathcal{E}}_0^c)\lrn{\bl}_2^{k/2}} \\
        & = O(n^{-{k/2}}).
\end{align*}
\end{proof}

\begin{lemma} \label{Hlm:EDD} Under the conditions of Lemma \ref{Hlm:E0} (1), $\E\lri{\lrN{\bd - \bd_0}_2^{k/2}}  = O(n^{-k/2})$. Specifically, $\lrn{\E(\bd)}_2 = O(n^{-1})$.
\end{lemma}
\begin{proof}
Let $\bd_0 = -\bI_0^{-1}\bl(\btheta_0)$. Note that $\wh\btheta$ may not be in $U(\rho)$.
$$
    \bd - \bd_{0} = \bI_0^{-1}(\bI_0 - \ol\bH)\bd = \bI_0^{-1}(\bI_0 - \bH)\bd + \bI_0^{-1}(\bH - \ol\bH)\bd.
$$
By Lemma \ref{Hlm:HHD}, $\E\lrM{\lrn{(\ol\bH - \bH)\bd}_2^{k/2}} = O(n^{-k/2})$, and 
\begin{align*}
    \E\lri{\lrN{(\bH - \bI_0)\bd}_2^{k/2}} & \leq \lrM{\E\lrI{\lrN{\bH - \bI_0}_2^{k}\lrN{\bd}_2^{k}}}^{1/2} = O(n^{-k/2}) \\
    \E\lrI{\lrn{\bd - \bd_0}_2^{k/2}} & \leq 2^{k/2-1}\lambda_-^{-k/2}\E\lrI{\lrN{(\bI_0 - \bH)\bd}_2^{k/2}} + 2^{k/2-1}\lambda_-^{-k/2}\E\lrI{\lrN{( \bH - \ol\bH)\bd}_2^{k/2}} \\
        & = O(n^{-k/2}).
\end{align*}
\end{proof}

\begin{lemma} \label{Hlm:R1MC}
Given independent and identically distributed random matrices $\bW_i$ with $\E\lri{\lrN{\bW_i}_2^2} < \infty$ and rank $r > 0$,
\begin{equation*}
\E\lrI{\lrN{n^{-1}\sum \bW_i}_2^2} \leq n^{-1}r\E(\lrN{\bW_1}_2^2) + d\lrN{\E(\bW_1)}_2^2.
\end{equation*}
\end{lemma}
\begin{proof} Recall that $\lrn{\bW_i}_2^2 \leq \lrn{\bW_i}_F^2 \leq r\lrn{\bW_i}_2^2$.
\begin{align*}
    \E\lrI{\lrN{\sum\bW_i}_F^2} & = \E\lrO{\Tr\lrM{\lrI{\sum\bW_i}\lrI{\sum\bW_i}^\T}} \\
        & = \sum \E(\lrN{\bW_i}_F^2) + \sum_i\sum_{j \ne i} \Tr\lrM{\E(\bW_i)\E(\bW_j^\T)} \\
        & = n\E(\lrN{\bW_1}_F^2) + n(n-1)\lrN{\E(\bW_1)}_F^2 \\
    \E\lrI{\lrN{n^{-1}\sum\bW_i}_2^2} & \leq \E\lrI{\lrN{n^{-1}\sum\bW_i}_F^2} \leq n^{-1}r\E(\lrN{\bW_1}_2^2) + d\lrN{\E(\bW_1)}_2^2. 
\end{align*}
Higher order convergence can be controlled by the Marcinkiewicz–Zygmund inequality. 
\end{proof}

\begin{lemma} \label{Hlm:EZK} Let $\lrm{\bZ_i}_i^n$ be independent and identically distributed Hermitian matrices with $\E(\bZ_i) = \bzero$ and $\E(\lrn{\bZ_i}_2^k) \leq \zeta^k$ for some positive numbers $\zeta > 0$ and $k\geq 2$.
    \begin{equation*}
        \E\lrI{\lrN{n^{-1}\sum  \bZ_i}_2^k} \leq \lrO{n^{-\frac{k}{2}} 4^{k-1}\lrM{\text{e} \lrI{k + 2\log d}}^{\frac{k}{2}} + n^{-k+1} 2^{3k-1}\lrM{\text{e} \lrI{k + 2\log d}}^k} \zeta^k.
    \end{equation*}
\end{lemma}
\begin{proof}
    %Let $\lrt{\bW}_p$ denote the Schatten p-norm of $\bW$. By definition, $\lrt{\bW}_p \geq \lrt{\bW}_{p'}$ for $1 \leq p \leq p' \leq \infty$ and $\lrt{\bW}_{\infty} = \lrn{\bW}_2$.
    %Recall that $\bW$ is Hermitian.
    %Consider two independent and identically distributed copies $\bW$ and $\bW'$ and let $\bZ = \bW - \E(\bW)$ and $\bZ' = \bW' - \E(\bW')$. 
    Let $\lrm{\ve_i}^m$ be independent and identically distributed Rademacher random variables. 
    \begin{align*}
        \lrN{\sum  \bZ_i}_2^k & = \lrN{\sum  \bZ_i - \E\lrI{\sum \bZ_i'}}_2^k
            = \lrN{\E\lrM{\sum  \lri{\bZ_i- \bZ_i'}\mid\bZ_i}}_2^k \\
            & \leq \E\lrI{\lrN{\sum  \lri{\bZ_i- \bZ_i'}}_2^k\mid\bZ_i} \\
        \E\lrI{\lrN{\sum  \bZ_i}_2^k} & \leq \E\lrI{\lrN{\sum  \lri{\bZ_i- \bZ_i'}}_2^k}.
    \end{align*}
    Note that each element of $\bZ_i - \bZ_i'$ is symmetrically distributed. Then,
    \begin{align*}
        \E\lrI{\lrN{\sum  \bZ_i}_2^k} & \leq \E\lrI{\lrN{\sum  \lri{\bZ_i- \bZ_i'}}_2^k} = \E\lrI{\lrN{\sum  \ve_i\lri{\bZ_i- \bZ_i'}}_2^k\mid\ve_i} \\
            & = \E\lrI{\lrN{\sum  \ve_i\lri{\bZ_i- \bZ_i'}}_2^k} \leq 2^{k-1}\E\lrI{\lrN{\sum  \ve_i\bZ_i}_2^k}.
    \end{align*}
    By Theorem A.1 (2) in \cite{Chen2012}, since $\ve_i\bZ_i$ are independent and identically distributed symmetrically distributed Hermitian matrices, 
    \begin{align*}
        \lrM{\E\lrI{\lrN{\sum n^{-1}\ve_i\bZ_i}_2^k}}^{\frac{1}{k}} & \leq (\text{e} c)^{1/2} \lrN{\lrM{\sum \E\lrI{n^{-2}\bZ_i^2}}^{1/2}}_2 + 2 \text{e} c \lrM{\E\lrI{\max_i\lrN{n^{-1}\bZ_i}_2^k}}^{\frac{1}{k}} \\
            & \leq n^{-1/2}(\text{e} c)^{1/2} \lrN{\lrM{\E\lrI{\bZ_i^2}}^{1/2}}_2 + n^{-1} 2 \text{e} c \lrM{\E\lrI{\sum \lrN{\bZ_i}_2^k}}^{\frac{1}{k}} \\
            & \leq n^{-1/2} (\text{e} c)^{1/2} \lrN{\E\lri{\bZ_i^2}}_2^{1/2} + n^{-1+\frac{1}{k}} 2\text{e} c \lrM{\E\lri{\lrN{\bZ_i}_2^k}}^{\frac{1}{k}},
    \end{align*}
    where $c = k + 2\log d$. Therefore, 
    \begin{align*}
        \E\lrI{\lrN{n^{-1} \sum  \bZ_i}_2^k} & \leq 2^{k-1}\E\lrI{\lrN{\sum  n^{-1}\ve_i\bZ_i}_2^k} \\
            & \leq n^{-\frac{k}{2}} 4^{k-1}(\text{e} c)^{\frac{k}{2}} \lrN{\E\lri{\bZ_i^2}}_2^{\frac{k}{2}} + n^{-k+1} 2^{3k - 1}(\text{e} c)^k \lrM{\E\lri{\lrN{\bZ_i}_2^k}} \\
            & \leq \lrM{n^{-\frac{k}{2}} 4^{k-1}(\text{e} c)^{\frac{k}{2}} + n^{-k+1} 2^{3k-1}(\text{e} c)^k}\zeta^k.
    \end{align*}
    
\end{proof}

\begin{lemma}\label{Hlm:N3DD} 
    Under Assumptions \ref{as:T0Def}-\ref{Has:L3Cont}, 
\begin{equation}
    \bl(\btheta_0)\bd^\T(\ol\bH - \bH)\bI_0^{-1} = \bW + \bW',
\end{equation}
where $\bW = 2^{-1} \bl(\btheta_0)\lrM{\bQ(\bd_0^{\otimes 2})}^\T\bI_0^{-1}$, with $\lrn{\E(\bW)}_2 = O(n^{-2})$, $\E(\lrn{\bW}_2) = O(n^{-3/2})$, and  $\E\lrI{\lrN{\bW'}_2} = O(n^{-2})$.

Additionally, when Assumption \ref{Has:LH} holds, 
$
    \E\lrI{\lrn{\bW}_2^2} = O(n^{-3}).
$
\end{lemma}
\begin{proof}
Recall that by Assumption \ref{Has:LH} and Lemma \ref{Hlm:E0} (2), $\P(\wt{\mathcal{E}}_0^c) = O(n^{-2})$. 
\begin{align}
 \E\lrm{\mone(\wt{\mathcal{E}}_0^c)\lrn{\bl(\btheta_0)\bd^\T(\bH - \ol\bH)\bI_0^{-1}}_2}
    & \leq \lambda_-^{-1}\lrO{\P(\wt{\mathcal{E}}_0^c)\E\lrM{\lrn{\bl(\btheta_0)}_2^4}}^{1/4}\lrO{\E\lrM{\lrN{\bd^\T(\bH - \ol\bH)}_2^2}}^{1/2} \nonumber \\
    & = O(n^{-2}) \label{Heq:ECLDHH}.
\end{align}
Under $\wt{\mathcal{E}}_0$, we have $\wh\btheta \in U(\rho)$. Lemma \ref{Hlm:L3Dec} indicates
$$
    (\ol\bH - \bH)\bd = 2^{-1}\bQ(\bd^{\otimes 2}) + \bR_{13} + \bR_{23},
$$
where $\lrn{\bR_{13}}_2 \leq (6n)^{-1}\sum m(\bX_{ij})\lrn{\bd}_2^3$, and $\lrn{\bR_{23}}_2 \leq R_2 \lrn{\bd}_2^2$ with $\E(R_2^k) = O(n^{-k/2})$. 
\begin{align}
    \bl(\btheta_0)\bd^\T(\ol\bH - \bH)\bI_0^{-1} & =  2^{-1}\bl(\btheta_0)\lrM{\bQ(\bd_0^{\otimes 2})}^\T\bI_0^{-1} + \label{Heq:fst11} \\
        & \blkeq 2^{-1}\bl(\btheta_0)\lrM{\bQ(\bd^{\otimes 2}) - \bQ(\bd_0^{\otimes 2})}^\T\bI_0^{-1} + \label{Heq:fst12} \\
        & \blkeq \bl(\btheta_0)\bR_{13}^\T\bI_0^{-1}+\bl(\btheta_0)\bR_{23}^\T\bI_0^{-1}. \label{Heq:fst13}
\end{align}
Also, under $\wt{\mathcal{E}}_0$, $\lrn{\bd}_2 \leq 2\lrn{\bl(\btheta_0)}_2\lrm{(1 - \rho)\lambda_-}^{-1}$ and $\lrn{\bH - \bI_0}_2 \leq \rho\lambda_-/2$. 

\subsubsection*{First term \eqref{Heq:fst11}}
Consider $\bl(\btheta_0)\lrM{\bQ(\bd_0^{\otimes 2})}^\T\bI_0^{-1}$. 
Note that this random element does not depends on $\wt{\mathcal{E}}_0$, instead, on $\bl(\btheta_0;\bX_1)$ and $h(\bX_1)$. Lemma \ref{Hlm:L3bd} gives for $j = 1, \dots, d$, $\lrN{\bG_j (\btheta_0;\bX_i)}_2 \leq h(\bX_i)$ and $\lrN{\E\lrM{\bG_{j}(\btheta_0;\bX_i)}}_2 \leq \lambda_h$. Then,
\begin{align}
    \lrN{\bI_0^{-1}\bl(\btheta_0)\lrM{\bQ(\bd_0^{\otimes 2})}^\T}_2 & \leq \lrn{\bd_0}_2\lrN{ \begin{bmatrix} \bd_0^\T \E\lrM{\bG_{1}(\btheta_0;\bX_i)}\bd_0 \\ \vdots \\ \bd_0^\T \E\lrM{\bG_{d}(\btheta_0;\bX_i)}\bd_0 \end{bmatrix}}_1 \nonumber \\
        & \leq d\lambda_-^{-3}\lrn{\bl(\btheta_0)}_2^3\lambda_h \label{Heq:LLLDef} \\
     \E\lrO{\lrN{\bI_0^{-1}\bl(\btheta_0)\lrM{\bQ(\bd_0^{\otimes 2})}^\T}_2}  & \leq d\lambda_-^{-3}\E\lrM{\lrn{\bl(\btheta_0)}_2^3\lambda_h} = O(n^{-3/2}). \label{Heq:P1L3212}
        %\\
    %\E\lrO{\lrN{\bl(\btheta_0)\bI_0^{-1}\lrM{\nabla^3 L(\btheta_0;\bX_{i})(\bd_0^{\otimes 2})}^\T}_2^2} & \leq  d^2\lambda_-^{-6}\E\lrM{\lrn{\bl(\btheta_0)}_2^6\lambda_h(\bX_i)^2} = O(n^{-3}).
\end{align}
Note that $\E(\bd_0) = \bzero$. Let $\bd_{0, i} = -\bI_0^{-1}\bl(\btheta_0;\bX_i)$ and $\bd_0 = n^{-1}\sum \bd_{0, i}$. 
\begin{align*}
    -\bI_0^{-1}\bl(\btheta_0)\lrM{\bQ(\bd_0^{\otimes 2})}^\T & = n^{-3} \sum_{i=1}^n\sum_{j=1}^n\sum_{k=1}^n \bQ(\bd_{0, i}\otimes\bd_{0, j}\otimes\bd_{0, k}) \\
    \E\lrO{-\bI_0^{-1}\bl(\btheta_0)\lrM{\bQ(\bd_0^{\otimes 2})}^\T} & = n^{-3}\sum_i \E\lrM{\bQ(\bd_{0, i}\otimes\bd_{0, i}\otimes\bd_{0, i})} \\
    & = n^{-2} \E\lrM{\bQ(\bd_{0, 1}\otimes\bd_{0, 1}\otimes\bd_{0, 1})}. 
\end{align*}
Hence, its expectation can be controlled.
\begin{equation} 
    \lrN{\E\lrO{\bI_0^{-1}\bl(\btheta_0)\lrM{\bQ(\bd_0^{\otimes 2})}^\T}}_2 \leq n^{-2} \E\lrm{d \lambda_h \lrn{\bl(\btheta_0;\bX_1)}_2^3\lrn{\bI_0^{-1}}_2^3} = O(n^{-2}).  \label{Heq:P1L321}
\end{equation}
Let $\bW = 2^{-1}\bl(\btheta_0)\lrO{\bQ(\bd_0^{\otimes 2})}^\T\bI_0^{-1}$. Therefore, $\lrn{\E(\bW)}_2 = O(n^{-2})$ and $\E(\lrn{\bW}_2) = O(n^{-3/2})$. When Assumption \ref{Has:LH} holds, by definition \eqref{Heq:LLLDef}
$$
    \E\lrI{\lrn{\bW}_2^2} \leq d^2\lambda_-^{-4}\lambda_h^2 \E\lrM{\lrn{\bl(\btheta_0)}_2^6} = O(n^{-3}).
$$
On the other hand, by \eqref{Heq:LLLDef}, 
\begin{equation*}
\E\lrM{\lrN{\mone(\wt{\mathcal{E}}_0^c)\bW}_2} \leq d\lambda_-^{-3}\E\lrM{\mone(\wt{\mathcal{E}}_0^c) \lrn{\bl(\btheta_0)}_2^3\lambda_h} \leq d\lambda_-^{-3}\lambda_h\lrM{\P(\wt{\mathcal{E}}_0^c)}^{1/4}\lrm{\E\lrI{\lrn{\bl}_2^4}}^{3/4} = O(n^{-2}).
\end{equation*}
Also note that together with \eqref{Heq:ECLDHH},
\begin{equation}
    \E\lrM{\lrN{\mone(\wt{\mathcal{E}}_0^c)\bW'}_2} \leq  \E\lrm{\mone(\wt{\mathcal{E}}_0^c)\lrn{\bl(\btheta_0)\bd^\T(\bH - \ol\bH)\bI_0^{-1}}_2} + \E\lrM{\lrN{\mone(\wt{\mathcal{E}}_0^c)\bW}_2} = O(n^{-2}). \label{Heq:ECWP}
\end{equation}
We are going to show that $\E\lrM{\lrN{\mone(\wt{\mathcal{E}}_0)\bW'}_2} = O(n^{-2})$. 
\subsubsection*{Second term \eqref{Heq:fst12}}
By definition
\begin{align*}
    \nabla^3 \E\lrM{L(\btheta_0;\bX_{i})(\bd^{\otimes 2})} - \nabla^3 \E\lrM{L(\btheta_0;\bX_{i})(\bd_0^{\otimes 2})}
    & = 
    \begin{bmatrix} (\bd + \bd_0)^\T \E\lrM{\bG_{1}(\btheta_0;\bX_i)}(\bd - \bd_0) \\ \vdots \\ (\bd + \bd_0)^\T\E\lrM{\bG_{d}(\btheta_0;\bX_i)}(\bd - \bd_0) \end{bmatrix}.
    %\\
    %& =  \begin{bmatrix} (\bd + \bd_0)^\T\nabla^2 \lrM{\frac{\partial}{\partial \theta_{1}}L(\btheta_0;\bX_i)}\bI_0^{-1}(\bI_0 - \ol\bH)\bd \\ \vdots \\ (\bd + \bd_0)^\T\nabla^2 \lrM{\frac{\partial}{\partial \theta_{d}}L(\btheta_0;\bX_i)}\bI_0^{-1}(\bI_0 - \ol\bH)\bd \end{bmatrix} 
\end{align*}
Therefore, for $j = 1, \dots, d$, under $\wt{\mathcal{E}}_0$
\begin{align*}
    & \lrA{(\bd + \bd_0)^\T \E\lrM{\bG_{j}(\btheta_0;\bX_i)}(\bd - \bd_0)} \\
    \leq & (\lrn{\bd}_2 + \lrn{\bd_0}_2)\lambda_h \lambda_-^{-1}\lrI{\lrn{\ol\bH - \bH}_2 + \lrn{\bH - \bI_0}_2}\lrn{\bd}_2  \\
    \leq & \frac{2\lambda_h\lrO{2\lrm{\lambda_-(1 - \rho)}^{-1} + \lambda_-^{-1}}}{\lambda_-^2(1 - \rho)}\lrn{\bl(\btheta_0)}_2^2 \lrM{\frac{4\lambda_h}{(1 - \rho)\lambda_-}\lrn{\bl(\btheta_0)} + \lrn{\bH - \bI_0}_2}\\
    \leq & \frac{8\lambda_h^2\lrO{2\lrm{\lambda_-(1 - \rho)}^{-1} + \lambda_-^{-1}}}{\lambda_-^3(1 - \rho)^2}\lrn{\bl(\btheta_0)}_2^3 + \frac{2\lambda_h\lrO{2\lrm{\lambda_-\rho(1 - \rho)}^{-1} + \lambda_-^{-1}}}{\lambda_-^2(1 - \rho)}\lrn{\bl(\btheta_0)}_2^2\lrn{\bH - \bI_0}_2.
\end{align*}
Consequently,
\begin{align}
    & \E\lrI{\mone(\wt{\mathcal{E}}_0)\lrN{\bl(\btheta_0) \lrO{\bQ(\bd^{\otimes 2}) - \bQ(\bd_0^{\otimes 2})}^\T\bI_0^{-1}}_2} \nonumber \\
    \leq & \E\lrO{\lrN{\bl(\btheta_0)}_2 \lrN{ \lrO{\bQ(\bd^{\otimes 2}) - \bQ(\bd_0^{\otimes 2})}^\T\bI_0^{-1}}_1} \nonumber \\
    \leq & d c_{1\rho} \E\lrm{\lrn{\bl(\btheta_0)}_2^4} + d c_{2\rho} \E\lrm{\lrn{\bl(\btheta_0)}_2^3\lrn{\bI_0 - \bH}_2} \nonumber \\
    = & O(n^{-2}). \label{Heq:P1L322}
\end{align}

\subsubsection*{The last term \eqref{Heq:fst13}}
\begin{align}
    \E\lrM{\mone(\wt{\mathcal{E}}_0)\lrn{\bl(\btheta_0)\bR_{13}^\T\bI_0^{-1}}_2} & \leq \lambda_-^{-1}3^{-1}\lambda_m\E\lrM{ \mone(\wt{\mathcal{E}}_0)\lrn{\bl(\btheta_0)}\lrn{\bd}_2^3} \nonumber \\
        & \leq 8 \lambda_-^{-1} 3^{-1}\lambda_m\E\lrO{ \lrn{\bl(\btheta_0)}_2^4 \lrm{(1 - \rho)\lambda_-}^{-3}} = O(n^{-2}) \label{Heq:P1L323}, \\
    \E\lrM{\mone(\wt{\mathcal{E}}_0)\lrn{\bl(\btheta_0)\bR_{23}^\T\bI_0^{-1}}_2} & \leq \lambda_-^{-1}\E\lrM{ \mone(\wt{\mathcal{E}}_0)\lrn{\bl(\btheta_0)}\lrn{\bd}_2^2R_2} \nonumber \\
        & \leq 4 \lambda_-^{-1} \E\lrO{ \lrn{\bl(\btheta_0)}_2^3 R_2 \lrm{(1 - \rho)\lambda_-}^{-2}} = O(n^{-2}) \label{Heq:P1L324}.
\end{align}

With \eqref{Heq:P1L322}, \eqref{Heq:P1L323}, \eqref{Heq:P1L324}, and \eqref{Heq:ECWP}
\begin{align*}
    & \E\lrM{\lrN{\mone(\wt{\mathcal{E}}_0)\bW'}_2} = O(n^{-2}) \\
    & \E\lrM{\lrN{\bW'}_2} \leq \E\lrM{\lrN{\mone(\wt{\mathcal{E}}_0)\bW'}_2} + \E\lrM{\lrN{\mone(\wt{\mathcal{E}}_0^c)\bW'}_2} = O(n^{-2}).
\end{align*}
\end{proof}

\begin{lemma} \label{Hlm:LDD} Under Assumptions \ref{as:T0Def}-\ref{Has:LH}, for each node,
let $\bW_1 = \bl(\btheta_0)\bd_0^\T(\bI_0 - \bH)\bI_0^{-1}$ and $\bW_2 = 2^{-1} \bl(\btheta_0)\lrM{\nabla^3 \bQ(\bd_0^{\otimes 2})}^\T\bI_0^{-1}$. Then,
\begin{align*}
    & \bl(\btheta_0)(\bd - \bd_0)^\T = \bW_1 - \bW_2 + \bW' \\
    & \E\lri{\lrn{\bW_1}_2} = O(n^{-3/2}), \lrn{\E\lri{\bW_1}}_2 = O(n^{-2}) \\
    & \E\lri{\lrn{\bW_2}_2} = O(n^{-3/2}), \lrn{\E\lri{\bW_2}}_2 = O(n^{-2}) \\
    & \E\lri{\lrn{\bW'}_2} = O(n^{-2}).
\end{align*}

Additionally, when Assumption \ref{Has:LH} holds, 
$
    \E\lrI{\lrn{\bW_1}_2^2} = O(n^{-3})$ and $\E\lrI{\lrn{\bW_2}_2^2} = O(n^{-3}).
$
\end{lemma}
\begin{proof}
Recall that $\bd - \bd_{0} = \bI_0^{-1}(\bI_0 - \ol\bH)\bd$,
\begin{align*}
\bl(\btheta_0)(\bd - \bd_{0})^\T & = \bl(\btheta_0)\bd^\T(\bI_0 - \ol\bH)\bI_0^{-1} \\
    & = \bl(\btheta_0)\bd_0^\T(\bI_0 - \bH)\bI_0^{-1} + \\
    & \blkeq \bl(\btheta_0)(\bd - \bd_0)^\T(\bI_0 - \bH)\bI_0^{-1} + \\
    & \blkeq \bl(\btheta_0)\bd^\T(\bH - \ol\bH)\bI_0^{-1}.
\end{align*}

Consider $\bl(\btheta_0)\bd_0^\T(\bI_0 - \bH)\bI_0^{-1}$:
\begin{align*}
    \E\lrM{\lrN{\bl(\btheta_0)\bd_0^\T(\bI_0 - \bH)\bI_0^{-1}}_2} & \leq \lambda_-^{-2} \E\lrM{\lrN{\bl}_2^2\lrn{\bI_0 - \bH}_2} = O(n^{-3/2}) \\
    \E\lrM{\bl(\btheta_0)\bd_0^\T(\bI_0 - \bH)} & = n^{-3}\sum_{i=1}^n\sum_{j=1}^n\sum_{k=1}^n \E\lrM{\bl_i\bl_j^\T\bI_0^{-1}(\bH_k - \bI_0)} \\
        & = n^{-2}\E\lrM{\bl_1\bl_1^\T\bI_0^{-1}(\bH_1 - \bI_0)} \\
    \lrN{\E\lrM{\bl(\btheta_0)\bd_0^\T(\bI_0 - \bH)}}_2 & = O(n^{-2}). 
\end{align*}
Let $\bW_1 = \bl(\btheta_0)\bd_0^\T(\bI_0 - \bH)\bI_0^{-1}$. Then, $\E\lri{\lrn{\bW_1}_2} = O(n^{-3/2})$ and $\lrn{\E\lri{\bW_1}}_2 = O(n^{-2})$. When Assumption \ref{Has:LH} holds, 
$$
    \E\lrI{\lrn{\bW_1}_2^2} \leq \lambda_-^{-4} \E\lrM{\lrN{\bl}_2^4\lrn{\bI_0 - \bH}_2^2} \leq \lambda_-^{-4} \lrM{\E\lrI{\lrn{\bl}_2^6}}^{2/3}\lrM{\E\lrI{\lrn{\bI_0 - \bH}_2^6}}^{1/3} = O(n^{-3}).
$$

Consider $\bl(\btheta_0)(\bd - \bd_0)^\T(\bI_0 - \bH)\bI_0^{-1}$. 
By Lemma \ref{Hlm:EDD}, 
\begin{align*}
    \E\lrM{\lrn{\bl(\btheta_0)(\bd - \bd_0)^\T(\bH - \bI_0)\bI_0^{-1}}_2} & \leq \lambda_-^{-1}\E\lrM{\lrN{\bl(\btheta_0)}_2\lrN{\bd - \bd_0}_2\lrn{\bH - \bI_0}_2} \\
        & \leq \lambda_-^{-1}\lrM{\E\lri{\lrN{\bl(\btheta_0)}_2^4}\E\lri{\lrn{\bH - \bI_0}_2^4}}^{1/4}\lrM{\E\lri{\lrN{\bd - \bd_0}_2^2}}^{1/2} \\
        & = O(n^{-2}). 
\end{align*}

Consider $\bl(\btheta_0)\bd^\T(\bH - \ol\bH)\bI_0^{-1}$. By Lemma \ref{Hlm:N3DD}, 
\begin{equation*}
    \bl(\btheta_0)\bd^\T(\bH - \ol\bH)\bI_0^{-1} = -\bW_2 + \bW'_2,
\end{equation*}
where $\bW_2 = 2^{-1} \bl(\btheta_0)\lrM{\bQ(\bd_0^{\otimes 2})}^\T\bI_0^{-1}$, with $\lrn{\E(\bW_2)}_2 = O(n^{-2})$, $\E(\lrn{\bW_2}_2) = O(n^{-3/2})$, and  $\E\lrI{\lrN{\bW_2'}_2} = O(n^{-2})$.
The proof is complete by putting them together. 
\end{proof}

\begin{lemma} \label{Hlm:lcH} \ 
\begin{enumerate}
\item In a single-node, when the $k$th ($k \geq 2$) moments of $\bl(\btheta_0;\bX_i)$ and $\bH(\btheta_0;\bX_i)$ exist, then,
    \begin{align*}
        & \E\lrm{\lrn{\bl \otimes (\bH - \bI_0)}_2^{k/2}} = O(n^{-k/2}) \\
        & \E\lrm{\bl \otimes (\bH - \bI_0)} = n^{-1}\bQ_{12} \text{, where } \bQ_{12} = \E\lro{\bl(\btheta_0;\bX_1) \otimes \lrm{\bH(\btheta_0;\bX_1) - \bI_0}}.
    \end{align*}
\item When $m$ independent and identically distributed copies $\bl_i$ and $\bH_i$ exist, and $k \geq 4$, 
    \begin{equation*}
        m^{-1}\sum \bl_i \otimes (\bH_i - \bI_0) = n^{-1}\bQ_{12} + \bR, \text{ with } \E\lri{\lrN{\bR}_2^{k/2}} = O(m^{-k/4}n^{-k/2}),
    \end{equation*}
    where $\bQ_{12} = \E\lrO{\bl(\btheta_0;\bX_1) \otimes \lrm{\bH(\btheta_0;\bX_1) - \bI_0}}$. 
\end{enumerate}
\end{lemma}
\begin{proof}
For the first statement, note that $\E\lrM{\bl(\btheta_0;\bX_i)} = \bzero$ and $\E\lrM{\bH(\btheta_0;\bX_i) - \bI_0} = \bzero$. Let $\bW = \bl \otimes (\bH - \bI_0) = \lrm{l_1(\bH - \bI_0), \cdots, l_d(\bH - \bI_0)}^\T$.
\begin{align*}
    & \bW^\T\bW = \lrm{\bl \otimes (\bH - \bI_0)}^\T \times \lrm{\bl \otimes (\bH - \bI_0)} = (\bl^\T\bl) \otimes (\bH - \bI_0)^2 \\
    & \lrn{\bl \otimes (\bH - \bI_0)}_2 = \lrn{\bW^\T\bW}_2^{1/2} = \lrn{\bl}_2\lrn{\bH - \bI_0}_2 \\
    & \E(\lrn{\bW}_2^{k/2}) = \E(\lrn{\bl}_2^{k/2}\lrn{\bH - \bI_0}_2^{k/2}) \leq \lrm{\E(\lrn{\bl}_2^{k})\E(\lrn{\bH - \bI_0}_2^{k})}^{1/2} = O(n^{-k/2}) \\
    & \E(\bW) = n^{-2} \sum \E\lrm{\bl_i \otimes (\bH_i - \bI_0)} = n^{-1}\bQ_{12}.
\end{align*}
For the second statement, denote the $j$th element of $\bW_i$ and $\bQ_{12}$ by $\bW_{i, j} = l_j(\bH_i - \bI_0)$ and $\bE_{12, j} = \E(\bW_{1, j})$, and $\bW_{i, j}$ is Hermitian. Let $\bZ_{i, j} = \bW_{i, j} - n^{-1}\bE_{12, j}$, and
$$\E\lrI{\lrn{\bZ_{i, j}}_2^{k/2}} \leq 2^{k/2-1} \E\lrI{\lrn{\bW_{i, j}}_2^{k/2}} + 2^{k/2-1}\lrI{n^{-1}\lrN{\bE_{12, j}}_2}^{k/2} = O(n^{-k/2}).$$ By Lemma \ref{Hlm:EZK}, 
\begin{align*}
        & \E\lrI{\lrN{\sum_i m^{-1}\bZ_{i, j}}_2^{k/2}} = O(m^{-k/4}n^{-k/2}) \\
        & \E\lrI{\lrN{m^{-1}\sum_i \bW_i - n^{-1}\bQ_{12}}_2^{k/2}} \\
    = & \E\lrM{\lrN{\lrI{m^{-1}\sum_i\bW_i - n^{-1}\bQ_{12}}^\T\lrI{m^{-1}\sum_i\bW_i - n^{-1}\bQ_{12}}}_2^{k/4}} \\
    = & \E\lrM{\lrN{\sum_{j=1}^d \lrI{m^{-1}\sum_i \bZ_{i, j}}^\T\lrI{m^{-1}\sum_i \bZ_{i, j}}}_2^{k/4}} \\
    \leq & d^{k/4-1} \sum_{j=1}^d \E\lrI{\lrN{ m^{-1}\sum_i \bZ_{i, j}}_2^{k/2}} = O(m^{-k/4}n^{-k/2}).
\end{align*}
\end{proof}
\end{document}